\documentclass[12pt,oneside,reqno,titlepage, hyperfootnotes=false]{amsart}
\usepackage{lmodern}
\usepackage[T1]{fontenc}
\usepackage[latin9]{inputenc}
\usepackage{color}
\usepackage{bm}
\usepackage{amstext}
\usepackage{amsthm}
\usepackage{amssymb}
\usepackage[a4paper]{geometry}
\usepackage{setspace}
\usepackage[authoryear]{natbib}
\usepackage[pdfusetitle,
 bookmarks=true,bookmarksnumbered=false,bookmarksopen=false,
 breaklinks=false,pdfborder={0 0 0},pdfborderstyle={},backref=false,colorlinks=true]
 {hyperref}

\makeatletter
\numberwithin{equation}{section}
\numberwithin{figure}{section}

\usepackage{amsfonts}
\usepackage{amstext}
\usepackage{amsthm}

\usepackage{hyperref}
\hypersetup{
    colorlinks = true,
    citecolor = black
}

\usepackage{threeparttable}
\usepackage{graphicx}
\usepackage{enumerate}
\usepackage{listings}
\usepackage{subcaption}

\DeclareMathOperator*{\argmin}{arg\,min}

\DeclareMathOperator*{\argmax}{arg\,max}

\usepackage{float}

\newtheorem{prop}{Proposition}
\newtheorem{thm}{Theorem}

\newtheorem{lem}{Lemma}

\newtheorem*{lem*}{Lemma}
\newtheorem*{thm*}{Theorem}

\newtheorem*{asm1}{Assumption 1}
\newtheorem*{asm2}{Assumption 2}
\newtheorem*{asm3}{Assumption 3}
\newtheorem*{asm4}{Assumption 4}
\newtheorem*{asm5}{Assumption 5}
\newtheorem*{asm6}{Assumption 6}

\newtheorem*{asmA1}{Assumption A1}
\newtheorem*{asmA2}{Assumption A2}
\newtheorem*{asmA3}{Assumption A3}
\newtheorem*{asmA4}{Assumption A4}

\theoremstyle{definition}

\makeatother

\begin{document}
\title{You've Got to be Efficient: \\
Ambiguity, Misspecification and Variational Preferences}
\author{Karun Adusumilli$^\dagger$}
\begin{abstract}
We develop a framework for evaluating statistical decisions under
both prior ambiguity and likelihood misspecification, which involves
uniformly expanding a frequentist ambiguity set by a Kullback--Leibler
radius. Optimal decisions are equivalent to minimax decisions under
an exponentially tilted loss: misspecification tilts the loss, while
ambiguity selects the least favorable prior. This separation permits
local asymptotic analysis under general misspecification, allowing
for both local and global priors. For estimation and treatment-assignment,
optimal decisions coincide with those under correct specification
for any degree of misspecification, including in semiparametric models.
This favors maximum-likelihood over simulated method of moments, and
efficient GMM over diagonally weighted alternatives.
\end{abstract}

\thanks{\textit{This version}: \today{}\\
\thispagestyle{empty}I would like to thank Xu Cheng, Frank Diebold,
Wayne Gao, George Mailath, and seminar participants at Chicago Booth,
Paris Econometrics Seminar, Toulouse School of Economics, and the
University of Pennsylvania for valuable discussions and comments that
substantially improved this article.\\
$^\dagger$Department of Economics, University of Pennsylvania}
\maketitle

\section{Introduction \label{sec:Introduction}}

\citet{box1976science} famously observed that all models are wrong,
since they are necessarily approximations of reality. Any researcher
or decision-maker who relies on a statistical model to learn about
a parameter of interest must therefore contend with the possibility
that the likelihood is misspecified. At the same time, researchers
are often unable or unwilling to commit to a single prior over the
parameter. In practice, then, decision-makers confront both prior
ambiguity and likelihood misspecification. 

This article introduces a framework for evaluating statistical decisions
under both sources of concern. Following Bayesian practice, we define
a statistical model as a joint distribution comprising a prior and
a likelihood. We argue that both components are necessary because
they capture fundamentally different types of uncertainty. The prior
encodes epistemic uncertainty, i.e., subjective uncertainty arising
from incomplete knowledge about the parameter of interest, while the
likelihood captures aleatoric uncertainty, i.e., the objective randomness
inherent in any statistical experiment.

To account for prior ambiguity, we define an ambiguity set: a frequentist
model that pairs a possibly misspecified likelihood with every possible
prior. Following \citet{cerreia2026making}, we then uniformly expand
this set by a Kullback--Leibler radius to accommodate likelihood
misspecification. The optimal decision is defined as the one that
achieves the lowest expected loss under the worst-case model from
this expanded set. 

We show that optimal decisions under this formulation are equivalent
to minimax decisions with an exponentially tilted loss function. Likelihood
misspecification manifests as an exponential tilting of the loss,
while prior ambiguity corresponds to a search for the least favorable
prior. Our framework thus enables a clean separation between ambiguity
and misspecification. Furthermore, when there is no fear of misspecification,
the optimal decisions reduce to the standard \citet{wald1950statistical}
formulation of minimax decisions under ambiguity alone, which is the
formulation underlying most of frequentist analysis.

This separation between ambiguity and misspecification also enables
us to develop a local asymptotic theory under general, and, in some
sense, global, misspecification (see Section \ref{sec:Local-Asymptotics-with}).
Under mild conditions, the finite-sample likelihoods, which may themselves
be misspecified, can be replaced by a limit experiment involving the
Gaussian family as the reference likelihood. Our framework may be
contrasted with the substantial literature on local asymptotics under
local misspecification (see, among others, \citealp{andrews2020informativeness};
\citealp{bonhomme2022minimizing}; and \citealp{muller2024locally}).
These approaches embed the parametric likelihood in a larger parametric
class indexed by additional parameters and study local deviations
of these parameters from zero. In large samples, the enlarged model
still admits a Gaussian limit experiment, and misspecification manifests
as a bounded bias in that limit. This follows because a bound on the
KL divergence within the local parametric model translates, via the
Gaussian approximation, into a bound on the bias; the KL divergence
between the local likelihoods transfers to their Gaussian counterparts.
Our approach, by contrast, permits the Gaussian family itself to be
misspecified in the limit experiment, thereby accommodating much richer
classes of misspecification.\footnote{Put differently, the types of misspecification we allow cannot be
accounted for by embedding the likelihood in a larger parametric model,
see Section \ref{subsec:Comparison-with-local} for more details.
The KL deviations we consider are nonetheless of the same order of
magnitude, $1/n$ per observation, where $n$ is the sample size.
The approaches differ in the form of misspecification they allow,
rather than the magnitude of KL divergence. } In fact, as we demonstrate in Section \ref{subsec:Structure_of_lf_lik},
misspecification in our setting can also be rationalized as a random,
parameter-specific location-scale contamination of the score statistic,
connecting the framework to the classical literature on contamination
models initiated by \citet{huber1964robust}. Local misspecification,
by comparison, contaminates the score only by a fixed drift, a considerably
narrower class of deviations.

We also develop new approaches to local asymptotic theory that accommodates
both local and global ambiguity over the unknown parameter. This substantially
generalizes the standard local asymptotic minimax theory (e.g., \citealp[Chapter 8]{van2000asymptotic}),
which considers ambiguity only against a finite set of parameter values.

Local asymptotics also simplifies the search for optimal decisions,
as these are considerably easier to characterize in the limit experiment.
Quite remarkably, we find that for estimation and treatment assignment
problems, optimal decisions coincide with those under correct specification,
regardless of the degree of misspecification. For these problems,
it is therefore always optimal for the decision-maker to proceed as
if the likelihood were correctly specified and select the resulting
optimal decision rule. Intuitively, these results arise because misspecification
under our formulation is symmetric around the reference Gaussian likelihood.
Since the estimation and treatment-assignment loss functions are also
symmetric around the parameter of interest, any estimator that is
not efficient under correct specification would break this symmetry.
Because nature chooses the least favorable likelihood specification
given the decision-maker's choice of estimator, departures from symmetry
necessarily incur higher decision risk. 

These results contrast with those of \citet{bonhomme2022minimizing},
who find that efficient estimators are generally not minimax optimal
under local misspecification. As we explain in Section \ref{subsec:Characterization-of-optimal-decisions},
the difference arises because misspecification in their setting is
structured; it operates through a larger, correctly specified model
under which the data are informative about the direction of misspecification,
and which, in our terminology, amounts to restricted ambiguity over
a larger parameter space. We consider unstructured misspecification,
where there is no direction for the decision-maker to usefully hedge.

We extend our local asymptotic theory to semi-parametric models and
show that the results on optimal decisions carry over to that setting.
The analysis requires care, and is materially different from the parametric
setting, as misspecification in this context arises when the population
distribution $P$ does not coincide with the outcome distribution
$\hat{P}$ in the experimental sample.

Finally, we develop a framework for model aggregation in which the
decision-maker contemplates multiple candidate likelihoods, each subject
to a different degree of misspecification concern.

Our findings have a number of practical consequences. In applications,
researchers often employ inefficient estimators over efficient ones,
a practice sometimes justified on the grounds that under misspecification,
no estimand recovers the precise parameter of economic interest (\citealp{andrews2025purpose}).
Our results, however, suggest that this reasoning might be incomplete.
While the parameter of interest cannot be recovered with certainty
under misspecification, we find that efficient estimators under correct
specification do deliver the lowest decision risk under arbitrary
and undirected misspecification. For parametric models, these results
suggest that practitioners should prefer maximum likelihood over the
simulated method of moments, irrespective of the degree of misspecification.
Similarly, in the context of GMM, they imply that practitioners should
prefer efficient estimation methods, such as two-step GMM, over diagonally
weighted or inefficient alternatives. Unstructured misspecification
concerns alone cannot justify the use of the inefficient estimators
over parametrically or semi-parametrically efficient alternatives. 

\subsection{Related literature}

This article relates to an extensive literature on ambiguity and misspecification
spanning economics, statistics, and computer science. A detailed comparison
with alternative decision-theoretic frameworks is deferred to Section
\ref{subsec:Alternative-approaches}, after we describe our approach.
Here, we restrict ourselves to a broad survey of the literature on
ambiguity and misspecification. 

The analysis of optimal decisions under prior ambiguity originates
with \citet{wald1950statistical}. A substantial body of work in statistics
has extended this framework to local asymptotics; we refer to \citet{ibragimov1981problem,le1986asymptotic,van1996weak,van2000asymptotic}
for textbook treatments. A central result from this literature is
that semi-parametrically efficient estimators are asymptotically minimax
optimal under prior ambiguity.

The literature on misspecification is equally extensive. \citet{huber1964robust}
proposes a contamination model to address likelihood misspecification.
\citet{hansen2011robustness} develop an approach that involves selecting
a worst-case likelihood from an ambiguity set defined by surrounding
a reference or approximate likelihood with a Kullback--Leibler divergence
ball of finite radius. The related field of Distributionally Robust
Optimization (DRO) takes the reference distribution to be the empirical
distribution $\hat{\mathbb{P}}$ of the data in the experiment, and
employs more general measures of distance from $\hat{\mathbb{P}}$
to define ambiguity sets, including $\phi$-divergence measures (e.g.,
reverse KL divergence and total-variation distance), Wasserstein distances
and Levy-Prokhorov distances. We refer to \citet{Ben-Tal:RobustOptimization}
for a textbook treatment of DRO, and to \citet{kuhn2025distributionally}
for a recent survey. These methods do not account for prior ambiguity
and, consequently, do not reduce to the standard minimax formulation
that underpins frequentist analysis in the absence of misspecification
concerns. Our approach instead follows the recent work of \citet{cerreia2026making}
by first defining an ambiguity set to address prior ambiguity and
then uniformly expanding this entire set by a KL divergence radius
to accommodate likelihood misspecification. In contrast to the results
from DRO, we find that optimal estimation and treatment assignment
are invariant to the degree of misspecification.

This article adopts a decision-theoretic approach to ambiguity and
misspecification, which is grounded in the literature on variational
preferences in economic theory, as expounded in \citet{maccheroni2006ambiguity}
and \citet{cerreia2026making}. In concurrent and independent work,
\citet{andrews2026misspecification} develop a closely related axiomatic
framework for characterizing optimal decisions under misspecification.
The main conceptual difference is that we take the payoff-relevant
state of the world to comprise both the parameter and the data, which
naturally leads us to incorporate Bayesian priors as part of the model.
The axiomatic characterization is therefore very different, even as
it yields the same characterization of decision risk; see Section
\ref{sec:Axiomatics} for details. This article differs from \citet{andrews2026misspecification}
in several additional respects. First, \citet{andrews2026misspecification}
introduce and axiomatize \textit{constrained multiplier preferences},
which impose constraints on the direction of misspecification. As
shown in Section \ref{subsec:Model-averaging}, this corresponds to
a limiting case of our model aggregation framework in which the decision-maker
has no misspecification concerns about a larger, nesting model.\footnote{We do not, however, provide an axiomatic justification for the model
aggregation approach.} Second, we cover semi-parametric estimation as well as a broader
class of losses, including treatment-assignment loss. Third, we develop
new approaches to local asymptotic theory that are substantially more
powerful than the local asymptotic minimax theory employed in \citet{andrews2026misspecification}:
for instance, our results can address global minimax risk in addition
to local minimax risk. 

The econometrics literature has also studied alternative, non-decision-theoretic
approaches to misspecification, and we refer to \citet{armstrong2025misspecification}
for a recent survey. For instance, \citet{white1982maximum} defines
pseudo-parameters as the probability limits of estimators, and views
them as suitably defined approximations to the underlying parameter
of interest. The partial identification approach of \citet{manski2003partial}
proposes set-identifying the parameter under misspecification, while
\citet{masten2021salvaging} develop methods for sensitivity analysis.
Perhaps a limitation of these approaches relative to the decision-theoretic
framework is that they do not directly identify the optimal decision
that a decision-maker should employ. 

\section{Decision-Making Under Ambiguity and Misspecification\label{sec:Diffusion-asymptotics-and}}

\subsection{An illustrative example\label{subsec:An-illustrative-example}}

To illustrate our formalism, we introduce the following running example.
A decision-maker, Alice, is tasked with determining whether a drug
should be approved for use in the US population. She is therefore
interested in learning about the parameter $\theta\in\Theta$, defined
as the average population treatment effect. We assume binary outcomes,
so that the population outcome distribution is $\text{Bernoulli}(\theta)$.

To assess the drug's efficacy, the pharmaceutical company has conducted
a randomized controlled trial with $n$ observations. Given the observed
data $\bm{x}\in\mathcal{X}$ from the trial, Alice seeks to choose
a decision $\delta:\mathcal{X}\to\mathcal{A}$ so as to maximize her
utility $u_{n}(\theta,\delta)$, or equivalently, minimize the loss
function $l_{n}(\theta,\delta)=-u_{n}(\theta,\delta)$. 

Examples of loss functions considered in this article include the
class of estimation losses,
\[
l_{n}(\theta,\delta)=\ell\bigl(\sqrt{n}(\theta-\delta)\bigr);\quad\delta\in\Theta,
\]
for some bowl-shaped function $\ell(\cdot)$---e.g., $\ell(z)=\left\Vert z\right\Vert ^{2}$
for mean squared error---where $\ell(\cdot)$ is said to be bowl-shaped
if the sub-level sets $\{z:\ell(z)\le c\}$ are convex and symmetric
about the origin. A second example is the treatment-assignment loss,
\[
l_{n}(\theta,\delta)=\sqrt{n}\bigl(\theta\,\mathbf{1}\{\theta\geq0\}-\theta\,\delta\bigr);\quad\delta\in\{0,1\}.
\]
Under the estimation loss, the goal is to learn directly about the
parameter $\theta$, whereas under the treatment-assignment loss,
the goal is to either approve ($\delta=1$) or reject ($\delta=0$)
the drug for use in the entire population. The loss functions are
indexed with $n$ with a view toward the asymptotic analysis developed
in Section \ref{sec:Local-Asymptotics-with}, where we also discuss
the resulting implications for the normalization of the misspecification
aversion parameter.

Unfortunately for Alice, the trial was conducted exclusively in the
state of Pennsylvania. Because the drug is novel, she has no formal
basis for judging whether, or to what degree, treatment effects observed
in Pennsylvania are representative of those in the broader US population.
This gives rise to likelihood misspecification concerns. At the same
time, Alice also faces ambiguity concerns, as she is unable to form
an initial prior over $\theta$. We now describe a formalism that
accommodates both.

\subsection{Bayesian, Frequentist and Misspecified models\label{subsec:Bayesian,-Frequentist-and-misspecified}}

\subsubsection{(Bayesian) Statistical models}

We begin by formally defining the notion of a statistical model in
the absence of ambiguity or misspecification concerns.

Since the loss function takes the form $l_{n}(\theta,\delta(\bm{x}))$,
the payoff-relevant state of the world is given by $\omega=(\theta,\bm{x})$:
an oracle who knows $\omega$ would recover Alice's loss with certainty.
Following the framework of Savage or Anscombe--Aumann (\citealp{anscombe1963definition}),
we define a model $m\equiv m(\theta,\bm{x})$ as a probability distribution
over the payoff-relevant state $\omega=(\theta,\bm{x})$. This distribution
admits a natural decomposition into a prior and a likelihood:
\begin{equation}
m(\theta,\bm{x})=\pi(\theta)\otimes p_{\theta}(\bm{x}).\label{eq:Bayesian_model}
\end{equation}

Here, $\pi(\theta)$ denotes the posited prior, the marginal distribution
over $\theta$, while $p_{\theta}(\bm{x})=p(\bm{x}|\theta)$ denotes
the posited likelihood, the conditional distribution of $\bm{x}$
given $\theta$. Equation (\ref{eq:Bayesian_model}) is nothing more
than the definition of a Bayesian statistical model; see \citet[Definition 1.2.1]{robert2007bayesian}.

The decomposition of a model into a prior and a likelihood is a canonical
feature of Bayesian decision-making. Given the importance of this
decomposition for what follows, it is worth understanding why both
components are necessary. As we argue below, they capture fundamentally
different sources of uncertainty: epistemic and aleatoric. In the
terminology of \citet{anscombe1963definition}, these correspond to
the uncertainties involved in horse gambles and roulette wheels.

Epistemic uncertainty refers to uncertainty arising from a lack of
knowledge --- uncertainty that can, in principle, be reduced through
the acquisition of additional data or evidence.\footnote{In the Anscombe--Aumann framework, this is the uncertainty of a horse
gamble.} Because the parameter $\theta$ enters Alice's loss function directly,
it is natural to regard it as a quantity that exists in principle
but that Alice does not know. The prior $\pi(\theta)$ thus encodes
Alice's epistemic uncertainty due to her imperfect knowledge of $\theta$.
Crucially, Alice can conceptualize $\theta$ independently of the
likelihood. She may, for instance, have access to prior information,
such as data from related studies, that enables her to form a prior
$\pi$ without reference to whatever experiment the pharmaceutical
company may have conducted.

Aleatoric uncertainty, by contrast, refers to inherent randomness,
which is implicit in the design of any statistical experiment. The
likelihood $p_{\theta}(\bm{x})$ captures precisely this source of
uncertainty. In conducting the trial, the pharmaceutical company presumably
drew a random sample of $n$ observations from the population of Pennsylvania.
This sampling procedure requires the use of an implicit or explicit
random number generator and therefore introduces genuine randomness.\footnote{This corresponds to the uncertainty generated by roulette wheels in
the Anscombe--Aumann framework.} The likelihood thus describes the distribution of the data $\bm{x}$
induced by this randomness, for any given value of $\theta$.

Importantly, in our framework, the likelihood does not rise to the
status of a model. It provides only a mapping from the parameter $\theta$
to the distribution of $\bm{x}$. Because $\theta$ enters Alice's
loss function, knowledge of the correct likelihood would not enable
Alice to obtain a probabilistic forecast of her loss, as she would
still face epistemic uncertainty over $\theta$. 

\subsubsection{Models with prior ambiguity, aka Frequentist models\label{subsec:Models-with-prior-ambiguity}}

We now incorporate prior ambiguity into our framework.\footnote{\citet{cerreia2026making} define ambiguity as the inability to choose
a specific model within a structured class of models. Here, we have
appropriated this terminology to solely refer to prior ambiguity.} Suppose that Alice is unable to form a single prior, perhaps because
she is ambiguity-averse in the sense of \citet{maccheroni2006ambiguity}.
Instead, she posits a structured set of models,
\begin{equation}
\mathcal{Q}:=\bigl\{\pi(\theta)\otimes p_{\theta}(\bm{x}):\pi\in\Delta(\Theta)\bigr\},\label{eq:ambiguity-set}
\end{equation}
where $\Delta(\Theta)$ denotes the set of all probability distributions
over $\theta$, while continuing to treat the likelihood $p_{\theta}(\bm{x})$
as correctly specified. In the spirit of \citet{wald1950statistical},
Alice could then choose the decision rule that performs best against
the worst-case model in $\mathcal{Q}$ --- effectively, the one associated
with the least favorable prior --- thereby guarding against prior
ambiguity:
\[
\delta_{n,f}^{*}:=\argmin_{\delta}\left[\sup_{m\in\mathcal{Q}}\mathbb{E}_{m}\bigl[l_{n}(\theta,\delta)\bigr]\right].
\]

Because the Wald approach underpins much of frequentist analysis,
we refer to $\mathcal{Q}$ as a frequentist model and to $\delta_{n,f}^{*}$
as a frequentist (or minimax) decision rule. 

\subsubsection{Models with prior ambiguity and likelihood misspecification\label{subsec:Ambiguity_and_miss}}

Now suppose that Alice entertains the possibility that the likelihood
$p_{\theta}(x)$ employed in her frequentist model $\mathcal{Q}$
may not be correctly specified. Likelihood misspecification can arise
in two distinct ways. The first is misspecification of functional
form, e.g., specifying a Gaussian likelihood when the true data-generating
process follows a $t$-distribution. In our running example, however,
the outcomes are Bernoulli, so the likelihood is necessarily a product
of Bernoulli distributions and functional-form misspecification is
not a concern. In fact, as we show in Section \ref{subsec:Nuisance-and-structural-parameters},
functional-form misspecification can always be ameliorated by making
the likelihood sufficiently flexible. The misspecification that Alice
confronts here is of the second kind: the link between the parameter
of interest $\theta$ and the distribution over $x$ may be incorrectly
specified. In the running example, the true outcome distribution is
$Y\sim\text{Bernoulli}(\theta_{P})$, where $\theta_{P}$ is the average
treatment effect in Pennsylvania. Likelihood misspecification arises
because, in general, $\theta\neq\theta_{P}$.

To accommodate misspecification concerns, we follow \citet{cerreia2026making}
and place a protective belt around the frequentist model $\mathcal{Q}$.
Formally, we posit a set of unstructured models of the form
\begin{equation}
\mathcal{M}=\Bigl\{ m\in\Delta(\Theta,\mathcal{X}):\min_{q\in\mathcal{Q}}\,R_{q}(m)\leq K\Bigr\},\label{eq:contraint_class}
\end{equation}
where $R_{q}(m)=\text{KL}(m\,\|\,q)$ denotes the Kullback--Leibler
divergence and $K$ is a constant reflecting Alice's degree of ambiguity
aversion. Alice could then choose the decision rule that performs
best against the worst-case model in $\mathcal{M}$, thereby guarding
against both prior ambiguity and misspecification:
\begin{equation}
\delta_{n}^{*}:=\argmin_{\delta}\left[\sup_{m\in\mathcal{M}}\mathbb{E}_{m}\bigl[l_{n}(\theta,\delta)\bigr]\right].\label{eq:optimal_decision}
\end{equation}

For the remainder of this article, we decompose any generic model
$m$ as $m(\theta,\bm{x})=\pi(\theta)\otimes m_{\theta}(\bm{x})$,
and reserve the notation $p_{\theta}(\bm{x})$ for a reference likelihood
specification, which may itself be misspecified. Define $R_{\mathcal{Q}}(m):=\min_{q\in\mathcal{Q}}R_{q}(m)$.
Straightforward algebra yields
\begin{equation}
R_{\mathcal{Q}}(m)=\int\text{KL}\bigl(m_{\theta}(\cdot)\,\|\,p_{\theta}(\cdot)\bigr)\,d\pi(\theta),\label{eq:characterization_of_R_Q}
\end{equation}
so that the set of unstructured models can be equivalently written
as
\[
\mathcal{M}=\left\{ \pi(\theta)\otimes m_{\theta}(x)\;:\;\int\text{KL}\bigl(m_{\theta}(\cdot)\,\|\,p_{\theta}(\cdot)\bigr)\,d\pi(\theta)\leq K,\;\;\pi\in\Delta(\Theta)\right\} .
\]
The class of unstructured models thus comprises every possible prior
$\pi$, paired with all likelihoods $m_{\theta}(\cdot)$ satisfying
the integrated KL constraint. 

The set $\mathcal{M}$ thereby expands $\mathcal{Q}$ by adding a
protective radius of size $K$ in units of the integrated KL divergence
$\int\text{KL}\bigl(m_{\theta}(\cdot)\,\|\,p_{\theta}(\cdot)\bigr)\,d\pi(\theta)$.
Because $\mathcal{Q}$ already accommodates unrestricted prior ambiguity,
this protective belt serves entirely to guard against likelihood misspecification.
Within the class of alternative models, however, the prior and likelihood
may interact in subtle ways: the constraint permits $m_{\theta}(\cdot)$
to deviate substantially from $p_{\theta}(\cdot)$ for certain values
of $\theta$, provided the associated prior $\pi$ places low weight
on those values.

To further understand this interaction between the prior and likelihood,
it is instructive to see how likelihood misspecification would be
modeled in the absence of prior ambiguity. If Alice were able to commit
to a single prior $\pi$, then $\mathcal{Q}$ would be a singleton
and $\mathcal{M}$ would consist of all models $\pi(\theta)\otimes m_{\theta}(\bm{x})$
such that $m_{\theta}(\cdot)\in\mathcal{M}_{\bm{x}|\theta}(\pi)$,
where
\[
\mathcal{M}_{\bm{x}|\theta}(\pi):=\left\{ m_{\theta}(\cdot):\int\text{KL}\bigl(m_{\theta}(\cdot)\,\|\,p_{\theta}(\cdot)\bigr)\,d\pi(\theta)\leq K\right\} .
\]
When the prior is fixed, the class of candidate likelihoods thus depends
on $\pi$; the prior shapes which deviations from the reference likelihood
are admissible. In the general case with unrestricted priors, $\mathcal{M}$
is the union of $\pi(\theta)\otimes\mathcal{M}_{\bm{x}|\theta}(\pi)$
over all $\pi\in\Delta(\Theta)$.

The misspecification framework we have described is certainly asymmetric
in its treatment of priors and likelihoods. To some extent, this reflects
the different natures of these objects as the prior is a probability
distribution over the parameter, whereas the likelihood is a Markov
kernel, which a considerably richer and more complex object. The primary
motivation for leaving the prior unrestricted, however, is to ensure
that the framework reduces to standard frequentist (i.e., minimax)
analysis in the absence of misspecification concerns. One could alternatively
fix a prior and take $K=\infty$, but this would render the problem
trivial, as it results in infinite decision risk.

\subsection{Nuisance and structural parameters\label{subsec:Nuisance-and-structural-parameters}}

A nuisance parameter is an unknown quantity that enters the likelihood
but does not affect the loss function. In our running example, if
the outcomes are distributed as $\mathcal{N}(\mu,\sigma^{2})$ but
Alice is solely interested in learning about the mean treatment effect
$\mu$, then $\sigma^{2}$ is a nuisance parameter. Nuisance parameters
make the likelihood more flexible and can be used to ameliorate functional-form
misspecification; indeed, one can allow for nonparametric specifications
by making the nuisance parameters infinite-dimensional. As noted earlier,
however, nuisance parameters cannot address the second type of misspecification
concerning the link between the parameter of interest and the data.
Even if Alice were to adopt a fully nonparametric specification of
the outcome distribution, she would still face misspecification concerns,
as treatment effects in Pennsylvania may differ fundamentally from
those in the broader US population.

In contrast, we refer to the parameters that enter the utility function
directly as structural parameters. In what follows, $\theta$ denotes
the full collection of unknown parameters, which may include both
structural and nuisance components. The structural parameters are
modeled as known functions $\mu(\theta)$ of $\theta$. With this
notation, the loss functions introduced earlier take the more general
form
\begin{align*}
l_{n}(\theta,\delta) & =\ell\bigl(\sqrt{n}\,(\mu(\theta)-\delta)\bigr);\quad\delta\in\mu(\Theta),\ \textrm{and}\\
l_{n}(\theta,\delta) & =\sqrt{n}\bigl(\mu(\theta)\,\mathbf{1}\{\mu(\theta)\geq0\}-\mu(\theta)\,\delta\bigr);\quad\delta\in\{0,1\},\mu(\theta)\in\mathbb{R},
\end{align*}
for estimation and treatment-assignment losses, respectively. The
definitions of Bayesian, frequentist, and misspecified models remain
unchanged as the introduction of nuisance parameters only affects
the form of the loss functions.

\subsection{Alternative approaches to misspecification: A comparison\label{subsec:Alternative-approaches}}

\citet{andrews2025purpose} define an econometric model $(\theta,p_{\theta}(\cdot))$
as a combination of the parameter and the likelihood.\footnote{In the terminology of \citet{andrews2025purpose}, the likelihood
is referred to as a data-generating process.} Apart from the prior over $\theta$, this coincides with the definition
of a Bayesian statistical model. Introducing the prior allows us to
account for prior ambiguity, which, as we have seen, plays a key role
even in the standard frequentist approach.

An important and growing literature has considered accounting for
misspecification through partial identification; see, e.g., \citet{ishihara2021evidence,yata2021optimal,christensen2022optimal,montiel2026decision}.
This literature supposes that the parameter of interest $\theta$
lies within a bounded distance, $d(\theta,\theta_{P})\leq L$, of
an identifiable parameter $\theta_{P}$. In our running example, this
would require Alice to assume that the population treatment effect
$\theta$ differs from the treatment effect in Pennsylvania $\theta_{P}$
by at most $L$. Bounds of this form arise naturally in a number of
applications. However, the approach can fall short as a general framework
for misspecification for a couple of reasons. First, bounding the
parameters directly lacks an axiomatic justification. Second, imposing
a uniform bound $d(\theta,\theta_{P})\leq L$ for all $\theta_{P}$
implies comparisons across different values of $\theta_{P}$ that
may be at odds with the decision-maker's actual preferences over ambiguity
and misspecification. To see why, note that there is epistemic uncertainty
over the value of $\theta_{P}$. There is no a priori reason to believe
that the bound $d(\theta,\theta_{P})\leq L$ provides the same degree
of protection against misspecification when $\theta_{P}=0.9$ as when
$\theta_{P}=0.1$. Depending on Alice's preferences, i.e., her loss
function, she may be less concerned about misspecification at high
values of $\theta_{P}$ (which suggests the treatment is highly effective)
than at low values. 

In more closely related work, \citet{christensen2023counterfactual}
measure misspecification through a statistical divergence between
the true and reference distributions of the latent variables in a
structural model, but their interest is in sensitivity analysis and
partial identification of counterfactuals, rather than in characterizing
optimal decisions. \citet{andrews2020informativeness}, \citet{bonhomme2022minimizing},
and \citet{muller2024locally} study estimation under local misspecification.
In these settings, the likelihood is embedded in a larger parametric
class or in local (parametric) sub-models, and one considers local
deviations of the additional parameters from zero. As $n\to\infty$,
the reference likelihood continues to deliver a Gaussian limit experiment,
and misspecification manifests as a deterministic bias in that limit,
with a bound on the statistical divergence translating into a bound
on the bias. The misspecification we allow, by contrast, is much richer
in form, since it permits the Gaussian limit experiment itself to
be misspecified. As shown in Section \ref{subsec:Structure_of_lf_lik},
it corresponds to a random, parameter-specific location-scale contamination
of the score statistic, whereas local misspecification corresponds
to a fixed drift.

The setting of \citet{bonhomme2022minimizing} also differs from ours
in a more basic respect. There, the data are generated by a larger
structural model that is assumed to be correctly specified. In our
terminology, this is a form of bounded ambiguity over a larger parameter
space, rather than misspecification of the likelihood. Section \ref{subsec:Characterization-of-optimal-decisions}
explains why the two frameworks lead to different optimal estimators.

Several of these articles, e.g., \citet{bonhomme2022minimizing} and
\citet{christensen2023counterfactual}, impose the same divergence
bound at every parameter value. As with the bound $d(\theta,\theta_{P})\leq L$,
there is no a priori reason to believe that such a bound provides
the same degree of protection against misspecification across different
values of $\theta$. Alice may be less concerned about misspecification
at high values of $\theta$ than at low values. Moreover, KL divergences
are sensitive to the choice of reference measure, so divergences such
as $R_{p_{\theta_{1}}}(m_{\theta_{1}})$ and $R_{p_{\theta_{2}}}(m_{\theta_{2}})$,
computed against different reference likelihoods, need not be comparable.
The current formulation avoids these difficulties by first postulating
an infinite-dimensional ambiguity set $\mathcal{Q}$ and then uniformly
expanding it by a KL radius $K$. For the same value of $K$, the
resulting class of models is also strictly larger:
\[
\Bigl\{\pi\otimes m_{\theta}:\ \sup_{\theta}\text{KL}\bigl(m_{\theta}\,\|\,p_{\theta}\bigr)\leq K,\ \pi\in\Delta(\Theta)\Bigr\}\subset\mathcal{M}.
\]
As in \citet{cerreia2026making}, the value of $K$ can be tied to
the decision-maker's underlying preferences over ambiguity and misspecification.
But rather remarkably, as it turns out, our optimal decisions are
asymptotically independent of the choice of $K$.

\section{Characterizing Optimal Decisions\label{sec:Characterizing-Optimal-Decisions}}

We now characterize optimal decisions under ambiguity and misspecification.
For reasons that will become apparent shortly, it is convenient to
start with utility maximization rather than loss minimization. Following
the framework described in Section \ref{subsec:Ambiguity_and_miss},
the optimal decision takes the form:
\[
\delta_{n}^{*}:=\argmax_{\delta}\left[\inf_{m\in\mathcal{M}}\,\mathbb{E}_{m}\bigl[u_{n}(\theta,\delta)\bigr]\right]=\argmax_{\delta}\,\inf_{m}\,\left\{ \mathbb{E}_{m}\bigl[u_{n}(\theta,\delta)\bigr]:R_{\mathcal{Q}}(m)\le K\right\} .
\]

The function $R_{\mathcal{Q}}(\cdot):\Delta(\Theta\times\mathcal{X})\to\mathbb{R}$
is convex, so we may apply a minimax theorem to show that for each
$K$ at which the constraint binds, there exists a multiplier $\lambda>0$
such that\footnote{Formally, we suppose that $\Theta\times\mathcal{X}$ is compact, which
implies $\Delta(\Theta\times\mathcal{X})$ is compact. Then, an application
of Sion's minimax theorem for each $\delta$ shows that 
\begin{align*}
 & \max_{\delta}\,\inf_{m\in\Delta(\Theta\times\mathcal{X})}\,\left\{ \mathbb{E}_{m}\bigl[u_{n}(\theta,\delta)\bigr]:R_{\mathcal{Q}}(m)\le K\right\} \\
 & =\max_{\delta}\,\inf_{m\in\Delta(\Theta\times\mathcal{X})}\,\max_{\lambda\ge0}\left\{ \mathbb{E}_{m}\bigl[u_{n}(\theta,\delta)\bigr]+\lambda(R_{\mathcal{Q}}(m)-K)\right\} \\
 & =\max_{\delta}\,\max_{\lambda\ge0}\,\inf_{m\in\Delta(\Theta\times\mathcal{X})}\left\{ \mathbb{E}_{m}\bigl[u_{n}(\theta,\delta)\bigr]+\lambda(R_{\mathcal{Q}}(m)-K)\right\} .
\end{align*}
The claim then follows by interchanging the max operations.}
\begin{align}
\delta_{n}^{*} & =\argmax_{\delta}V_{n}(\delta),\textrm{ where}\nonumber \\
V_{n}(\delta) & :=\inf_{m}\left\{ \mathbb{E}_{m}\bigl[u_{n}(\theta,\delta)\bigr]+\lambda\,R_{\mathcal{Q}}(m)\right\} .\label{eq:variational_criterion}
\end{align}
Following \citet{cerreia2026making}, we refer to $V_{n}(\cdot)$
as the variational decision criterion. In fact, the variational criterion
has a direct axiomatic justification through preferences, as in \citet{cerreia2026making},
more so than the constraint version (\ref{eq:optimal_decision}),
which is only used as a motivation for the optimal decision. This
axiomatic justification is discussed in more detail in Section \ref{sec:Axiomatics}.

Recalling the definition $R_{\mathcal{Q}}(m):=\min_{q\in\mathcal{Q}}R_{q}(m)$
and interchanging the order of the $\min_{q\in\mathcal{Q}}$ and $\inf_{m}$
operations, we can write
\[
V_{n}(\delta)=\min_{q\in\mathcal{Q}}\inf_{m}\left\{ \mathbb{E}_{m}\bigl[u_{n}(\theta,\delta)\bigr]+\lambda\,R_{q}(m)\right\} .
\]
The Donsker-Varadhan variational formula yields
\begin{equation}
\inf_{m}\left\{ \mathbb{E}_{m}\bigl[u_{n}(\theta,\delta)\bigr]+\lambda\,R_{q}(m)\right\} =-\lambda\ln\mathbb{E}_{q}\!\left[e^{-u_{n}(\theta,\delta)/\lambda}\right].\label{eq:DV formula}
\end{equation}
Converting back to the loss function via $l_{n}(\theta,\delta)=-u_{n}(\theta,\delta)$,
we obtain
\begin{equation}
V_{n}(\delta)=-\lambda\ln\left\{ \max_{q\in\mathcal{Q}}\,\mathbb{E}_{q}\left[e^{l_{n}(\theta,\delta)/\lambda}\right]\right\} ,\label{eq:decision-risk-eq}
\end{equation}
and consequently, the optimal decision can be characterized as
\begin{equation}
\delta_{n}^{*}=\argmin_{\delta}\,\max_{q\in\mathcal{Q}}\,\mathbb{E}_{q}\!\left[e^{l_{n}(\theta,\delta)/\lambda}\right].\label{eq:optimal_decision_Q}
\end{equation}

So far, the calculations above follow \citet{cerreia2026making}.
However, due to the special structure of $\mathcal{Q}$ in our setting,
comprising all possible priors paired with the reference likelihood
$p(\bm{x}|\theta)$, we can simplify (\ref{eq:optimal_decision_Q})
further:
\begin{equation}
\delta_{n}^{*}=\argmin_{\delta}\,\max_{\pi\in\Delta(\Theta)}\int\mathbb{E}_{p(\bm{x}|\theta)}\!\left[e^{l_{n}(\theta,\delta)/\lambda}\right]d\pi(\theta).\label{eq:optimal_decision_pi}
\end{equation}
In this expression, the quantity $R_{n}(\theta,\delta):=\mathbb{E}_{p(\bm{x}|\theta)}\!\left[e^{l_{n}(\theta,\delta)/\lambda}\right]$
admits a natural interpretation as the frequentist risk of the decision
rule $\delta$ under the reference likelihood $p(\bm{x}|\theta)$,
evaluated with respect to the exponentiated loss $e^{l_{n}(\theta,\delta)/\lambda}$. 

Notice that (\ref{eq:optimal_decision_pi}) corresponds to a standard
minimax decision framework under exponentiated loss: we can interpret
the optimal decision as the result of a two-player game in which nature
chooses the least favorable prior while the decision-maker chooses
the optimal rule. Equation (\ref{eq:optimal_decision_pi}) is therefore
a key result of this article. It establishes that optimal decisions
under both ambiguity and misspecification are equivalent to optimal
decisions under ambiguity alone, but with an exponentiated loss function.
The result also reveals how the two sources of concern separate naturally.
The effect of misspecification is to transform the loss function into
an exponentiated version; intuitively, the decision-maker magnifies
the impact of large losses while attenuating the impact of small ones.
The effect of ambiguity, as in the setting without misspecification,
manifests in the search for the least favorable prior.

\section{Local Asymptotics with General Misspecification\label{sec:Local-Asymptotics-with}}

It is rarely feasible to solve the minimax problem (\ref{eq:optimal_decision_pi})
exactly in finite samples. Instead, as is standard even in classical
frequentist (i.e., minimax) settings, we turn to local asymptotic
approximations.

Following the usual approach, we fix a reference parameter $\theta_{0}$
and consider local perturbations of the form $\theta_{0}+h/\sqrt{n}$.
Priors $\pi(\theta)$ over $\theta$ are then mapped to local priors
$\pi(h)$ over $h$.

In the case of treatment-assignment loss, $l_{n}(\theta,\delta)=n\bigl(\mu(\theta)\,\mathbf{1}\{\mu(\theta)\geq0\}-\mu(\theta)\,\delta\bigr)$,
we formally show that local asymptotics arise naturally from the global
minimax problem (\ref{eq:optimal_decision_pi}). The key observation
is that the least favorable prior concentrates its mass on regions
where the treatment effect $\mu(\theta)$ is of order $1/\sqrt{n}$.
When $\mu(\theta)$ is of a higher order of magnitude than $1/\sqrt{n}$,
determining the optimal assignment is asymptotically trivial. Conversely,
when $\mu(\theta)$ is of a lower order of magnitude than $1/\sqrt{n}$,
the difference between treatment and status quo is negligible, so
the loss is close to zero regardless of the choice of $\delta$. It
is therefore natural to choose a reference parameter $\theta_{0}$
satisfying $\mu(\theta_{0})=0$, since the least favorable prior would
concentrate around this value in any case. As in \citet{hirano2009asymptotics},
these same considerations apply to treatment-assignment problems in
the absence of misspecification as well.

But how should one choose a reference $\theta_{0}$ for estimation
loss, and what is the meaning of local asymptotics in this setting?
After all, there is no notion of a true, or even, pseudo-true parameter
in our framework. We offer two interpretations for rationalizing local
asymptotics.\footnote{We use local asymptotics to formally characterize optimal decisions.
This approach does not require the definition of pseudo-true parameters.
We do not address how to conduct inference on pseudo-true parameters,
which depend on both the optimal decision and the true value of $\theta$.
Indeed, two decisions that are both optimal may nonetheless be associated
with different pseudo-true parameters, further complicating such inference.}

The first is that local asymptotics amounts to localizing prior ambiguity
around a reference parameter $\theta_{0}$ that the decision-maker
believes is close to the true value. In our running example, Alice
may have a priori reason to believe that the true treatment effect
lies near $\theta_{0}$, even if she is uncertain about its exact
value. It would then be natural to restrict her ambiguity set to a
$1/\sqrt{n}$ neighborhood of $\theta_{0}$. Because this prior information
is obtained independently of the data, likelihood misspecification
does not affect the choice of $\theta_{0}$. Consequently, we can
localize the priors, even as we can---and do---allow for general
misspecification of the likelihood.

Under the second interpretation, the choice of the reference $\theta_{0}$
is itself subject to adversarial optimization, i.e., we evaluate performance
against the worst-case value of $\theta$. The basic idea, following
\citet{ibragimov1981problem}, is to decompose the global minimax
problem into two stages: first, fix a reference $\theta_{0}$ and
evaluate the local minimax performance of a decision rule $\delta$
against local alternatives of the form $\theta_{0}+h/\sqrt{n}$; then,
in an outer step, select the least favorable reference $\theta_{0}$. 

Here, we focus on a theoretical development of the first interpretation.
The second interpretation is detailed in Section \ref{subsec:Non-local-priors},
with the theory developed in Appendix \ref{sec:Local-Asymptotics-with-non-local-prior}.

\subsection{Fixed $\lambda$ asymptotics and loss function normalizations\label{subsec:Fixed_lambda}}

For our asymptotic analysis, the dual formulation of $V_{n}(\delta)$
takes precedence and we consider a regime in which $\lambda$ is fixed
even as $n$ increases. While this may appear restrictive at first
glance, the choice is closely connected to the scaling of the loss
functions $l_{n}(\cdot)$ with $n$. For example, in the case of squared
error loss, we set $l_{n}(\theta,\delta)=n(\theta-\delta)^{2}$. We
could alternatively allow the loss to be unscaled $l(\theta,\delta)=(\theta-\delta)^{2}$
and take $\lambda_{n}=\lambda/n$; this yields the same ordering for
the decision risk $V_{n}(\delta)$. Note that either case corresponds
to taking $K$ in (\ref{eq:contraint_class}) to be fixed, implying
that the KL divergences considered are $O(1)$ overall, or equivalently,
of the order $1/n$ per observation.

Our choice of normalizations for $l_{n}(\cdot)$, and by extension,
that of $\lambda$ or $K$, is guided by the principle that the asymptotic
framework should keep the problem non-trivial in the limit. If we
normalize the loss function $l_{n}(\theta,\delta)$ directly, then
$\lambda$ and $K$ must be held fixed. Otherwise, taking $\lambda\to0$
or $K\to\infty$ would imply $\lim_{n\to\infty}V_{n}(\delta)=\infty$
for any $\delta$, rendering the problem asymptotically trivial. We
prefer to scale the loss while keeping $\lambda$ fixed since, following
\citet{cerreia2026making}, $\lambda$ admits a natural interpretation
as a preference parameter measuring the decision-maker's degree of
misspecification aversion. Under this interpretation, it is natural
to treat $\lambda$ (or $K$) as invariant to $n$.\footnote{This is similar to how the degree of risk aversion is typically regarded
as independent of the sample size.}

\subsection{Parametric models: Setup}

We assume that the data consists of an i.i.d collection of outcomes
$\ensuremath{\bm{x}:=\{Y_{i}^ {}\}_{i=1}^{n}}$. Under the reference
likelihood, $Y_{i}$ is distributed as $P_{\theta}$. Let $\nu$ denote
a dominating measure for $\{P_{\theta}:\theta\in\mathbb{R}^{d}\}$,
and set $p_{\theta}^ {}:=dP_{\theta}^ {}/d\nu$. We require the reference
class of likelihoods, $\{P_{\theta}^ {}\}_{\theta}$, to be quadratic
mean differentiable (qmd): 

\begin{asm1} The class $\{P_{\theta}^ {}:\theta\in\mathbb{R}^{d}\}$
is qmd around $\theta_{0}^ {}$, i.e., there exists a score function
$\psi_{}(\cdot)$ such that for each $h^ {}\in\mathbb{R}^{d},$
\[
\int\left[\sqrt{p_{\theta_{0}^ {}+h^ {}}}-\sqrt{p_{\theta_{0}}}-\frac{1}{2}h^{\intercal}\psi_{}\sqrt{p_{\theta_{0}^ {}}}\right]^{2}d\nu=o(\vert h^ {}\vert^{2}).
\]
Furthermore, the information matrix $I_{0}:=\mathbb{E}_{\theta_{0}}[\psi\psi_{}^{\intercal}]$
is invertible. \end{asm1}

In the illustrative example, the outcomes are modeled as Bernoulli,
so Assumption 1 holds with $\psi_{}(y)=\left(\theta_{0}^ {}(1-\theta_{0}^ {})\right)^{-1}(y-\theta_{0}^ {})$.
More broadly, this assumption is rather mild and satisfied for almost
all commonly used likelihood models, including the Normal, Cauchy,
Exponential, and Poisson distributions. It is important to bear in
mind that Assumption 1 constrains the reference class of likelihoods,
$p_{\theta}$, not the actual likelihoods, which are unknown. 

We also assume that the function $\mu(\theta)$, which maps $\theta$
to structural parameters, satisfies a mild differentiability condition:

\begin{asm2} There exists $\dot{\mu}_{0}\in\mathbb{R}^{d}\backslash\{0\}$
and $\epsilon_{n}\to0$ independent of $h$ such that $\sqrt{n}\left(\mu(\theta_{0}+h/\sqrt{n})-\mu(\theta_{0})\right)=\dot{\mu}_{0}^{\intercal}h+\epsilon_{n}\vert h\vert^{2}$
for all bounded $h$. \end{asm2}

Let $P_{n,h}$ denote the joint probability measure over the iid $Y_{1},\dots,Y_{n}$
when each $Y_{i}\sim P_{\theta_{0}+h/\sqrt{n}}$, and let $\mathbb{E}_{n,h}[\cdot]$
denote the corresponding expectation. Under local asymptotics, the
minimal risk attained under the minimax problem (\ref{eq:optimal_decision_pi})
can be written as: 
\begin{equation}
V_{n}^{*}=\min_{\delta}\,\max_{\pi(h)}\int\mathbb{E}_{n,h}\left[e^{l_{n}(\theta_{0}+h/\sqrt{n},\delta)/\lambda}\right]d\pi(h).\label{eq:minimax_value_local_approx}
\end{equation}

\subsection{Limit approximations and the Gaussian limit experiment\label{subsec:Limit-approximations}}

Define the standardized score statistic as
\[
x_{n}=\frac{I_{0}^{-1/2}}{\sqrt{n}}\sum_{i=1}^{n}\psi(Y_{i}).
\]
It is well known, see e.g., \citet[Chapter 7]{van2000asymptotic},
that quadratic mean differentiability (Assumption 1) implies $\mathbb{E}_{n,0}[\psi(Y_{i})]=0$.
Then, by the central limit theorem,
\begin{equation}
x_{n}\xrightarrow[P_{n,0}]{d}x\sim\mathcal{N}(0,I).\label{eq:Convergence of score process}
\end{equation}
Assumption 1 also implies the important property of Local Asymptotic
Normality (LAN; \citealp[Chapter 7]{van2000asymptotic}):
\begin{equation}
\ln\frac{dP_{n,\theta_{0}+h/\sqrt{n}}}{dP_{n,\theta_{0}}}=h^{\intercal}I_{0}^{1/2}x_{n}-\frac{1}{2}h^{\intercal}I_{0}h+o_{P_{n,0}}(1),\ \textrm{uniformly over bounded }h.\label{eq:LAN property}
\end{equation}

Consider now a limit experiment in which the decision-maker observes
a $d$-dimensional signal $x$, posited to be drawn from a reference
Gaussian likelihood, 
\[
P_{h}(\cdot)\equiv\mathcal{N}(I_{0}^{1/2}h,I).
\]
By the properties of the Gaussian distribution,
\[
\ln\frac{dP_{h}}{dP_{0}}=h^{\intercal}I_{0}^{1/2}x_{}-\frac{1}{2}h^{\intercal}I_{0}h.
\]
It follows from (\ref{eq:Convergence of score process}) and (\ref{eq:LAN property})
that the reference likelihood ratios in the finite-sample experiment
converge to their counterparts in the limit experiment:
\[
\ln\frac{dP_{n,\theta_{0}+h/\sqrt{n}}}{dP_{n,\theta_{0}}}\xrightarrow[P_{n,0}]{d}\ln\frac{dP_{h}}{dP_{0}},\ \textrm{for each }h.
\]
Furthermore, Assumption 2 implies that the loss functions admit asymptotic
approximations. For estimation-loss, defining $\tilde{\delta}_{n}=\sqrt{n}\dot{\mu}_{0}^{\intercal}(\delta_{n}-\theta_{0})$
and assuming $\tilde{\delta}_{n}$ has a weak limit $\tilde{\delta}$,
we have
\begin{equation}
l_{n}(\theta_{0}+h/\sqrt{n},\delta_{n})\equiv\ell\left(\sqrt{n}\left(\mu(\theta_{0}+h/\sqrt{n})-\delta_{n}\right)\right)\rightsquigarrow\ell(\dot{\mu}_{0}^{\intercal}h-\tilde{\delta}),\label{eq:estimation_loss_approx}
\end{equation}
where $`$$\rightsquigarrow$' represents weak convergence. For treatment-assignment
loss, since the reference parameter satisfies $\mu(\theta_{0})=0$,
Assumption 2 implies
\begin{equation}
l_{n}(\theta_{0}+h/\sqrt{n},a)\to\dot{\mu}_{0}^{\intercal}h\,\mathbf{1}\{\dot{\mu}_{0}^{\intercal}h\geq0\}-(\dot{\mu}_{0}^{\intercal}h)a,\ \textrm{uniformly over }a\in[0,1]\textrm{ and bounded }h.\label{eq:treatment_loss_approx}
\end{equation}

Convergence of likelihood ratios implies asymptotic equivalence between
the actual and limit experiments in the sense of \citet{le1986asymptotic}.
Combined with the loss function approximations above, this suggests
that the minimax value $V_{n}^{*}$ in (\ref{eq:minimax_value_local_approx})
should converge to the minimax value $V^{*}$ in the limit experiment,
where 
\begin{align}
V^{*} & :=\min_{\tilde{\delta}}\,\max_{\pi(h)}\int\mathbb{E}_{h}\left[e^{l(h,\tilde{\delta})/\lambda}\right]d\pi(h),\ \textrm{with}\label{eq:minimax_value_limit_experiment}\\
l(h,\tilde{\delta}) & =\begin{cases}
\ell(\dot{\mu}_{0}^{\intercal}h-\tilde{\delta}) & \textrm{for estimation loss},\\
\dot{\mu}_{0}^{\intercal}h\,\left\{ \mathbf{1}\{\dot{\mu}_{0}^{\intercal}h\geq0\}-\tilde{\delta}\right\}  & \textrm{for treatment-assignment loss}.
\end{cases}\nonumber 
\end{align}
Formal statements to this effect are provided in Section \ref{subsec:Formal-results-parametric}.

Equation (\ref{eq:minimax_value_limit_experiment}) suggests that
asymptotically optimal decision rules can be derived by solving the
minimax problem in the limit experiment and mapping the solutions
back to the finite-sample setting. Since optimal decisions are considerably
easier to characterize under Gaussian likelihoods, this reduction
illustrates the key benefit of the local asymptotic approach.

\subsubsection{Comparison with local misspecification\label{subsec:Comparison-with-local}}

It is instructive to contrast our approach with the more traditional
analysis of local misspecification. That literature embeds the likelihood
in a larger parametric family indexed by an additional parameter $\eta$,
with $\eta=0$ recovering the original model, and studies local deviations
of the form $\eta=\bar{\eta}/\sqrt{n}$. Because the enlarged likelihood
continues to admit a Gaussian limit experiment, misspecification manifests
as an asymptotic bias that is linear in $\bar{\eta}$ (see \citealp{andrews2020informativeness};
\citealp{bonhomme2022minimizing}; and \citealp{muller2024locally}).
Furthermore, in this setting, a KL constraint on the alternatives
is asymptotically equivalent to a bound on the quadratic form $\bar{\eta}^{\intercal}I_{\eta}\bar{\eta}$---where
$I_{\eta}$ denotes the information matrix for $\eta$---since KL
divergences between local parametric alternatives carry over to their
Gaussian limits under some regularity conditions. This bound, in turn,
translates directly into a bound on the bias. 

Our framework, by contrast, allows for alternatives that cannot be
generated by any such embedding, so that the Gaussian limit experiment
may itself be misspecified. This is possible because our theory relies
on the Gaussian approximation only for the likelihood ratios of the
reference likelihood and makes no claim about the convergence of the
true likelihood ratios. Hence, misspecification manifests as an exponential
tilting of the loss function. Note that the two approaches are alike
in scale, in that the KL divergences we allow are also of order $1/n$
per observation. What distinguishes the two frameworks is the form
of misspecification permitted, rather than the magnitude of the deviations.

\subsection{Characterization of optimal decisions in the limit experiment\label{subsec:Characterization-of-optimal-decisions}}

We begin with estimation-loss. Since $\ell(\cdot)$ is bowl-shaped,
so is $e^{\ell(\cdot)/\lambda}$. It then follows from Anderson's
lemma, see e.g., \citet[Proposition 8.6]{van2000asymptotic}, that
the minimax-optimal estimator in the limit experiment---the solution
to (\ref{eq:minimax_value_limit_experiment})---is simply
\[
\tilde{\delta}^{*}=\dot{\mu}_{0}^{\intercal}I_{0}^{-1/2}x.
\]
Remarkably, $\tilde{\delta}^{*}$ is independent of $\lambda$, which
governs the degree of misspecification. In fact, $\tilde{\delta}^{*}$
coincides with most efficient estimator, the maximum likelihood estimator,
under correct specification, which corresponds to $\lambda=\infty$.
In other words, the optimal estimator under ambiguity and misspecification
is identical to the optimal estimator under ambiguity alone.

For treatment-assignment loss, Anderson's lemma does not apply. Nevertheless,
as the following proposition shows, the optimal decision again takes
a simple form: it recommends treatment whenever the MLE of the treatment
effect under correct specification is positive. 

\begin{prop}\label{prop1}The minimax-optimal decision rule in the
limit experiment under the treatment assignment loss is $\tilde{\delta}^{*}=\mathbf{1}\{\dot{\mu}_{0}^{\intercal}I_{0}^{-1/2}x\geq0\}$.
The corresponding least-favorable prior is a symmetric two-point prior
supported on $(-h^{*},h^{*})$ with $h^{*}:=\frac{\Delta^{*}}{\sigma^{2}}I_{0}^{-1}\dot{\mu}$,
where $\sigma^{2}:=\dot{\mu}_{0}^{\intercal}I_{0}^{-1}\dot{\mu}_{0}$
and 
\[
\Delta^{*}=\arg\max_{\Delta\ge0}\left\{ \left(e^{\frac{\Delta}{\lambda}}-1\right)\Phi(-\Delta/\sigma)\right\} .
\]
 \end{prop}

As with the optimal estimator, the optimal treatment-assignment rule
is independent of the degree of misspecification. Intuitively, these
results arise because misspecification under our formulation is unstructured
and therefore symmetric around the reference Gaussian likelihood.
Since the loss functions are also symmetric around the reference $\theta_{0}$,
any estimator that is not efficient under correct specification would
break this symmetry. Because nature chooses the least favorable likelihood
specification given the decision-maker's choice of estimator, departures
from symmetry would necessarily incur higher decision risk. It is
therefore optimal for the decision-maker to proceed as if the likelihood
were correctly specified and select the resulting optimal decision
rule.

Our results may appear to conflict with those of \citet{bonhomme2022minimizing},
who also work with KL neighborhoods at the $1/n$ rate, yet find that
the efficient estimator under the reference likelihood is generally
not minimax optimal. The two sets of results are, however, complementary.
In \citet{bonhomme2022minimizing}, the data are generated by a larger
structural model $f_{\beta,\pi}$ that is assumed to be correctly
specified, the parameter of interest depends on $\pi$, and misspecification
consists of $\pi$ lying in a shrinking neighborhood of some reference
family $\pi(\gamma)$ indexed by $\gamma$. In our terminology, this
is a form of ambiguity over a larger parameter $(\beta,\pi)$, in
which the local priors over $\pi$ are restricted to a bounded set,
rather than misspecification of the likelihood. It corresponds to
the limiting case of our model aggregation framework in Section \ref{subsec:Model-averaging}
with no misspecification concern about the nesting model. Minimax
estimation over a bounded parameter set naturally leads to shrinkage,
which is the source of the departures from efficiency in their setting.

\subsection{Formal results\label{subsec:Formal-results-parametric}}

We now formally establish the asymptotic equivalence of experiments
through two results. First, we show that the minimax value $V^{*}$
in the limit experiment forms an asymptotic lower bound on the sequence
of optimal decision risks in the finite-sample experiments. Second,
we show that plug-in versions of the optimal limit-experiment decision
rules $\tilde{\delta}^{*}$---obtained by replacing $I_{0}^{-1/2}x$
with the finite sample MLE $\hat{\theta}_{\textrm{mle}}$---are asymptotically
optimal, in the sense that their decision-risks converge to $V^{*}$.
Specifically, we argue that the asymptotically decisions are given
by 
\begin{equation}
\hat{\delta}_{n}^{*}=\begin{cases}
\mu(\hat{\theta}_{\textrm{mle}}) & \textrm{for estimation,}\\
\mathbf{1}\left\{ \mu(\hat{\theta}_{\textrm{mle}})\geq0\right\}  & \textrm{for treatment assignment.}
\end{cases}\label{eq:asymptotically_optimal_decisions}
\end{equation}

As discussed at the beginning of this section, our formal results
require localization of ambiguity. This involves restricting attention
to the set of compactly supported priors $\Delta_{M}(\mathcal{H})\equiv\{\pi(h):\textrm{supp}(\pi)\in[-M,M]\}$
for some $M<\infty$. This restriction is not needed for the lower
bound, but plays a role in establishing that the bound is attained
by the plug-in rules. 

\begin{thm}\label{thm:lower_bound}(Lower bound) Suppose that Assumptions
1 and 2 hold. Then, under both the estimation and treatment-assignment
loss functions, 
\[
\lim_{M\to\infty}\liminf_{n\to\infty}\,\min_{\delta}\,\max_{\pi(h)\in\Delta_{M}(\mathcal{H})}\int\mathbb{E}_{n,h}\left[e^{l_{n}(\theta_{0}+h/\sqrt{n},\delta)/\lambda}\right]d\pi(h)\ge V^{*}.
\]
\end{thm}

In order to show that the decision rules in (\ref{eq:asymptotically_optimal_decisions})
attain the lower bound, we place a mild regularity condition on the
MLE:

\begin{asm3} The maximum-likelihood estimator $\hat{\theta}_{\textrm{mle}}$
admits a locally linear score-function approximation:
\[
\hat{\theta}_{\textrm{mle}}-\theta_{0}=I_{0}^{-1/2}x_{n}+o_{P_{n,0}}(1).
\]
 \end{asm3}

To avoid technical issues relating to the existence of moments for
the estimation problem, our theory also requires that $\ell(\cdot)$
be bounded. We state this as an additional assumption. 

\begin{asm4} The function $\ell(\cdot)$ is bowl-shaped and bounded.\end{asm4}

Assumption 4 implies the estimation loss $l_{n}(\theta,\delta)$ is
bounded. As for the treatment-assignment loss, since we work with
compact priors, it is automatically bounded.

\begin{thm}\label{thm:upper_bound}(Asymptotic optimality of plug-in
rules) Suppose that Assumptions 1-4 hold. Then, under both the estimation
and treatment-assignment loss functions, 
\[
\lim_{M\to\infty}\limsup_{n\to\infty}\max_{\pi(h)\in\Delta_{M}(\mathcal{H})}\int\mathbb{E}_{n,h}\left[e^{l_{n}(\theta_{0}+h/\sqrt{n},\hat{\delta}_{n}^{*})/\lambda}\right]d\pi(h)=V^{*}.
\]
\end{thm}

In fact, the MLE can be replaced by any asymptotically efficient estimator
satisfying Assumption 3. All such estimators attain the same minimax
lower bound $V^{*}$; our theory therefore does not distinguish among
them.

The statements of Theorems \ref{thm:lower_bound} and \ref{thm:upper_bound}
are new even in the absence of misspecification. Standard local asymptotic
minimax theorems are typically stated in terms of a maximum over discrete
sets of $h$ values, rather than the $\max_{\pi(h)\in\Delta_{M}(\mathcal{H})}$
operation employed here, see, e.g., \citet[Proposition 8.11]{van2000asymptotic}.
This distinction is consequential for our setting, as the discrete
formulation does not lend itself to a natural interpretation of prior
ambiguity. Our results avoid this limitation through a novel method
of proof that directly accommodates optimization over a compact space
of local priors.

\subsection{Non-local priors\label{subsec:Non-local-priors}}

The formal results above were derived under priors localized around
a reference parameter $\theta_{0}$. As noted earlier, with a slight
strengthening of the assumptions, we can allow for global priors over
a compact set $\Theta$; essentially, we require the assumptions to
be valid uniformly over $\theta_{0}\in\Theta$. In this case, the
lower bounds on decision risk correspond to the least favorable choice
of the reference.

Formally, we show that under both loss functions, 
\begin{align*}
 & \liminf_{n\to\infty}\,\min_{\delta}\,\max_{\pi(\theta)\in\Delta(\Theta)}\int\mathbb{E}_{n,\theta}\left[e^{l_{n}(\theta,\delta)/\lambda}\right]d\pi(\theta)\\
 & \ge\begin{cases}
\sup_{\theta_{0}\in\Theta}V_{\theta_{0}}^{*} & \textrm{for estimation loss},\\
\sup_{\{\theta_{0}\in\Theta:\mu(\theta_{0})=0\}}V_{\theta_{0}}^{*} & \textrm{for treatment-assignment loss},
\end{cases}
\end{align*}
where $V_{\theta_{0}}^{*}$ represents the minimal decision-risk in
the limit experiment (as defined in Section \ref{subsec:Limit-approximations})
evaluated at a given reference parameter $\theta_{0}$. The formal
statement, together with the required assumptions, is provided in
Appendix \ref{sec:Local-Asymptotics-with-non-local-prior}. 

In the same appendix, we show that the decisions in (\ref{eq:asymptotically_optimal_decisions})
are asymptotically optimal under global priors as well, in that they
attain this bound:
\begin{align*}
 & \lim_{K\to\infty}\,\limsup_{n\to\infty}\,\max_{\pi(\theta)\in\Delta(\Theta)}\int\mathbb{E}_{n,\theta}\left[e^{l_{n,K}(\theta,\delta)/\lambda}\right]d\pi(\theta)\\
 & =\begin{cases}
\sup_{\theta_{0}\in\Theta}V_{\theta_{0}}^{*} & \textrm{for estimation loss},\\
\sup_{\{\theta_{0}\in\Theta:\mu(\theta_{0})=0\}}V_{\theta_{0}}^{*} & \textrm{for treatment-assignment loss},
\end{cases}
\end{align*}
where $l_{n,K}(\cdot):=K\wedge l_{n}(\cdot)$ denotes the loss function
truncated at level $K$. 

These results appear new even in the context of no misspecification
concerns. A number of authors, e.g., \citet[Section 3.2]{hirano2025wald},
have raised concerns about the interpretation of asymptotic analysis
under local reparameterizations. Our results address these concerns
by showing that global minimax analysis is equivalent to local asymptotic
analysis with a least favorable choice of reference.

\subsection{The structure of the least favorable likelihood\label{subsec:Structure_of_lf_lik}}

The variational criterion guards against arbitrary forms of misspecification.
One may wonder, however, whether the least favorable likelihood implied
by the criterion imposes deviations that a decision-maker may not
consider plausible, such as violations of the i.i.d assumption. As
we show below, this concern is largely unfounded. The least favorable
likelihood can be interpreted as a random, parameter-specific location-scale
contamination of the score statistic. Such a contamination could be
generated by perturbing the outcomes in a way that keeps them exchangeable,
and hence conditionally i.i.d, even if not exactly i.i.d.\footnote{There is a key distinction here that deserves highlighting. Nature\textquoteright s
least favorable likelihood contaminates $x_{n}$ and this contamination
can be rationalized as a conditionally i.i.d transformation of the
outcomes. It does not formally follow, however, that nature\textquoteright s
least favorable choice must involve such a transformation of the outcomes:
the same contamination of $x_{n}$ may be achieved in many other ways.}

To understand the nature of the contamination, consider the setting
where the outcomes are normal $Y_{i}|h\sim\mathcal{N}(h/\sqrt{n},I_{0}^{-1})$
and $\mu_{n}(h):=\dot{\mu}_{0}^{\intercal}h$. Under normality, $x_{n}:=\frac{I_{0}^{1/2}}{\sqrt{n}}\sum_{i=1}^{n}Y_{i}$
is a sufficient statistic and the limit experiment provides an exact
representation of the data. As shown earlier, the optimal decisions
under both estimation and treatment-assignment loss are functions
solely of $x_{n}$. Given the optimal decision $\tilde{\delta}^{*}(\cdot)$,
the least favorable likelihood---the one that optimizes (\ref{eq:DV formula})---satisfies
\[
m^{*}(\cdot|h)\propto p_{h}(\cdot)\cdot\exp\left\{ l(h,\tilde{\delta}^{*}(\cdot))/\lambda\right\} ,
\]
where $p_{h}(\cdot)\equiv\mathcal{N}(I_{0}^{1/2}h,I)$ represents
the baseline likelihood for $x_{n}$. Since any distribution can be
represented as a continuous location-scale mixture of normals, there
exists some (possibly degenerate) mixing distribution $\tau_{h}(\nu,\Gamma)$,
such that
\[
m^{*}(\cdot|h):=\int\phi\left(\cdot;\Gamma^{1/2}I_{0}^{1/2}h+\nu,\Gamma\right)d\tau_{h}(\nu,\Gamma),
\]
where $\phi(\cdot|m,A)$ represents the normal density with mean $m$
and variance $A$. 

The representation above implies $m^{*}(\cdot|h)$ can be induced
by transforming the sample mean as $\Gamma^{1/2}x_{n}+\nu$, or equivalently,
transforming the outcome variables as
\[
Y_{i}=(I_{0}^{-1/2}\Gamma^{1/2}I_{0}^{1/2})\tilde{Y}_{i}+(I_{0}^{-1/2}\nu)/\sqrt{n},
\]
where $(\nu,\Gamma)|h\sim\tau_{h}(\cdot)$ and $\tilde{Y}_{i}$ denotes
the outcome under the baseline likelihood $\mathcal{N}(h/\sqrt{n},I_{0}^{-1})$. 

The transformed outcome variables $(Y_{1},\dots,Y_{n})$ are no longer
i.i.d, but they are exchangeable and conditionally i.i.d given $\nu,\Gamma$.
They admit a natural interpretation as location-scale contaminations
of the original outcomes, with the contamination parameters drawn
from the distribution $\tau_{h}(\cdot)$ before the data are observed.
In the special case of squared error loss, taking $h$ to be scalar
for simplicity and defining $\sigma^{2}=\dot{\mu}_{0}^{2}/I_{0}$,
it can be shown that for $\lambda>2\sigma^{2}$, the transformed outcomes
are induced by setting $\Gamma=1$ and
\[
\nu\sim\mathcal{N}\left(0,\frac{2\sigma^{2}}{\lambda-2\sigma^{2}}\right),
\]
independently of $h$.

In our running example, this least favorable likelihood implies that
the outcomes in Pennsylvania are generated as $\mathcal{N}((h+I_{0}^{-1/2}\nu)/\sqrt{n},I_{0}^{-1})$.
The latent variable $\nu$ can thus be interpreted as a multiple of
the difference in average treatment effects between Pennsylvania and
the US population, modeled with its own prior $\mathcal{N}\left(0,2\sigma^{2}/(\lambda-2\sigma^{2})\right)$
that is independent of $h$. This stands in contrast to local misspecification,
where the least favorable likelihood corresponds to a fixed worst-case
drift $\bar{h}/\sqrt{n}$ added to the likelihood (rather than a random
draw from the unbounded $\mathcal{N}\left(0,2\sigma^{2}/(\lambda-2\sigma^{2})\right)$
distribution). For more general losses, the transformation can induce
scale shifts and the distribution of the latent location and scale
parameters generally depends on $h$ as well. 

Although the analysis so far has been restricted to Gaussian reference
likelihoods, for which it is exact, the result extends asymptotically
to general parametric likelihoods $p_{h}(\cdot)$ with score function
$\psi(\cdot)$: the asymptotically least favorable likelihood still
contaminates $x_{n}$ as $\Gamma^{1/2}x_{n}+\nu$, where $(\nu,\Gamma)|h\sim\tau_{h}(\cdot)$.
When the score function $\psi(\cdot)$ is a bijection, this contamination
of $x_{n}$ can, in turn, be induced by transforming the outcomes
as
\[
Y_{i}=\psi^{-1}\left((I_{0}^{-1/2}\Gamma^{1/2}I_{0}^{1/2})\psi(\tilde{Y}_{i})+I_{0}^{1/2}\nu/\sqrt{n}\right),
\]
where $(\nu,\Gamma)|h\sim\tau_{h}(\cdot)$ and $\tilde{Y}_{i}$ denotes
the outcome under the reference likelihood $p_{h}(\cdot)$. 

\subsection{Application: Maximum likelihood vs Simulated method of moments}

\citet{andrews2025purpose} note that in applied work involving parametric
models, researchers often prefer to use simulated method of moments,
which involves subjective selection of moment functions, over maximum
likelihood, despite the latter being more efficient. This preference
may be justified by the argument that ``with misspecification concerns,
moment estimators are often more reliable'' (\citealp{bordalo2020overreaction}). 

Our results suggest, however, that this reasoning might be incomplete.
Under our framework, there is no tradeoff between misspecification
robustness and efficiency. Efficient estimators are misspecification-robust
to an arbitrary degree, and inefficient estimators remain suboptimal
even in the presence of potential misspecification. Unstructured misspecification
concerns alone therefore cannot justify the use of the simulated method
of moments over maximum likelihood.

This is not to suggest that simulated method of moments should not
be used. The selection of moment functions may be understood as a
form of model selection. For instance, if the parameter of interest
is the treatment effect on a subgroup of the population, it is reasonable
to selectively overweight the relevant portion of the sample. This,
however, is a problem of model selection rather than likelihood misspecification
per se.

\section{Semi-parametric Models\label{sec:Semi-parametric-Models}}

The preceding sections focused on parametric models for the likelihood.
In many applications, however, the exact distributional form is left
unspecified, motivating the use of semi-parametric models. As discussed
in Section \ref{subsec:Nuisance-and-structural-parameters}, our methodology
extends naturally to the semi-parametric and nonparametric settings
by allowing for infinite-dimensional nuisance parameters.

In semi-parametric models, the structural parameter is typically a
regular functional $\mu:=\mu(P)$, of the unknown population distribution
$P$. Common examples of regular functionals include the mean, median,
and quantiles. For simplicity, we assume that $\mu$ is scalar-valued.
The loss functions are then
\begin{align*}
l_{n}(P,\delta) & =\ell\bigl(\sqrt{n}\,(\mu(P)-\delta)\bigr),\quad\textrm{and }\\
l_{n}(P,\delta) & =\sqrt{n}\bigl(\mu(P)\,\mathbf{1}\{\mu(P)\geq0\}-\mu(P)\,\delta\bigr),\quad\delta\in\{0,1\},
\end{align*}
for estimation and treatment-assignment losses, respectively. The
population distribution $P$ plays the same role as $\theta$ in the
parametric analysis. While $P$ is unknown, we are only interested
in its scalar functional $\mu(P)$; the remaining features of $P$
are treated as an infinite dimensional nuisance parameter. 

As in Section \ref{subsec:Bayesian,-Frequentist-and-misspecified},
we consider a framework in which the decision-maker confronts both
prior ambiguity, i.e., an inability to form a single prior over the
parameter $P$, and (semi-parametric) model misspecification, i.e.,
the concern that the population distribution $P$ may not coincide
with the outcome distribution $\hat{P}$ in the experimental sample.
In our running example, Alice may adopt a fully nonparametric specification
of the outcome distribution yet still face misspecification concerns,
as treatment effects in Pennsylvania could differ fundamentally from
those in the broader US population. More generally, misspecification
also arises when the semi-parametric model imposes restrictions that
are not satisfied by the combination of the true $\mu$ and the sample
outcome distribution $\hat{P}$. For instance, in the GMM framework,
the researcher specifies a moment condition $\mathbb{E}_{P}[m(Y_{i},\mu)]=0$,
where $m(\cdot)\in\mathbb{R}^{p}$ is a known vector of moment functions,
$\mu\in\mathbb{R}^{d}$ is the structural parameter of interest, and
$P$ is the unknown population distribution. Misspecification arises
if this restriction does not hold in the sample, i.e., $\mathbb{E}_{\hat{P}}[m(Y_{i},\mu)]\neq0$
for the true $\mu$. In our running example, which fits within the
GMM framework with $m(Y_{i},\mu)=Y_{i}-\mu$, Alice may be concerned
that $\mathbb{E}_{\hat{P}}[Y_{i}-\mu]\neq0$. 

The central message of our semi-parametric results is that the sample
analogs of the optimal decisions characterized in Section \ref{subsec:Characterization-of-optimal-decisions}
remain optimal in the semi-parametric setting. In essence, one replaces
the score statistic from the parametric setting with an efficient
influence function process associated with the functional of interest.
For example, if the goal is to conduct inference on the mean, one
replaces $x_{n}$ with the cumulative sum of outcomes $n^{-1/2}\sum_{i=1}^{n}(Y_{i}/\sigma)$,
where $\sigma^{2}:=\textrm{Var}[Y_{i}]$. 

\subsection{Local asymptotics for semi-parametric models\label{subsec:local_asymptotics_nonparametrics}}

It is easiest to discuss ambiguity and misspecification in semi-parametric
settings using local asymptotics, employing the formalism of \citet[Section 25.3]{van2000asymptotic}.
Let $\mathcal{P}$ denote the class of candidate population distributions
with bounded variance, dominated by some measure $\nu$. We fix a
reference distribution $P_{0}\in\mathcal{P}$, and surround it with
various smooth one-dimensional parametric sub-models, $\{P_{s,h}:s\le\zeta\}$
for some $\zeta>0$, whose score function is $h(\cdot)$, and that
pass through $P_{0}$ at $s=0$ (i.e., $P_{0,h}=P_{0}$). Formally,
we require these sub-models to satisfy
\begin{equation}
\int\left[\frac{dP_{s_{n},h_{n}}^{1/2}-dP_{0}^{1/2}}{s_{n}}-\frac{1}{2}h_{n}dP_{0}^{1/2}\right]^{2}d\nu\to0\ \textrm{as}\ s_{n}\to0\ \textrm{and }h_{n}\to h\ \textrm{in }L_{2}(P_{0}).\label{eq:qmd non-parametrics}
\end{equation}
Here, $L^{2}(P_{0})$ represents the Hilbert space with the inner
product $\left\langle f,g\right\rangle =\mathbb{E}_{P_{0}}[fg]$ and
norm $\left\Vert f\right\Vert =\mathbb{E}_{P_{0}}[f^{2}]^{1/2}$,
and $h_{n}\to h$ is in the sense of $\left\Vert h_{n}-h\right\Vert \to0$. 

As in the parametric setting, for the treatment-assignment problem,
we choose $P_{0}$ such that $\mu(P_{0})=0$. 

By \citet[Lemma 25.14]{van2000asymptotic}, condition (\ref{eq:qmd non-parametrics})
implies $\int hdP_{0}=0$ and $\int h^{2}dP_{0}<\infty$. The set
of all such functions $h$ is termed the tangent space $T(P_{0})$,
which is a subset $L^{2}(P_{0})$. For any $h\in T(P_{0})$, let $P_{n,h}$
denote the joint probability measure over $Y_{1},\dots,Y_{n}$, when
each $Y_{i}$ is an iid draw from $P_{1/\sqrt{n},h}$, and let $\mathbb{E}_{n,h}[\cdot]$
denote the corresponding expectation. An important consequence of
(\ref{eq:qmd non-parametrics}) is the sequential LAN property: 
\begin{align}
\sum_{i=1}^{n}\ln\frac{dP_{1/\sqrt{n},h_{n}}}{dP_{0}}(Y_{i}) & =\frac{1}{\sqrt{n}}\sum_{i=1}^{n}h(Y_{i})-\frac{1}{2}\left\Vert h\right\Vert ^{2}+o_{P_{n,0}}(1),\ \textrm{for all }h_{n}\to h.\label{eq:SLAN nonparametric setting}
\end{align}

Let $\psi\in T(P_{0})$ denote the efficient influence function corresponding
to $\mu$, defined by the property that for any $h\in T(P_{0})$,
\begin{equation}
\frac{\mu(P_{s,h})-\mu(P_{0})}{s}-\left\langle \psi,h\right\rangle =o(s).\label{eq:influence function}
\end{equation}
Set $\sigma^{2}=\mathbb{E}_{P_{0}}[\psi^{2}]>0$. The semi-parametric
analogue of the score statistic in the semi-parametric setting is
the standardized efficient influence function process
\[
x_{n}:=\frac{\sigma^{-1}}{\sqrt{n}}\sum_{i=1}^{n}\psi(Y_{i}).
\]

Any element $h\in T(P_{0})$ admits the orthogonal decomposition $h=\left\langle \psi/\sigma,h\right\rangle \psi/\sigma+\tilde{h}$,
where $\tilde{h}$ is orthogonal to $\psi$ (i.e., $\left\langle \psi,\tilde{h}\right\rangle =0$).
The component $\mu=\left\langle \psi,h\right\rangle $ represents
the structural parameter, while $\tilde{h}$ represents an infinite
dimensional nuisance parameter. Although the full perturbation direction
$h$ is unknown, only the projection onto the efficient influence
function is relevant for learning about $\mu$. 

For each $h\in T(P_{0})$, define $\mu_{n}(h):=\mu(P_{1/\sqrt{n},h})$,
and let $\bm{x}:=(Y_{1},\dots,Y_{n})$ denote the collection of outcomes.
We can rewrite the loss functions in terms of $h$ as 
\[
l_{n}(h,\delta)=\begin{cases}
\ell\bigl(\sqrt{n}\,(\mu_{n}(h)-\delta)\bigr) & \textrm{for estimation loss},\\
\sqrt{n}\bigl(\mu_{n}(h)\,\mathbf{1}\{\mu_{n}(h)\geq0\}-\mu_{n}(h)\,\delta\bigr) & \textrm{for treatment-assignment loss}.
\end{cases}
\]

In the semi-parametric setting, a Bayesian statistical model, $m(h,\bm{x})$,
is defined as a joint probability distribution over both $h,\bm{x}$.
As in Section \ref{subsec:Bayesian,-Frequentist-and-misspecified},
it admits the decomposition
\[
m(h,\bm{x})=\pi(h)\otimes p_{n,h}(\bm{x}),
\]
where $\pi\in\Delta(T(P_{0}))$ denotes a prior over the tangent space
$T(P_{0})$, and $p_{n,h}(\bm{x}):=\prod_{i}dP_{1/\sqrt{n},h}(Y_{i})/d\nu$
represents the likelihood, a parametric sub-model. Prior ambiguity
is incorporated by defining a frequentist model, a structured set
of models, as
\[
\mathcal{Q}:=\bigl\{\pi(h)\otimes p_{n,h}(\bm{x}):\pi\in\Delta(T(P_{0}))\bigr\}.
\]
Finally, misspecification concerns are addressed by placing a protective
belt around $\mathcal{Q}$, yielding the set of unstructured models
\[
\mathcal{M}=\Bigl\{ m\in\Delta(T(P_{0}),\mathcal{X}):\min_{q\in\mathcal{Q}}\,R_{q}(m)\leq K\Bigr\}.
\]
Expanding the set of structured models allows us to account for the
possibility that the true distribution of the experimental data $\bm{x}$
is not captured by $p_{n,h}(\bm{x})$ for any $h$. 

The decision-maker chooses the decision rule that performs best against
the worst-case model in $\mathcal{M}$, thereby guarding against both
prior ambiguity and misspecification:
\[
\delta_{n}^{*}:=\argmin_{\delta}\left[\sup_{m\in\mathcal{M}}\mathbb{E}_{m}\bigl[l_{n}(h,\delta)\bigr]\right].
\]
As in Section \ref{sec:Characterizing-Optimal-Decisions}, applying
a Lagrangian formulation and the same sequence of calculations yields
the following characterization of minimal decision risk:
\begin{equation}
V_{n}^{*}=\min_{\delta}\,\max_{\pi\in\Delta(T(P_{0}))}\int\mathbb{E}_{n,h}\!\left[e^{l_{n}(h,\delta)/\lambda}\right]d\pi(h).\label{eq:minimax_value_non_parametrics}
\end{equation}

\subsection{Formal results: Semi-parametric models}

We impose the following regularity conditions throughout this section:

\begin{asm5}The sub-models $\{P_{s,h};h\in T(P_{0})\}$ satisfy (\ref{eq:qmd non-parametrics}).
Furthermore, they admit an efficient influence function $\psi(\cdot)$
for $\mu(P)$ in that for each $h_{n}\to h\in T(P_{0})$ satisfying
$\left\Vert h\right\Vert \le M<\infty$, and with $\mu_{0}:=\mu(P_{0})$,
\[
\sqrt{n}\left(\mu(P_{1/\sqrt{n},h_{n}})-\mu_{0}\right)=\left\langle \psi,h\right\rangle +o(1).
\]
\end{asm5}

The first part of Assumption 5 restates the definition of parametric
sub-models. The second part of Assumption 5 slightly strengthens (\ref{eq:influence function}). 

Since $L^{2}(P_{0})$ is a Hilbert space, it is possible to select
$\{\phi_{1},\phi_{2},\dots\}\in L^{2}(P_{0}^ {})$ in such a manner
that $\{\psi/\sigma,\phi_{1},\phi_{2},\dots\}$ is a set of orthonormal
basis functions for the closure of $T(P_{0}^ {})$; the division by
$\sigma_{}$ in the first component ensures $\left\Vert \psi/\sigma_{}\right\Vert ^{2}=1$.
We can also choose these bases so they lie in $T(P_{0}^ {})$, i.e.,
$\mathbb{E}_{P_{0}}[\phi_{j}]=0$ for all $j$. By the Hilbert space
isometry, each $h_{}\in T(P_{0}^ {})$ is then associated with an
element from the $l_{2}$ space of square integrable sequences, $(\mu/\sigma,\gamma_{1},\gamma_{2},\dots)$,
where $\mu=\left\langle \psi,h\right\rangle $ and $\gamma_{k}=\left\langle \phi_{k},h\right\rangle _{}$
for all $k\neq0$. Consequently, any prior $\pi(h)$ over $T(P_{0})$
can be represented as a prior over $l_{2}$. 

As in Section \ref{subsec:Formal-results-parametric}, our formal
results require localization of ambiguity. Our approach, which seems
to be novel, involves restricting attention to priors that are supported
on a compact subset, $K_{M}$, of $l_{2}$, defined as 
\[
K_{M}\equiv\left\{ h=\left(\mu/\sigma,\gamma_{1},\dots\right):\left\Vert h\right\Vert \le M,\ \sum_{j=J}^{\infty}\vert\gamma_{j}\vert^{2}\le\varepsilon_{J}\ \textrm{for some }\varepsilon_{J}\downarrow0\ \textrm{as }J\to\infty\right\} .
\]
The compactness condition essentially requires the set of candidate
$h$ to be sufficiently smooth. It is therefore quite mild.

\subsubsection{Lower bounds\label{subsec:Lower-bounds-nonparametrics}}

As in the parametric setting, we show that the minimax value $V^{*}$
in the limit experiment also forms an asymptotic lower bound on the
sequence of optimal decision risks, $V_{n}^{*}$, in the finite sample
semi-parametric experiments. However, since the previous definition
of the limit experiment used a different interpretation of $h$, we
will need to modify the construction slightly.

Specifically, we now consider a limit experiment where we observe
a one dimensional signal $x$, posited to be drawn from a reference
Gaussian likelihood, $P_{\mu}(\cdot)\equiv\mathcal{N}(\mu/\sigma,1)$.
Let $\mathbb{E}_{\mu}[\cdot]$ denote the expectation corresponding
to $P_{\mu}$. The minimax value $V^{*}$ in this limit experiment
is then defined as
\begin{align}
V^{*} & :=\min_{\tilde{\delta}}\,\max_{\rho(\mu)\in\Delta(\mathbb{R})}\int\mathbb{E}_{\mu}\left[e^{l(\mu,\tilde{\delta})/\lambda}\right]d\rho(\mu),\ \textrm{with}\label{eq:minimax_value_limit_experiment_nonparametrics}\\
l(\mu,\tilde{\delta}) & =\begin{cases}
\ell(\mu-\tilde{\delta}) & \textrm{for estimation loss},\\
\mu\cdot\left\{ \mathbf{1}\{\mu\ge0\}-\tilde{\delta}\right\}  & \textrm{for treatment-assignment loss}.
\end{cases}\nonumber 
\end{align}
It is straightforward to verify that the value of $V^{*}$ in (\ref{eq:minimax_value_limit_experiment_nonparametrics})
is, in fact, the same as that in (\ref{eq:minimax_value_limit_experiment})
when $\theta=\mu$ and $I_{0}=1/\sigma^{2}$. 

\begin{thm}\label{thm:lower_bound-nonparametric}Suppose that Assumption
5 holds. Then, under both the estimation and treatment-assignment
loss functions, 
\[
\lim_{M\to\infty}\liminf_{n\to\infty}\min_{\delta}\,\max_{\pi(h)\in\Delta(K_{M})}\int\mathbb{E}_{n,h}\left[e^{l_{n}(h,\delta)/\lambda}\right]d\pi(h)\ge V^{*}.
\]
\end{thm}

\subsubsection{Asymptotically optimal decisions}

As in the parametric setting, asymptotically optimal decisions under
ambiguity and misspecification are the same as those under ambiguity
alone. Let $\hat{\mu}_{n}$ denote any semi-parametrically efficient
estimator for $\mu$, understood as satisfying the following assumption:

\begin{asm6} The estimator $\hat{\mu}$ attains the semi-parametric
efficiency bound, in that it admits a locally linear influence-function
approximation:
\[
\hat{\mu}_{n}-\mu_{0}=\sigma x_{n}+o_{P_{n,0}}(1).
\]
 \end{asm6}

As we show below, asymptotically optimal decisions are then given
by 
\[
\hat{\delta}_{n}^{*}=\begin{cases}
\hat{\mu}_{n} & \textrm{for estimation,}\\
\mathbf{1}\left\{ \hat{\mu}_{n}\geq0\right\}  & \textrm{for treatment assignment.}
\end{cases}
\]

\begin{thm}\label{thm:upper_bound-non-parametric}Suppose that Assumptions
4-6 hold. Then, under both the estimation and treatment-assignment
loss functions, 
\[
\lim_{M\to\infty}\limsup_{n\to\infty}\max_{\pi(h)\in\Delta(K_{M})}\int\mathbb{E}_{n,h}\left[e^{l_{n}(h,\hat{\delta}_{n}^{*})/\lambda}\right]d\pi(h)=V^{*}.
\]
\end{thm}

\subsection{Application: Optimal GMM estimators under misspecification}

The Generalized Method of Moments (GMM) is an example of a semi-parametric
model that is widely used in economic applications. Recall that in
the GMM framework, the researcher specifies a moment condition $\mathbb{E}_{P}[m(Y_{i},\mu)]=0$,
where $m(\cdot)\in\mathbb{R}^{p}$ is a known vector of moments, $\mu\in\mathbb{R}^{d}$
is the structural parameter, and $P$ is the population distribution.
When $p>d$, the GMM model is said to be over-identified. In this
setting, the efficient influence function is given by 
\[
\psi(Y_{i})=G_{0}^{\intercal}\Omega_{0}^{-1}m(Y_{i},\mu_{0}),
\]
where $\mu_{0}:=\mu(P_{0})$ is the unique solution to $\mathbb{E}_{P_{0}}[m(Y_{i},\mu)]=0$
under the reference distribution $P_{0}$, $G_{0}:=\mathbb{E}_{P_{0}}[\nabla_{\mu}m(Y_{i},\mu)]$
and $\Omega_{0}:=\mathbb{E}_{P_{0}}[m(Y_{i},\mu_{0})m(Y_{i},\mu_{0})^{\intercal}]$. 

In the absence of misspecification concerns, it is well known that
several estimators are asymptotically efficient, including 2-step
GMM and continuously updated GMM (CUGMM), among others. In practice,
however, researchers are often concerned that the model may be misspecified,
i.e., $\mathbb{E}_{\hat{P}}[m(Y_{i},\mu)]\neq0$ under the true structural
parameter $\mu^ {}$ and the sample outcome distribution $\hat{P}$.
Under misspecification, different estimators converge to different
limits under the distribution of sample outcomes $\hat{P}$, and researchers
often resort to inefficient estimators employing a weighing matrix
other than the optimal $W^{*}=G_{0}^{\intercal}\Omega_{0}^{-1}$.
This practice may be justified on the grounds that ``...under misspecification
the two-step GMM estimator is no more efficient than any other estimator...
each weighting leads us to recover a different parameter'' (\citealp{andrews2025purpose}). 

Our results suggest, however, that this reasoning might be incomplete.
When a decision-maker confronts both prior ambiguity and model misspecification,
as in our framework, the optimal estimator coincides with that under
prior ambiguity alone. Consequently, 2-step GMM remains superior to
diagonally-weighted GMM, even under an arbitrary degree of misspecification.\footnote{\citet{armstrong2021sensitivity} derive related results under the
more restrictive local misspecification framework (of the kind discussed
in Section \ref{subsec:Alternative-approaches}). Their analysis is
also restricted to asymptotically linear estimators.} Researchers who wish to employ inefficient weighting may therefore
want to provide explicit justification for doing so. If, for instance,
the researcher believes that misspecification is not unstructured
but that certain forms or directions of misspecification are more
likely than others, this could lead to different estimation strategies.
Even so, it may be difficult to justify use of identity or diagonal
weighting on the basis of directional misspecification alone, since
such weighting typically preserves symmetry across directions. Of
course, going beyond misspecification concerns, one could appeal to
other, finite-sample rationales for picking these weighting matrices.

Apart from requiring efficiency, our results do not distinguish among
different efficient estimators. Under local asymptotics, all efficient
estimators attain the same decision risk $V^{*}$, so additional criteria
must be employed to select among them.\footnote{Note that various efficient estimators may also be associated with
different distributional limits under a given $\theta$. Our results
only concern worst-case performance, under which they are indistinguishable.} For instance, imposing invariance would lead one to prefer CUGMM
over two-step GMM. 

\section{Axiomatics\label{sec:Axiomatics}}

This article is primarily concerned with the econometric and statistical
implications of employing misspecification-averse preferences as formulated
by \citet{cerreia2026making}. As noted in Section \ref{sec:Characterizing-Optimal-Decisions},
while the constrained and variational decision criteria in (\ref{eq:optimal_decision})
and (\ref{eq:variational_criterion}) lead to equivalent characterizations
of optimal decisions, it is the variational criterion that admits
a more direct axiomatic justification. Since the latter is just a
special case of \citet{cerreia2026making}, the axiomatic foundations
for our framework follow directly from that paper once the decision-theoretic
ingredients are specified.

Following \citet{cerreia2026making}, we work within an Anscombe--Aumann
framework. The state space $\omega=(\theta,\bm{x})$ comprises both
the parameter and the data. The consequence space is defined separately
for the two problems under consideration. For estimation, we associate
consequences $z\in\mathbb{R}$ with estimation errors $\mu(\theta)-\delta$.
Allowing for randomized estimators enables us to view these as Anscombe--Aumann
acts, i.e., measurable mappings from $\Omega$ to $\mathbb{R}$. For
treatment assignment, we associate consequences with ex-post regret,
$\mu(\theta)\,\mathbf{1}\{\mu(\theta)\geq0\}-\mu(\theta)\,\delta$.
Following the normative framework of \citet{cerreia2026making}, we
assume the decision-maker posits the set of structured models $\mathcal{Q}$
defined in (\ref{eq:ambiguity-set}). We then adopt the same axioms,
either the single-preference variational model of \citet[Appendix A]{cerreia2026making}
or the two-preference model of \citet[Sections 3 and 4]{cerreia2026making}.
Since the axioms can be employed unchanged, we do not restate them
here. \citet{cerreia2026making} show that both sets of preferences
admit a representation via the variational decision criterion (\ref{eq:variational_criterion}).

Two remarks are in order. First, as in any decision-theoretic environment,
the loss function $l_{n}(\theta,\delta)$, i.e., the negative of the
Bernoulli utility function, is endogenously determined by preferences
rather than imposed exogenously as we have done throughout. That said,
as in the standard Anscombe--Aumann framework, Bernoulli utilities
can be elicited through direct lotteries over the consequence space.
In the case of estimation, our loss functions can be motivated by
supposing that the decision-maker is risk-averse and therefore employs
a bowl-shaped loss over consequences. Similarly, the treatment-assignment
loss is motivated by supposing that the decision-maker is risk-neutral
and therefore employs a linear map on the consequence space of ex-post
regret.

Second, the two-preference model of \citet{cerreia2026making} permits
the use of discrepancy measures other than KL divergence. Pinning
down the latter requires an additional axiom, as in \citet[Appendix A]{cerreia2026making}.
We discuss alternative discrepancy measures in Section \ref{subsec:Alternative-discrepancy-measures}.

\citet{andrews2026misspecification} provide an alternative axiomatization
of the variational decision-risk criterion (\ref{eq:variational_criterion}).
Their approach begins by axiomatizing a family of conditional preferences
representing the decision-maker's evaluations given each parameter
value $\theta$, and then imposes an axiom requiring that the degree
of misspecification concern $\lambda$ be independent of $\theta$.
In our framework, by contrast, $\lambda$ represents a single degree
of misspecification concern against the entire class of models $\mathcal{Q}$.
Since $\mathcal{Q}$ includes all possible priors, our definition
of a statistical model is richer, and the two axiomatizations---while
yielding the same decision criterion---fundamentally differ in the
structure they impose on statistical models.

\section{Extensions}

In this section, we discuss several variations and extensions of our
framework. 

\subsection{Alternative discrepancy measures\label{subsec:Alternative-discrepancy-measures}}

So far, we have used relative entropy $R_{q}(m)$ to measure the discrepancy
between statistical models. \citet[Chapter 1.8]{hansen2011robustness}
provide two important reasons for using relative entropy. First, it
leads to a tractable characterization of minimal decision risk and
optimal decisions, as demonstrated in Section \ref{sec:Characterizing-Optimal-Decisions}.
Second, as discussed in \citet[Chapter 9]{hansen2011robustness},
it can be linked to risk-sensitivity adjustment through the theory
of large deviations. 

It is, however, possible to employ alternative discrepancy measures.
Relative entropy, is but a special case of the $\phi$-divergence
class of discrepancies, which take the general form 
\[
D_{\phi}(m\|q)=\int\phi\left(\frac{dm}{dq}\right)dq,
\]
where $\phi:[0,\infty)\to(-\infty,\infty]$ is a convex function satisfying
$\phi(x)<\infty$ for all $x>0$, $\phi(1)=0$ and $\phi(0)=\lim_{x\to0}\phi(x)$.
Setting $\phi(x)=x\ln x$ recovers KL divergence. 

Under this more general class of discrepancies, the decision-risk
of a rule $\delta$ can be characterized in variational form as 
\[
V_{n,\phi}(\delta)=\min_{q\in\mathcal{Q}}\inf_{m}\left\{ \mathbb{E}_{m}\bigl[u_{n}(\theta,\delta)\bigr]+\lambda\,D_{\phi}(m\|q)\right\} .
\]
Let $\phi^{*}(\cdot)$ denote the convex conjugate of $\phi(\cdot)$.
The decision risk $V_{n,\phi}(\delta)$ then admits the dual representation
\begin{align*}
V_{n,\phi}(\delta) & =\lambda\min_{q\in\mathcal{Q}}\sup_{\eta\in\mathbb{R}}\left\{ \eta-\int\phi^{*}\left(\eta-\frac{u_{n}(\theta,\delta)}{\lambda}\right)dq\right\} \\
 & =\lambda\min_{\pi\in\Delta(\Theta)}\sup_{\eta\in\mathbb{R}}\left\{ \eta-\int\mathbb{E}_{p(\bm{x}|\theta)}\left[\phi^{*}\left(\eta+\frac{l_{n}(\theta,\delta)}{\lambda}\right)\right]d\pi(\theta)\right\} ,
\end{align*}
where the first equality is due to \citet[Proposition 1]{cerreia2026making}
and the second equality exploits the specific structure of $\mathcal{Q}$
in our setting. 

By standard properties of convex conjugates, $\phi^{*}(x)<\infty$
for all bounded $x$ if and only if $\lim_{x\to\infty}\phi(x)/x=\infty$.
Since our loss functions $l_{n}(\theta,\delta)$ are generally unbounded,
$V_{n,\phi}(\delta)$ is therefore $\infty$ for any $\phi$-divergence
measure satisfying $\lim_{x\to\infty}\phi(x)/x<\infty$, unless sharp
support restrictions are imposed on the class of priors. In other
words, the decision risk is trivially infinite for these divergence
measures whenever the class of priors is sufficiently rich. We therefore
argue that such discrepancies are not well suited for a framework
that accommodates both ambiguity and misspecification. Notable examples
in this category include total variation $(\phi(x)=\vert x-1\vert/2$),
squared Hellinger distance ($\phi(x)=(\sqrt{x}-1)^{2}/2)$ and Pearson
$\chi^{2}$ divergence ($\phi(x)=(x-1)^{2}/x$). 

The underlying issue with these divergence measures is that the misspecification
they permit is too broad. Consider, for instance, the total-variation
metric: it is possible to have $dm/dq=\infty$, meaning $q\in\mathcal{Q}$
may not share the same support as the true model, even as total variation
remains finite. This would imply that $\mathcal{Q}$ is blatantly
misspecified, in the sense that any specification test would almost
surely reject it. In contrast, as \citet[Chapter 9]{hansen2011robustness}
argue, when $R_{\mathcal{Q}}(m)<\infty$, the approximating model
can still be regarded as plausibly correct, since it would not be
rejected with probability one. 

Among the $\phi$-divergence measures satisfying $\lim_{x\to\infty}\phi(x)/x=\infty$,
the only commonly used divergence apart from KL divergence is the
Neyman $\chi^{2}$ divergence ($\phi(x)=(x-1)^{2}$). The Neyman $\chi^{2}$
divergence is a stronger discrepancy measure than KL divergence: it
is possible to have $\chi^{2}(m\|q)<\infty$ while $R_{q}(m)=\infty$.
Consequently, a KL-based a misspecification set $\mathcal{M}=\Bigl\{ m\in\Delta(\Theta,\mathcal{X}):\min_{q\in\mathcal{Q}}\,D_{q}(m)\leq K\Bigr\}$
is strictly larger than the corresponding $\chi^{2}$-based set $\mathcal{M}_{\chi^{2}}=\Bigl\{ m\in\Delta(\Theta,\mathcal{X}):\min_{q\in\mathcal{Q}}\,D_{\chi^{2}}(m\|q)\leq K\Bigr\}$.
Based on this, we conjecture that our results on the optimality of
efficient decisions continue to hold for the Neyman $\chi^{2}$ divergence
as well, though we leave the formal analysis to future work. 

\subsection{Asymmetric loss functions}

The estimation and treatment assignment loss functions considered
thus far share a crucial property: they are symmetric in the sense
that overestimating $\theta$ by a given amount incurs the same loss
as underestimating it by that amount. This symmetry is crucial to
the result that optimal decisions do not depend on the degree of misspecification. 

There are, however, many loss functions that lack this symmetry. A
prominent example is the linex loss $l(\theta,\delta)=e^{(\mu(\theta)-\delta)}-(\mu(\theta)-\delta)-1$
which penalizes positive errors much more than negative errors. Under
such losses, the optimal estimator is biased even in the absence of
misspecification. Incorporating misspecification concerns introduces
additional bias whose magnitude depends on the misspecification parameter
$\lambda$; see Appendix \ref{sec:Asymmetric-Loss-Functions} for
details.\footnote{Under misspecification, the linex loss function must be truncated
to keep the minimax risk finite. The optimal estimator therefore also
depends on the level of truncation. However, for any level of truncation,
the bias of the optimal estimator decreases in $\lambda$. } Intuitively, misspecification entails an exponential tilting of the
loss function, which further exacerbates any asymmetry already present
in the loss. Consequently, for asymmetric losses, optimal decisions
under ambiguity and misspecification may not coincide with those under
ambiguity alone.

\subsection{Relaxing caution: Smooth ambiguity aversion and hierarchical Bayes\label{subsec:Relaxing-caution}}

Our framework employs the Waldian approach to ambiguity by selecting
the worst-case model within the ambiguity set $\mathcal{Q}$. \citet{cerreia2026making}
term the preference axiom underlying this approach as `caution'. An
alternative is to adopt smooth ambiguity aversion, as in \citet{klibanoff2005smooth},
by introducing a probability distribution $\varrho_{Q}(\cdot)$ over
$\mathcal{Q}$. Intuitively, smooth ambiguity aversion corresponds
to making the decision-maker less cautious: rather than guarding against
the worst case, she averages over models according to $\varrho_{Q}(\cdot)$.
Since $\mathcal{Q}$ pairs a candidate likelihood with every possible
prior, a distribution over $\mathcal{Q}$ is equivalent to a distribution
$\varrho(\pi)$ over the space of priors $\Delta(\Theta)$; the latter
can simply be interpreted as a hyperprior in a Bayesian hierarchical
model.

\citet[Section 6.1]{cerreia2026making} show that under smooth ambiguity
aversion, the decision risk takes the form 
\begin{align*}
\tilde{V}_{n}(\delta) & =\int_{\mathcal{Q}}\phi_{Q}\left(\inf_{m}\left\{ \mathbb{E}_{m}\bigl[u_{n}(\theta,\delta)\bigr]+\lambda\,R_{q}(m)\right\} \right)d\varrho_{\mathcal{Q}}(q)\\
 & =\int_{\mathcal{Q}}\phi_{Q}\left(-\lambda\ln\mathbb{E}_{q}\!\left[e^{-u_{n}(\theta,\delta)/\lambda}\right]\right)d\varrho_{Q}(q),
\end{align*}
where $\phi_{Q}(\cdot)$ is a monotone function and the second equality
uses (\ref{eq:DV formula}). Setting $\phi_{Q}(t)=-e^{-t/\lambda}$
and converting $\varrho_{Q}(q)$ to a prior $\varrho(\pi)$ over $\Delta(\Theta)$
yields
\begin{align*}
\tilde{V}_{n}(\delta) & =-\int_{\Delta(\Theta)}\left(\int\mathbb{E}_{p(\bm{x}|\theta)}\!\left[e^{l_{n}(\theta,\delta)/\lambda}\right]d\pi(\theta)\right)d\varrho(\pi)\\
 & =-\int\mathbb{E}_{p(\bm{x}|\theta)}\!\left[e^{l_{n}(\theta,\delta)/\lambda}\right]d\bar{\pi}(\theta),
\end{align*}
where $\bar{\pi}(\theta)$ is the effective prior induced by the hyperprior
$\varrho(\cdot)$ over $\Delta(\Theta)$. Hence, misspecification
combined with smooth ambiguity aversion when $\phi_{Q}(t)=-e^{-t/\lambda}$
is equivalent to misspecification with a single hierarchical prior. 

The choice of $\phi_{Q}(t)=-e^{-t/\lambda}$, however, uses the same
parameter $\lambda$ to govern both aversion to prior ambiguity and
sensitivity to likelihood misspecification. It may therefore be more
natural to set $\phi_{Q}(t)=-e^{-t/\xi}$, where $\xi$ captures aversion
to prior uncertainty separately from the misspecification-aversion
parameter $\lambda$. With this choice, the decision risk becomes
\begin{align}
\tilde{V}_{n}(\delta) & =-\int_{\mathcal{Q}}\left(\mathbb{E}_{q}\!\left[e^{l_{n}(\theta,\delta)/\lambda}\right]\right)^{\lambda/\xi}d\varrho_{Q}(q)\nonumber \\
 & =-\int_{\Delta(\Theta)}\left(\int\mathbb{E}_{p(\bm{x}|\theta)}\!\left[e^{l_{n}(\theta,\delta)/\lambda}\right]d\pi\right)^{\lambda/\xi}d\varrho(\pi).\label{eq:decision-risk-general}
\end{align}
As $\xi\to0$, aversion to prior ambiguity grows without bound and
$\tilde{V}_{n}(\delta)$ reduces to the decision risk $V_{n}(\delta)$
from (\ref{eq:decision-risk-eq})---the Waldian formulation involving
the least favorable prior, and the primary focus of this article.
For $\xi\in(0,\infty)$, the decision-maker exhibits less aversion
to prior ambiguity, but this comes at the cost of a nonlinear objective
when $\lambda\neq\xi$, which complicates the analysis. A formal treatment
of the criterion (\ref{eq:decision-risk-general}) when $\xi\notin\{0,\lambda\}$
is therefore left for future research. 

\subsection{Model aggregation\label{subsec:Model-averaging}}

While we have so far focused on misspecification of a single likelihood,
practitioners are often interested in aggregating multiple competing
likelihood specifications, each potentially subject to varying degrees
of misspecification concern. In this section, we propose a framework
for model aggregation under such misspecification concerns, albeit
without a formal axiomatic justification. 

Let $p_{1,\theta}(\bm{x})\equiv p_{1}(\bm{x}|\theta)$ and $p_{2,\theta}(\bm{x})\equiv p_{2}(\bm{x}|\theta)$
denote two candidate likelihoods, and let $\mathcal{Q}_{1},\mathcal{Q}_{2}$
denote the corresponding frequentist models, where, as in Section
\ref{subsec:Models-with-prior-ambiguity},
\[
\mathcal{Q}_{a}:=\bigl\{\pi(\theta)\otimes p_{a,\theta}(\bm{x}):\pi\in\Delta(\Theta)\bigr\},\quad a\in\{1,2\}.
\]
Suppose that, treating each likelihood in isolation, the decision
maker contemplates a misspecification set for each of the form
\[
\mathcal{M}_{a}=\Bigl\{ m\in\Delta(\Theta,\mathcal{X}):\min_{q\in\mathcal{Q}_{a}}R_{q}(m)\leq K_{a}\Bigr\},\quad a\in\{1,2\}.
\]
Here, $K_{a}$ quantifies the decision-maker's misspecification concern
for each model: $K_{2}\gg K_{1}$ implies that the decision maker
has greater concern about likelihood 2 being misspecified than about
likelihood 1. Our framework permits the two likelihoods to be nested,
i.e., $\{p_{2,\theta}(\bm{x})\}_{\theta}\subseteq\{p_{1,\theta}(\bm{x})\}_{\theta}$,
which can be achieved by fixing some components of $\theta$ in $p_{2,\theta}(\bm{x})$.
For instance, in the running example, Alice may be choosing between
two specifications $Y\sim\mathcal{N}(\mu,\sigma^{2})$ and $Y\sim\mathcal{N}(\mu,1)$;
here, $\theta:=(\mu,\sigma^{2})$, but $\sigma^{2}=1$ is fixed in
the second specification. 

Rather than treating the models in isolation, the decision-maker may
wish to aggregate them. A natural choice is to take the overall set
of misspecified models to be the intersection
\[
\mathcal{M}=\mathcal{M}_{1}\cap\mathcal{M}_{2},
\]
which we assume to be non-empty; otherwise, the two models would be
deemed incompatible. The intuition is that if the decision-maker,
having contemplated each likelihood in isolation, has judged the models
in $\mathcal{M}_{1},\mathcal{M}_{2}$ to be individually reasonable,
then it is natural to restrict attention to models contained in both.\footnote{Setting $\mathcal{M}=\mathcal{M}_{1}\cap\mathcal{M}_{2}$ corresponds
to a normative framework in which Alice seeks to \textit{aggregate},
or pool the information present across different models. Alternatively,
Alice may be interested in model \textit{selection}, i.e., choosing
the best-approximating model. This corresponds to $\mathcal{M}=\mathcal{M}_{1}\cup\mathcal{M}_{2}$.
The axiomatization and comparison of these two normative frameworks
is described in a companion paper. }

With this choice of $\mathcal{M}$, the optimal decision becomes
\begin{align*}
\delta_{n}^{*} & :=\argmin_{\delta}\sup_{m\in\mathcal{M}}\mathbb{E}_{m}\bigl[l_{n}(\theta,\delta)\bigr]\\
 & =\argmax_{\delta}\inf_{m}\left\{ \mathbb{E}_{m}\bigl[u_{n}(\theta,\delta)\bigr]:R_{\mathcal{Q}_{1}}(m)\leq K_{1},\;R_{\mathcal{Q}_{2}}(m)\leq K_{2}\right\} .
\end{align*}
As $R_{\mathcal{Q}_{a}}(\cdot)$ is convex, standard duality arguments
yield the more convenient variational characterization:
\begin{align*}
\delta_{n}^{*} & =\argmax_{\delta}V_{n}(\delta);\ \textrm{ where}\\
V_{n}(\delta) & :=\inf_{m}\left\{ \mathbb{E}_{m}\bigl[u_{n}(\theta,\delta)\bigr]+\lambda\left(\alpha\,R_{\mathcal{Q}_{1}}(m)+(1-\alpha)\,R_{\mathcal{Q}_{2}}(m)\right)\right\} ,
\end{align*}
for some $\alpha\in[0,1]$ and $\lambda\geq0$ that depend on $K_{1},K_{2}$.
In particular, $\alpha>1/2$ whenever $K_{2}\gg K_{1}$, as the decision
risk places greater weight on the likelihood that is less likely to
be misspecified. 

Let
\[
\bar{p}_{\alpha,\theta}(\bm{x})\equiv\bar{p}_{\alpha}(\bm{x}|\theta):=p_{1}^{\alpha}(\bm{x}|\theta)\cdot p_{2}^{1-\alpha}(\bm{x}|\theta)
\]
denote the unnormalized geometric mixture density that combines $p_{1}(\bm{x}|\theta)$
and $p_{2}(\bm{x}|\theta)$ with mixing weights $\alpha$ and $1-\alpha$,
respectively. Recalling the decomposition $m=\pi(\theta)\otimes m_{\theta}(\bm{x})$
and applying (\ref{eq:characterization_of_R_Q}) yields
\begin{align*}
V_{n}(\delta) & =\inf_{\pi(\theta)\otimes m_{\theta}(\bm{x})}\int\left\{ \mathbb{E}_{m_{\theta}(\bm{x})}\bigl[u_{n}(\theta,\delta)\bigr]+\lambda\left(\alpha\,\text{KL}(m_{\theta}\|p_{1,\theta})+(1-\alpha)\,\text{KL}(m_{\theta}\|p_{2,\theta})\right)\right\} d\pi(\theta)\\
 & =\inf_{\pi(\theta)\otimes m_{\theta}(\bm{x})}\,\int\left\{ \mathbb{E}_{m_{\theta}(\bm{x})}\left[u_{n}(\theta,\delta)\right]+\lambda\cdot\text{KL}\bigl(m_{\theta}\|\bar{p}_{\alpha,\theta}\bigr)\right\} d\pi(\theta),
\end{align*}
where the last step makes use of the fact that the weighted sum of
KL divergences can be expressed as a single KL divergence against
a geometric mixture. Defining $\mathcal{Q}_{\alpha}:=\bigl\{\pi(\theta)\otimes\bar{p}_{\alpha,\theta}(\bm{x}):\pi\in\Delta(\Theta)\bigr\}$
and applying (\ref{eq:characterization_of_R_Q}) again then gives
\begin{align*}
V_{n}(\delta) & =\inf_{m}\left\{ \mathbb{E}_{m}\bigl[u_{n}(\theta,\delta)\bigr]+\lambda\,R_{\mathcal{Q}_{\alpha}}(m)\right\} \\
 & =-\lambda\ln\left\{ \max_{\pi\in\Delta(\Theta)}\int\mathbb{E}_{\bar{p}_{\alpha}(\bm{x}|\theta)}\!\left[e^{l_{n}(\theta,\delta)/\lambda}\right]d\pi(\theta)\right\} ,
\end{align*}
where the last step follows by the same arguments as in Section \ref{sec:Characterizing-Optimal-Decisions}. 

The optimal decision therefore solves
\begin{align*}
\delta_{n}^{*} & =\argmin_{\delta}\,\max_{\pi\in\Delta(\Theta)}\int\mathbb{E}_{\bar{p}_{\alpha}(\bm{x}|\theta)}\!\left[e^{l_{n}(\theta,\delta)/\lambda}\right]d\pi(\theta).
\end{align*}
Optimal decisions under multiple, possibly misspecified, candidate
likelihoods are thus equivalent to minimax decisions with an exponentiated
loss function and a single unnormalized geometric mixture likelihood
$\bar{p}_{\alpha}(\bm{x}|\theta)$. While $\bar{p}_{\alpha}(\bm{x}|\theta)$
is not a proper density, we can alternatively characterize $\delta_{n}^{*}$
as 
\[
\delta_{n}^{*}=\argmin_{\delta}\,\max_{\pi\in\Delta(\Theta)}\int w(\theta)\mathbb{E}_{p_{\alpha}(\bm{x}|\theta)}\!\left[e^{l_{n}(\theta,\delta)/\lambda}\right]d\pi(\theta)
\]
where $w(\theta):=\int\bar{p}_{\alpha}(\bm{x}|\theta)d\nu(\bm{x})$
is the normalizing constant and $p_{\alpha}(\bm{x}|\theta):=\bar{p}_{\alpha}(\bm{x}|\theta)/w(\theta)$
is the normalized mixture likelihood.

The calculations above are exact. The local asymptotic analysis of
this setting is more delicate, however, owing to the presence of the
weight $w(\theta)$, and depends on whether the competing likelihoods
$p_{1}(\cdot|\theta),p_{2}(\cdot|\theta)$ are separated by a fixed
distance or allowed to converge. Hence, we defer the formal results
to a companion paper.

Consider the limiting case in which $\mathcal{Q}_{1}$ strictly nests
$\mathcal{Q}_{2}$ and $K_{1}=0$, so that the decision-maker has
no misspecification concerns about the larger nesting model $\mathcal{Q}_{1}\equiv\mathcal{M}_{1}$
while still harboring concerns about the smaller class of models $\mathcal{M}_{2}$.
The optimal decision then becomes
\begin{align*}
\delta_{n}^{*} & =\argmax_{\delta}\inf_{m\in\mathcal{M}_{1}}\left\{ \mathbb{E}_{m}\bigl[u_{n}(\theta,\delta)\bigr]:R_{\mathcal{Q}_{2}}(m)\leq K_{2}\right\} \\
 & =\argmax_{\delta}\inf_{m\in\mathcal{M}_{1}}\left\{ \mathbb{E}_{m}\bigl[u_{n}(\theta,\delta)\bigr]+\bar{\lambda}R_{\mathcal{Q}_{2}}(m)\right\} ,
\end{align*}
for some $\bar{\lambda}$ that depends on $K_{2}$. The decision criterion
$\inf_{m\in\mathcal{M}_{1}}\left\{ \mathbb{E}_{m}\bigl[u_{n}(\theta,\delta)\bigr]+\bar{\lambda}R_{\mathcal{Q}_{2}}(m)\right\} $
is equivalent to the \textit{constrained multiplier criterion} axiomatized
and studied by \citet{andrews2026misspecification}. The latter thus
arises as a special case of our model selection framework when there
are no misspecification concerns about the larger, nesting model.

\section{Conclusion}

In this article, we have introduced a framework for evaluating statistical
decisions under both prior ambiguity and likelihood misspecification.
Misspecification manifests as an exponential tilting of the loss function,
while ambiguity corresponds to a search for the least favorable prior.
We also develop a theory of local asymptotics under general misspecification,
achieved by localizing the priors around a reference parameter, and
use this theory to characterize optimal estimation and treatment-assignment
decisions. Remarkably, in both cases, optimal decisions coincide with
those under correct likelihood specification.

The proposed framework opens several avenues for further research.
While Section \ref{subsec:Model-averaging} provides an initial treatment
of model aggregation, a richer characterization would involve taking
the misspecification risks of competing likelihoods to converge to
each other, so that more nuanced tradeoffs between efficiency and
misspecification robustness may emerge. A further extension would
be to separate ambiguity concerns over the prior from those over candidate
likelihoods, employing smooth ambiguity aversion for the latter, as
discussed in Section \ref{subsec:Relaxing-caution}. This would bring
the framework closer to the literature on Bayesian model averaging. 

\bibliographystyle{econ-aea}
\bibliography{Ambiguity_and_misspecification}

@article{bordalo2020overreaction,
  title={{Overreaction in Macroeconomic Expectations}},
  author={Bordalo, Pedro and Gennaioli, Nicola and Ma, Yueran and Shleifer, Andrei},
  journal={American Economic Review},
  volume={110},
  number={9},
  pages={2748--2782},
  year={2020},
  publisher={American Economic Association 2014 Broadway, Suite 305, Nashville, TN 37203}
}

@article{andrews2026misspecification,
  title={{Misspecification-Averse Estimation}},
  author={Andrews, Isaiah and Li, Ricky and Shang, Yucheng},
  journal={arXiv preprint arXiv:2604.23176},
  year={2026}
}

@article{armstrong2021sensitivity,
  title={{Sensitivity Analysis Using Approximate Moment Condition Models}},
  author={Armstrong, Timothy B and Koles{\'a}r, Michal},
  journal={Quantitative Economics},
  volume={12},
  number={1},
  pages={77--108},
  year={2021},
  publisher={Wiley Online Library}
}

@article{hirano2025wald,
  title={{Wald's Statistical Decision Theory for Policy Analysis and Adaptive Experiments}},
  author={Hirano, Keisuke},
  year={2025}
}

@book{dontchev2009implicit,
  title={{Implicit Functions and Solution Mappings}},
  author={Dontchev, Asen L and Rockafellar, R Tyrrell},
  volume={543},
  year={2009},
  publisher={Springer}
}

@book{Ben-Tal:RobustOptimization,
  author = {Ben-Tal, A. and Ghaoui, L. El and Nemirovski, A.},
  publisher = {Princeton University Press},
  title = {Robust Optimization},
  year = {2009}
}

@article{masten2021salvaging,
  title={{Salvaging Falsified Instrumental Variable Models}},
  author={Masten, Matthew A and Poirier, Alexandre},
  journal={Econometrica},
  volume={89},
  number={3},
  pages={1449--1469},
  year={2021},
  publisher={Wiley Online Library}
}

@article{muller2024locally,
  title={{Locally Robust Semiparametrically Efficient Bayesian Inference}},
  author={M{\"u}ller, Ulrich K and Norets, ANDRIY},
  year={2024},
  journal={Working paper}
}

@article{armstrong2025misspecification,
  title={{Misspecification in Econometrics: A Selective Review}},
  author={Armstrong, Timothy B},
  year={2025}
}

@article{white1982maximum,
  title={{Maximum Likelihood Estimation of Misspecified Models}},
  author={White, Halbert},
  journal={Econometrica: Journal of the econometric society},
  pages={1--25},
  year={1982},
  publisher={JSTOR}
}

@book{manski2003partial,
  title={{Partial Identification of Probability Distributions}},
  author={Manski, Charles F},
  year={2003},
  publisher={Springer}
}

@article{kuhn2025distributionally,
  title={{Distributionally Robust Optimization}},
  author={Kuhn, Daniel and Shafiee, Soroosh and Wiesemann, Wolfram},
  journal={Acta Numerica},
  volume={34},
  pages={579--804},
  year={2025},
  publisher={Cambridge University Press}
}

@article{huber1964robust,
  title={{Robust Estimation of a Location Parameter}},
  author={Huber, Peter J},
  journal={Annals of Mathematical Statistics},
  volume={35},
  number={4},
  pages={73--101},
  year={1964}
}

@article{box1976science,
  title={{Science and Statistics}},
  author={Box, George EP},
  journal={Journal of the American Statistical Association},
  volume={71},
  number={356},
  pages={791--799},
  year={1976},
  publisher={Taylor \& Francis}
}

@book{hansen2011robustness,
  title={Robustness},
  author={Hansen, Lars Peter and Sargent, Thomas J},
  year={2011},
  publisher={Princeton university press}
}

@article{klibanoff2005smooth,
  title={{A Smooth Model of Decision Making under Ambiguity}},
  author={Klibanoff, Peter and Marinacci, Massimo and Mukerji, Sujoy},
  journal={Econometrica},
  volume={73},
  number={6},
  pages={1849--1892},
  year={2005},
  publisher={Wiley Online Library}
}

@book{ibragimov1981problem,
  title={{Statistical Estimation: Asymptotic Theory}},
  author={Ibragimov, IA and Hasminskii, RZ},
  year={1981},
  publisher={Springer}
}

@article{anscombe1963definition,
  title={{A Definition of Subjective Probability}},
  author={Anscombe, Francis J and Aumann, Robert J},
  journal={{The Annals of Mathematical Statistics}},
  volume={34},
  number={1},
  pages={199--205},
  year={1963},
  publisher={JSTOR}
}

@article{andrews2020informativeness,
  title={{On the Informativeness of Descriptive Statistics for Structural Estimates}},
  author={Andrews, Isaiah and Gentzkow, Matthew and Shapiro, Jesse M},
  journal={Econometrica},
  volume={88},
  number={6},
  pages={2231--2258},
  year={2020},
  publisher={Wiley Online Library}
}

@article{bonhomme2022minimizing,
  title={{Minimizing Sensitivity to Model Misspecification}},
  author={Bonhomme, St{\'e}phane and Weidner, Martin},
  journal={Quantitative Economics},
  volume={13},
  number={3},
  pages={907--954},
  year={2022},
  publisher={Wiley Online Library}
}

@article{christensen2023counterfactual,
  title={{Counterfactual Sensitivity and Robustness}},
  author={Christensen, Timothy and Connault, Benjamin},
  journal={Econometrica},
  volume={91},
  number={1},
  pages={263--298},
  year={2023},
  publisher={Wiley Online Library}
}

@article{ishihara2021evidence,
  title={{Evidence Aggregation for Treatment Choice}},
  author={Ishihara, Takuya and Kitagawa, Toru},
  journal={arXiv preprint arXiv:2108.06473},
  year={2021}
}

@article{yata2021optimal,
  title={{Optimal Decision Rules under Partial Identification}},
  author={Yata, Kohei},
  journal={arXiv preprint arXiv:2111.04926},
  year={2021}
}

@article{montiel2026decision,
  title={{Decision Theory for Treatment Choice Problems with Partial Identification}},
  author={Montiel Olea, Jos{\'e} Luis and Qiu, Chen and Stoye, J{\"o}rg},
  journal={Review of Economic Studies},
  pages={rdag015},
  year={2026},
  publisher={Oxford University Press}
}

@article{christensen2022optimal,
  title={{Optimal Discrete Decisions when Payoffs are Partially Identified}},
  author={Christensen, Timothy and Moon, Hyungsik Roger and Schorfheide, Frank},
  journal={arXiv preprint arXiv:2204.11748},
  year={2022}
}

@incollection{wald1950statistical,
  title={{Statistical Decision Functions}},
  author={Wald, Abraham},
  booktitle={Breakthroughs in Statistics: Foundations and Basic Theory},
  pages={342--357},
  year={1950},
  publisher={Springer}
}

@article{andrews2025purpose,
  title={{The Purpose of an Estimator is What it Does: Misspecification, Estimands, and Over-Identification}},
  author={Andrews, Isaiah and Chen, Jiafeng and Tecchio, Otavio},
  journal={arXiv preprint arXiv:2508.13076},
  year={2025}
}

@article{maccheroni2006ambiguity,
  title={{Ambiguity Aversion, Robustness, and the Variational Representation of Preferences}},
  author={Maccheroni, Fabio and Marinacci, Massimo and Rustichini, Aldo},
  journal={Econometrica},
  volume={74},
  number={6},
  pages={1447--1498},
  year={2006},
  publisher={Wiley Online Library}
}

@book{robert2007bayesian,
  title={{The Bayesian Choice: From Decision-Theoretic Foundations to Computational Implementation}},
  author={Robert, Christian P},
  year={2007},
  publisher={Springer}
}

@article{cerreia2026making,
  title={{Making Decisions under Model Misspecification}},
  author={Cerreia-Vioglio, Simone and Hansen, Lars Peter and Maccheroni, Fabio and Marinacci, Massimo},
  journal={Review of Economic Studies},
  volume={93},
  number={2},
  pages={892--925},
  year={2026},
  publisher={Oxford University Press UK}
}

@book{le1986asymptotic,
  title={{Asymptotic Methods in Statistical Decision Theory}},
  author={Le Cam, Lucien M},
  year={1986},
  publisher={Springer-Verlag}
}

@book{van1996weak,
  title={{Weak Convergence and Empirical Processes: With Applications to Statistics}},
  author={van der Vaart, Aad W and Wellner, Jon},
  year={1996},
  publisher={Springer Science \& Business Media}
}

@book{van2000asymptotic,
  title={{Asymptotic Statistics}},
  author={Van der Vaart, Aad W},
  year={2000},
  publisher={Cambridge university press}
}

@article{hirano2009asymptotics,
  title={{Asymptotics for Statistical Treatment Rules}},
  author={Hirano, Keisuke and Porter, Jack R},
  journal={Econometrica},
  volume={77},
  number={5},
  pages={1683--1701},
  year={2009},
  publisher={Wiley Online Library}
}

\appendix

\section{Proofs of Theorems \ref{thm:lower_bound} and \ref{thm:upper_bound}}

\subsection{Proof of Theorem \ref{thm:lower_bound}}

For estimation-loss, the theorem is a direct consequence of \citet[Proposition 8.11]{van2000asymptotic};
see also \citet[Theorem 3.11.5]{van1996weak}. 

We therefore focus on proving the theorem for treatment-assignment
loss. Let 
\begin{align*}
R_{n}(h,\delta) & :=\mathbb{E}_{n,h}\left[e^{l_{n}(\theta+h/\sqrt{n},\delta)/\lambda}\right]\textrm{ and }\\
R(h,\tilde{\delta}) & :=\mathbb{E}_{h}\left[e^{l(h,\tilde{\delta})/\lambda}\right]
\end{align*}
denote the frequentist-risks of $\delta$ and $\tilde{\delta}$ in
the finite-sample and limit experiments, respectively. As with the
proof of Proposition \ref{prop1}, we allow for randomized decisions.
For randomized decisions, $\delta,\tilde{\delta}_{n}$ should be understood
as the probability that the treatment is assigned, and the loss function
$l_{n}(\theta+h/\sqrt{n},\delta)/\lambda$ should be interpreted as
allowing for $\delta\in[0,1]$. 

We start by proving the following lemma:

\begin{lem} \label{lem:treatment_assignment_lower_bound}Suppose
Assumptions 1 and 2 hold. Let $\{\delta_{n}\}_{n}$ be a sequence
of treatment decisions with associated frequentist regret $R_{n}(h,\delta_{n})$.
Then, there exists a subsequence $\{n_{k}\}_{k}$ and a (possibly
randomized) treatment decision $\tilde{\delta}$ in the limit experiment
such that $\lim_{k\to\infty}R_{n_{k}}(h,\delta_{n_{k}})\ge R(h,\tilde{\delta})$
for each $h$. \end{lem} 
\begin{proof}
As $\delta_{n}\in[0,1]$ is uniformly bounded, it is tight. Since
$x_{n}\xrightarrow[P_{n,0}]{d}x\sim\mathcal{N}(0,I)$, it follows
that the joint $(\delta_{n},x_{n})$ is also tight. Hence, by Prohorov's
theorem, given any sequence $\{n_{}\}$, there exists a further sub-sequence
$\{n_{k}\}$ such that $(\delta_{n_{k}},x_{n_{k}})\xrightarrow[P_{n_{k},0}]{d}(\bar{\delta},x)$,
where $\bar{\delta}\in[0,1]$ is some tight limit of $\delta_{n_{k}}$.
Combined with (\ref{eq:LAN property}), this implies
\begin{align}
\left(\begin{array}{c}
\delta_{n_{k}}\\
\ln\frac{dP_{n_{k},\theta_{0}+h/\sqrt{n_{k}}}}{dP_{n_{k},\theta_{0}}}
\end{array}\right) & \xrightarrow[P_{n_{k},0}]{d}\left(\begin{array}{c}
\bar{\delta}\\
\ln V
\end{array}\right);\label{eq:pf:Thm-simple-regert:weak convergence}\\
V\sim & \exp\left\{ h^{\intercal}I_{0}^{1/2}x-\frac{1}{2}h^{\intercal}I_{0}h\right\} .\nonumber 
\end{align}
Observe that $V\ge0$ and $E[V]=1$. Therefore, by an application
of Le Cam's third lemma,
\begin{equation}
\delta_{n_{k}}\xrightarrow[P_{n_{k},h}]{d}\mathcal{L};\ \textrm{where }\mathcal{L}(B):=E\left[\mathbb{I}\{\bar{\delta}\in B\}V\right]\ \forall\ B\in\mathcal{B}(\mathbb{R}).\label{eq:pf:Thm-simple-regret-weak convergence 2}
\end{equation}
Together with (\ref{eq:treatment_loss_approx}), the above implies
\begin{equation}
e^{l_{n_{k}}(\theta_{0}+h/\sqrt{n_{k}},\delta_{n_{k}})/\lambda}\xrightarrow[P_{n_{k},h}]{d}e^{l(h,\cdot)/\lambda}\sharp\mathcal{L},\label{eq:pf:Them-simple-regret-weak convergence 3}
\end{equation}
where $f\sharp\mathcal{L}$ denotes the push-forward of the measure
$\mathcal{L}$ by $f$.

Define $\tilde{\delta}=E\left[\bar{\delta}|x\right]$. By construction,
$\tilde{\delta}$ is a valid treatment decision in the limit experiment. 

Now, $\exp\left\{ l_{n_{k}}(\theta_{0}+h/\sqrt{n_{k}},\delta_{n_{k}})/\lambda\right\} $
is a sequence of non-negative random variables. Hence, by an application
of Fatou's lemma for weak convergence on (\ref{eq:pf:Them-simple-regret-weak convergence 3}),
\begin{align*}
\liminf_{k\to\infty}R_{n_{k}}(h,\delta_{n_{k}}) & =\liminf_{k\to\infty}\mathbb{E}_{n_{k},h}\left[\exp\left\{ l_{n_{k}}(\theta_{0}+h/\sqrt{n_{k}},\delta_{n_{k}})/\lambda\right\} \right]\\
 & \ge E\left[\exp\left\{ l(h,\bar{\delta})/\lambda\right\} \exp\left\{ h^{\intercal}I_{0}^{1/2}x-\frac{1}{2}h^{\intercal}I_{0}h\right\} \right]\\
 & \ge E\left[\exp\left\{ l(h,\tilde{\delta})/\lambda\right\} \exp\left\{ h^{\intercal}I_{0}^{1/2}x-\frac{1}{2}h^{\intercal}I_{0}h\right\} \right]\\
 & =\mathbb{E}_{h}\left[\exp\left\{ l(h,\tilde{\delta})/\lambda\right\} \right]=R(h,\tilde{\delta}),
\end{align*}
where the second inequality follows from the conditional Jensen's
inequality after noting that $\exp\left\{ l(h,\delta)/\lambda\right\} $
is convex in $\delta$ for any $h$, and the penultimate equality
follows from standard change-of-measure arguments for Gaussian distributions. 
\end{proof}
Returning to the proof of Theorem \ref{thm:lower_bound}, let $\pi_{\Delta^{*}}$
denote the least-favorable prior, i.e., the symmetric two-point prior
supported on
\[
\left(-\frac{\Delta^{*}}{\sigma^{2}}I_{0}^{-1}\dot{\mu}_{0},\frac{\Delta^{*}}{\sigma^{2}}I_{0}^{-1}\dot{\mu}_{0}\right).
\]
Clearly, for all $M$ sufficiently large,
\begin{align*}
 & \liminf_{n\to\infty}\min_{\delta}\,\max_{\pi(h)\in\Delta_{M}(\mathcal{H})}\int\mathbb{E}_{n,h}\left[e^{l_{n}(\theta_{0}+h/\sqrt{n},\delta)/\lambda}\right]d\pi(h)\\
 & \ge\liminf_{n\to\infty}\min_{\delta}\int\mathbb{E}_{n,h}\left[e^{l_{n}(\theta_{0}+h/\sqrt{n},\delta)/\lambda}\right]d\pi_{\Delta^{*}}(h).
\end{align*}

Let $\{\delta_{n}\}_{n}$ denote any sequence along which the $\liminf_{n\to\infty}\min_{\delta}$
on the right hand side of the above expression is attained. By the
definition of $\pi_{\Delta^{*}}$, 
\[
\int\mathbb{E}_{n,h}\left[e^{l_{n}(\theta_{0}+h/\sqrt{n},\delta_{n})/\lambda}\right]d\pi_{\Delta^{*}}(h)=\frac{1}{2}R_{n}(-h^{*},\delta_{n})+\frac{1}{2}R_{n}(h^{*},\delta_{n}).
\]
Lemma \ref{lem:treatment_assignment_lower_bound} then implies the
existence of further subsequence, $\{n_{k}\}_{k}$, and a treatment
decision $\tilde{\delta}$ in the limit experiment such that 
\begin{align*}
\liminf_{k\to\infty}\left\{ \frac{1}{2}R_{n_{k}}(-h^{*},\delta_{n_{k}})+\frac{1}{2}R_{n_{k}}(h^{*},\delta_{n_{k}})\right\}  & \ge\frac{1}{2}R(-h^{*},\tilde{\delta})+\frac{1}{2}R(h^{*},\tilde{\delta})\\
 & =\int\mathbb{E}_{h}\left[e^{l(h,\tilde{\delta})/\lambda}\right]d\pi_{\Delta^{*}}(h)\\
 & \ge\int\mathbb{E}_{h}\left[e^{l(h,\tilde{\delta}^{*})/\lambda}\right]d\pi_{\Delta^{*}}(h)=V^{*},
\end{align*}
where the inequality follows from the fact that $\tilde{\delta}^{*}$
is the best-response to the least-favorable prior $\pi_{\Delta^{*}}$
in the limit experiment, and the final equality is just the definition
of minimax-risk. This completes the proof of the theorem. 

\subsection{Proof of Theorem \ref{thm:upper_bound}}

\subsubsection*{Estimation}

We start with the case of estimation. Consider any sequence $h_{n}\to h$.
By (\ref{eq:Convergence of score process}), (\ref{eq:LAN property})
and Assumption 3, 
\begin{align}
\left(\begin{array}{c}
\sqrt{n}(\hat{\theta}_{\textrm{mle}}-\theta_{0})\\
\ln\frac{dP_{n,\theta_{0}+h_{n}/\sqrt{n}}}{dP_{n,\theta_{0}}}
\end{array}\right) & \xrightarrow[P_{n,0}]{d}\left(\begin{array}{c}
I_{0}^{-1/2}x\\
h^{\intercal}I_{0}^{1/2}x-\frac{1}{2}h^{\intercal}I_{0}h
\end{array}\right),\ \textrm{where }x\sim\mathcal{N}(0,I).\label{eq:pf:Thm-simple-regert:weak convergence-2}
\end{align}
By Le Cam's third lemma, 
\begin{equation}
\sqrt{n}(\hat{\theta}_{\textrm{mle}}-\theta_{0})\xrightarrow[P_{n,h_{n}}]{d}\mathcal{N}(h,I_{0}^{-1}),\label{eq:pf:Thm2:0}
\end{equation}
and therefore, in view of Assumption 2, it follows that for each $h_{n}\to h$,
\begin{align*}
l_{n}\left(\theta_{0}+h_{n}/\sqrt{n},\hat{\delta}_{n}^{*}\right) & =\ell_{}\left(\sqrt{n}\left(\mu(\theta_{0}+h_{n}/\sqrt{n})-\mu(\hat{\theta}_{\textrm{mle}})\right)\right)\\
 & \xrightarrow[P_{n,h_{n}}]{d}\ell\left(\dot{\mu}_{0}^{\intercal}I_{0}^{-1/2}x\right),\ \textrm{where }x\sim\mathcal{N}(0,I).
\end{align*}
Since $l_{n}(\cdot)$ is uniformly bounded by Assumption 4, standard
properties of weak convergence and the fact $\tilde{\delta}^{*}(\cdot):=\dot{\mu}_{0}^{\intercal}I_{0}^{-1/2}\cdot$
imply 
\begin{equation}
\mathbb{E}_{n,h_{n}}\left[e^{l_{n}(\theta_{0}+h_{n}/\sqrt{n},\hat{\delta}_{n}^{*})/\lambda}\right]\to\mathbb{E}_{0}\left[e^{\ell\left(\dot{\mu}_{0}^{\intercal}I_{0}^{-1/2}x\right)/\lambda}\right]=\mathbb{E}_{h}\left[e^{\ell\left(\dot{\mu}_{0}^{\intercal}h-\tilde{\delta}^{*}\right)/\lambda}\right]\label{eq:pf:thm2:1}
\end{equation}
for every sequence $h_{n}\to h$. Define
\begin{align*}
f_{n}(h) & :=\mathbb{E}_{n,h_{}}\left[e^{l_{n}(\theta_{0}+h_{}/\sqrt{n},\hat{\delta}_{n}^{*})/\lambda}\right]\ \textrm{and}\\
f(h) & :=\mathbb{E}_{h}\left[e^{\ell_{}\left(\dot{\mu}_{0}^{\intercal}h-\tilde{\delta}^{*}\right)/\lambda}\right]=\mathbb{E}_{h}\left[e^{l_{}\left(h,\tilde{\delta}^{*}\right)/\lambda}\right].
\end{align*}
Equation (\ref{eq:pf:thm2:1}) implies continuous convergence of $f_{n}(\cdot)$
to $f(\cdot)$, i.e., $f_{n}(h_{n})\to f(h)$ for every $h_{n}\to h$.
But continuous convergence on compact sets implies uniform convergence,
so 
\begin{equation}
\sup_{\vert h\vert\le M}\vert f_{n}(h)-f(h)\vert\to0.\label{eq:pf:thm2:2}
\end{equation}

Now, consider a sequence of priors $\{\pi_{n}(h)\}_{n}$ along which
\[
\limsup_{n\to\infty}\max_{\pi(h)\in\Delta_{M}(\mathcal{H})}\int\mathbb{E}_{n,h}\left[e^{l_{n}(\theta_{0}+h/\sqrt{n},\hat{\delta}_{n}^{*})/\lambda}\right]d\pi(h)
\]
is attained. Since $\Delta_{M}(\mathcal{H})$ represents the space
of compactly supported priors, it is compact under the metric of weak
convergence. Hence, there exists a further sub-sequence $\{\pi_{n_{j}}(h)\}_{j}$
such that $\pi_{n_{j}}(h)$ converges weakly to some $\tilde{\pi}(h)\in\Delta_{M}(\mathcal{H})$.
Furthermore, as $e^{\ell_{}(\cdot)/\lambda}$ is uniformly bounded,
so is $f(\cdot)$. Combining these observations with (\ref{eq:pf:thm2:2}),
standard properties of weak convergence yield
\begin{align*}
 & \limsup_{n\to\infty}\max_{\pi(h)\in\Delta_{M}(\mathcal{H})}\int\mathbb{E}_{n,h}\left[e^{l_{n}(\theta_{0}+h/\sqrt{n},\hat{\delta}_{n}^{*})/\lambda}\right]d\pi(h)\\
 & =\lim_{j\to\infty}\int f_{n_{j}}(h)d\pi_{n_{j}}(h)\\
 & \le\lim_{j\to\infty}\sup_{\vert h\vert\le M}\vert f_{n_{j}}(h)-f(h)\vert+\lim_{j\to\infty}\left|\int f(h)d\pi_{n_{j}}(h)-\int f(h)d\tilde{\pi}(h)\right|+\int f(h)d\tilde{\pi}(h)\\
 & =\int f(h)d\tilde{\pi}(h)\le\max_{\pi(h)\in\Delta(\mathcal{H})}\int\mathbb{E}_{h}\left[e^{l\left(h,\tilde{\delta}^{*}\right)/\lambda}\right]d\pi(h):=V^{*}.
\end{align*}
Since the above is valid for any $M<\infty$, this completes the proof
of Theorem \ref{thm:upper_bound} for the estimation problem. 

\subsubsection*{Treatment assignment}

We now turn to the case of treatment assignment. Recall that here
we choose $\theta_{0}$ so that $\mu(\theta_{0})=0$. Then, (\ref{eq:pf:Thm2:0})
and Assumption 2 imply 
\begin{align*}
\mathbf{1}\left\{ \mu(\hat{\theta}_{\textrm{mle}})\ge0\right\}  & =\mathbf{1}\left\{ \sqrt{n}\left(\mu(\hat{\theta}_{\textrm{mle}})-\mu(\theta_{0})\right)\ge0\right\} \\
 & \xrightarrow[P_{n,h_{n}}]{d}\mathbf{1}\left\{ \dot{\mu}_{0}^{\intercal}I_{0}^{-1/2}x\ge0\right\} ,\ \textrm{where }x\sim\mathcal{N}(I_{0}^{-1/2}h,I).
\end{align*}
Hence, by standard properties of weak convergence, for each $h_{n}\to h$,
\begin{equation}
\mathbb{E}_{n,h_{n}}\left[\hat{\delta}_{n}^{*}\right]=P_{n,h_{n}}\left(\mu(\hat{\theta}_{\textrm{mle}})\ge0\right)\to P_{h}\left(\dot{\mu}_{0}^{\intercal}I_{0}^{-1/2}x\ge0\right)=\mathbb{E}_{h}\left[\tilde{\delta}^{*}\right].\label{eq:pf:Thm2:3}
\end{equation}

Note that under the treatment assignment loss, 
\[
\mathbb{E}_{n,h_{n}}\left[e^{l_{n}(\theta_{0}+h_{n}/\sqrt{n},\hat{\delta}_{n}^{*})/\lambda}\right]=e^{l_{n}(\theta_{0}+h_{n}/\sqrt{n},1)/\lambda}\mathbb{E}_{n,h_{n}}\left[\hat{\delta}_{n}^{*}\right]+e^{l_{n}(\theta_{0}+h_{n}/\sqrt{n},0)/\lambda}\mathbb{E}_{n,h_{n}}\left[1-\hat{\delta}_{n}^{*}\right].
\]
Now, (\ref{eq:treatment_loss_approx}) yields
\[
l_{n}(\theta_{0}+h_{n}/\sqrt{n},a)\to l(h,a),\ \textrm{for each }a\in\{0,1\}\ \textrm{and }h_{n}\to h.
\]
Combined with (\ref{eq:pf:Thm2:3}), this proves 
\begin{align*}
\mathbb{E}_{n,h_{n}}\left[e^{l_{n}(\theta_{0}+h_{n}/\sqrt{n},\hat{\delta}_{n}^{*})/\lambda}\right] & \to e^{l(h,1)/\lambda}\mathbb{E}_{h}\left[\tilde{\delta}^{*}\right]+e^{l(h,0)/\lambda}\mathbb{E}_{h}\left[1-\tilde{\delta}^{*}\right]\\
 & =\mathbb{E}_{h}\left[e^{l(h,\tilde{\delta}^{*})/\lambda}\right],\ \textrm{for each }h_{n}\to h.
\end{align*}

As before, define
\begin{align*}
f_{n}(h) & :=\mathbb{E}_{n,h_{}}\left[e^{l_{n}(\theta_{0}+h_{}/\sqrt{n},\hat{\delta}_{n}^{*})/\lambda}\right]\ \textrm{and}\\
f(h) & :=\mathbb{E}_{h}\left[e^{l_{}\left(h,\tilde{\delta}^{*}\right)/\lambda}\right],
\end{align*}
and observe that $f(h)$ is bounded under treatment-assignment loss
whenever $\vert h\vert\le M$. Consequently, the remainder of the
proof follows by applying the same arguments as in the case of estimation.

\newpage{}

\section*{\textbf{SUPPLEMENTARY APPENDIX (Online Appendix) }}

\section{Proofs Of the Remaining Results}

\subsection{Proof of Proposition \ref{prop1}}

The claim follows if we show that the decision maker's choice of $\tilde{\delta}^{*}=\mathbf{1}\{\dot{\mu}_{0}^{\intercal}I_{0}^{-1/2}x\geq0\}$
and nature's choice of a two point prior supported on $(-h^{*},h^{*})$
constitute a Nash equilibrium. Note that the proof applies even if
we allow for randomized decisions, where $\tilde{\delta}$ is interpreted
as the probability of treatment assignment and the loss function $l_{n}(\theta+h/\sqrt{n},\delta)/\lambda$
is interpreted as allowing for $\delta\in[0,1]$. The optimal decision,
however, does not involve randomization.

Let $\sigma=\sqrt{\dot{\mu}_{0}^{\intercal}I_{0}^{-1}\dot{\mu}_{0}}$.
Consider the best response of the decision-maker to any symmetric
two-point prior, $\pi_{\Delta}$, supported on
\[
\left(-\frac{\Delta^ {}}{\sigma^{2}}I_{0}^{-1}\dot{\mu}_{0},\frac{\Delta}{\sigma^{2}}I_{0}^{-1}\dot{\mu}_{0}\right),
\]
where $\Delta\in[0,\infty)$ is arbitrary. Let $\mathbb{P}_{h|x}$
denote the posterior probability over $h$ given this prior. Some
straightforward algebra shows that
\[
q(x):=\mathbb{P}_{h|x}\left(h=-\frac{\Delta}{\sigma^{2}}I_{0}^{-1}\dot{\mu}_{0}\right)=\frac{1}{1+\exp\left\{ 2\Delta\left(\dot{\mu}_{0}^{\intercal}I_{0}^{-1/2}x\right)/\sigma^{2}\right\} }.
\]
The Bayes-optimal response to $\pi_{\Delta}$ is therefore 
\begin{align*}
\tilde{\delta}_{\Delta} & =\mathbf{1}\left\{ \mathbb{E}_{h|x}\left[\exp\left\{ -\frac{\dot{\mu}_{0}^{\intercal}h}{\lambda}\mathbf{1}\{\dot{\mu}_{0}^{\intercal}h<0\}\right\} \right]\le\mathbb{E}_{h|x}\left[\exp\left\{ \frac{\dot{\mu}_{0}^{\intercal}h}{\lambda}\mathbf{1}\{\dot{\mu}_{0}^{\intercal}h\ge0\}\right\} \right]\right\} \\
 & =\mathbf{1}\left\{ e^{\frac{\Delta}{2\lambda}}q(x)+(1-q(x))\le q(x)+e^{\frac{\Delta}{2\lambda}}(1-q(x))\right\} \\
 & =\mathbf{1}\left\{ q(x)\le1/2\right\} =\mathbf{1}\left\{ \dot{\mu}_{0}^{\intercal}I_{0}^{-1/2}x\ge0\right\} .
\end{align*}
Hence, $\tilde{\delta}^{*}=\mathbf{1}\{\dot{\mu}_{0}^{\intercal}I_{0}^{-1/2}x\geq0\}$
is a Bayes optimal response to $\pi_{\Delta}$ for any $\Delta>0$.

We now consider nature's best response to $\tilde{\delta}^{*}$. Observe
that for any $\dot{\mu}_{0}^{\intercal}h\ge0$, 
\begin{align*}
\mathbb{E}_{h}\left[e^{l(h,\tilde{\delta}^{*})/\lambda}\right] & =\mathbb{E}_{h}\left[e^{\frac{\dot{\mu}_{0}^{\intercal}h}{\lambda}\mathbf{1}\{\dot{\mu}_{0}^{\intercal}I_{0}^{-1/2}x\le0\}}\right]\\
 & =\mathbb{E}_{0}\left[e^{\frac{\dot{\mu}_{0}^{\intercal}h}{\lambda}\mathbf{1}\{x\le-\dot{\mu}_{0}^{\intercal}h/\sigma\}}\right]=e^{\frac{\vert\dot{\mu}_{0}^{\intercal}h\vert}{\lambda}}\Phi(-\vert\dot{\mu}_{0}^{\intercal}h\vert/\sigma)+1-\Phi(-\vert\dot{\mu}_{0}^{\intercal}h\vert/\sigma).
\end{align*}
Similarly, for any $\dot{\mu}_{0}^{\intercal}h<0$, 
\begin{align*}
\mathbb{E}_{h}\left[e^{l(h,\tilde{\delta}^{*})/\lambda}\right] & =\mathbb{E}_{0}\left[e^{-\frac{\dot{\mu}_{0}^{\intercal}h}{\lambda}\mathbf{1}\{x\ge-\dot{\mu}_{0}^{\intercal}h/\sigma\}}\right].\\
 & =e^{\frac{\vert\dot{\mu}_{0}^{\intercal}h\vert}{\lambda}}\Phi(-\vert\dot{\mu}_{0}^{\intercal}h\vert/\sigma)+1-\Phi(-\vert\dot{\mu}_{0}^{\intercal}h\vert/\sigma).
\end{align*}
Thus, the expected loss is constant in $\vert\dot{\mu}_{0}^{\intercal}h\vert$,
implying that nature is indifferent between choosing $-h$ and $h$
for any $h$. Nature's best response to $\tilde{\delta}^{*}$ is therefore
any prior supported on $\left\{ h:\vert\dot{\mu}_{0}^{\intercal}h\vert=\Delta^{*}\right\} $,
where 
\[
\Delta^{*}:=\arg\max_{\Delta\ge0}\left\{ \left(e^{\frac{\Delta}{\lambda}}-1\right)\Phi(-\Delta/\sigma)\right\} .
\]
It is straightforward to verify that $\Delta^{*}$ exists and is unique.
Thus, the symmetric two-point prior, $\pi_{\Delta^{*}}$, supported
on $(-h^{*},h^{*})$ is a best response to $\tilde{\delta}^{*}$ . 

\subsection{Proof of Theorem \ref{thm:lower_bound-nonparametric}}

For estimation-loss, the theorem is a direct consequence of \citet[Theorem 25.21]{van2000asymptotic}.
We therefore focus on proving the theorem for treatment-assignment
loss. 

Recall that, via the Hilbert space isometry, we can construct an orthonormal
basis $(\psi/\sigma,\phi_{1},\phi_{2},\dots)$ such that each $h\in T(P_{0})$
can be identified with square integrable sequence of the form $(\mu/\sigma,\gamma_{1},\gamma_{2},\dots)\in l_{2}$,
where $\mu=\left\langle \psi,h\right\rangle $ and $\gamma_{k}=\left\langle \phi_{k},h\right\rangle $
for all $k\neq0$. Consider the class of priors $\tilde{\Pi}$ on
$l_{2}$ that assign a prior $\rho(\mu)\in\Delta_{M}(\mathbb{R})$
to $\mu$, supported on $\{\mu:\vert\mu\vert\le M\}$, while placing
a point-mass at the origin $(\gamma_{1},\gamma_{2},\dots)=(0,0,\dots)$
for the remaining components. Any $\pi\in\tilde{\Pi}$ is then equivalent
to a probability distribution over $T(P_{0})$ such that $h$ takes
the form $(\mu/\sigma^{2})\psi$, where $\mu\sim\rho(\cdot)$. 

Clearly, for $\bar{M}:=M/\sigma^{2}$, 
\begin{align*}
 & \liminf_{n\to\infty}\min_{\delta}\,\max_{\pi(h)\in\Delta(K_{M})}\int\mathbb{E}_{n,h}\left[e^{l_{n}(h,\delta)/\lambda}\right]d\pi(h)\\
 & \ge\liminf_{n\to\infty}\min_{\delta}\,\max_{\rho(\mu)\in\Delta_{\bar{M}}(\mathbb{R})}\int\mathbb{E}_{n,h}\left[\exp\left\{ l_{n}\left(\frac{\mu}{\sigma^{2}}\psi,\delta\right)/\lambda\right\} \right]d\rho(\mu).
\end{align*}

Now, consider sub-models of the form $P_{1/\sqrt{n},(\mu/\sigma)\psi/\sigma}$
for $\mu\in\mathbb{R}$. By (\ref{eq:SLAN nonparametric setting}),
\begin{equation}
\sum_{i=1}^{n}\ln\frac{P_{1/\sqrt{n},(\mu/\sigma)\psi/\sigma}}{dP_{0}}(Y_{i})=\frac{\mu}{\sigma\sqrt{n}}\sum_{i=1}^{n}\frac{\psi}{\sigma}(Y_{i})-\frac{\mu^{2}}{2\sigma^{2}}+o_{P_{1/\sqrt{n},0}}(1).\label{eq:LAN:parametric sub-models}
\end{equation}
Comparing with (\ref{eq:LAN property}), we observe that the family
\[
\left\{ P_{1/\sqrt{n},(\mu/\sigma)\psi/\sigma}:\mu\in\mathbb{R}\right\} 
\]
is equivalent to a parametric model with (normalized) score $\psi(\cdot)/\sigma$
and local parameter $\mu$ (observe that $\mathbb{E}_{P_{0}}[(\psi/\sigma)^{2}]=1$).
Furthermore, the second part of Assumption 5 implies
\begin{equation}
l_{n}\left(\frac{\mu}{\sigma^{2}}\psi,a\right)\to\mu\,\mathbf{1}\{\mu\geq0\}-\mu\,a,\ \textrm{uniformly over }a\in[0,1]\textrm{ and bounded }\mu.\label{eq:treatment_loss_approx-1}
\end{equation}
Consequently, we can apply the same arguments as in the proof of Theorem
\ref{thm:lower_bound} to show that for all $M$ sufficiently large
\begin{align*}
 & \liminf_{n\to\infty}\min_{\delta}\,\max_{\rho(\mu)\in\Delta_{\bar{M}}(\mathbb{R})}\int\mathbb{E}_{n,h}\left[\exp\left\{ l_{n}\left(\frac{\mu}{\sigma^{2}}\psi,\delta\right)/\lambda\right\} \right]d\rho(\mu)\\
 & \ge\min_{\tilde{\delta}}\,\max_{\rho(\mu)\in\Delta(\mathbb{R})}\int\mathbb{E}_{\mu}\left[e^{\left\{ \mu\,\mathbf{1}\{\mu\geq0\}-\mu\,\tilde{\delta}\right\} /\lambda}\right]d\rho(\mu).
\end{align*}
But by (\ref{eq:minimax_value_limit_experiment_nonparametrics}),
the term on the right is just the definition of $V^{*}$. 

\subsection{Proof of Theorem \ref{thm:upper_bound-non-parametric}}

\subsubsection*{Estimation}

We start with the case of estimation. Denote $\mu_{0}=\mu(P_{0})$.
By Assumption 6 and the central limit theorem, 
\begin{equation}
\sqrt{n}(\hat{\mu}_{n}-\mu_{0})\xrightarrow[P_{n,0}]{d}\mathcal{N}(0,\sigma^{2}).\label{eq:pf:thm4:1}
\end{equation}

Consider any sequence $h_{n}\to h$, where convergence is in terms
of the $L^{2}(P_{0})$ norm. Recall that any $h\in T(P_{0})$ admits
the orthogonal decomposition $h=(\mu/\sigma)\psi/\sigma+\tilde{h}$,
where $\mu:=\left\langle \psi,h\right\rangle $ and $\tilde{h}$ is
orthogonal to $\psi$. Applying the central limit theorem again, 
\begin{equation}
\left(\begin{array}{c}
\frac{1}{\sqrt{n}}\sum_{i=1}^{n}\psi(Y_{i})/\sigma\\
\frac{1}{\sqrt{n}}\sum_{i=1}^{n}\tilde{h}(Y_{i})
\end{array}\right)\xrightarrow[P_{n,0}]{d}\left(\begin{array}{c}
x\\
Z
\end{array}\right)\sim\mathcal{N}\left(\left(\begin{array}{c}
0\\
0
\end{array}\right),\left(\begin{array}{cc}
1 & 0\\
0 & \left\Vert \tilde{h}\right\Vert ^{2}
\end{array}\right)\right).\label{eq:pf:thm4:2}
\end{equation}
Combining (\ref{eq:pf:thm4:1}), (\ref{eq:pf:thm4:2}), (\ref{eq:SLAN nonparametric setting})
and the fact $\left\Vert h\right\Vert ^{2}=(\mu^{2}/\sigma^{2})+\left\Vert \tilde{h}\right\Vert ^{2}$,
we obtain
\begin{align}
\left(\begin{array}{c}
\sqrt{n}(\hat{\mu}_{n}-\mu_{0})\\
\ln\frac{dP_{n,h_{n}}}{dP_{n,0}}
\end{array}\right) & \xrightarrow[P_{n,0}]{d}\left(\begin{array}{c}
\sigma^ {}x\\
\ln V
\end{array}\right),\ \textrm{where}\label{eq:pf:Thm-4:weak convergence}\\
V & \sim\exp\left\{ \left(\frac{\mu}{\sigma}x-\frac{\mu^{2}}{2\sigma^{2}}\right)+\left(Z-\frac{1}{2}\left\Vert \tilde{h}\right\Vert ^{2}\right)\right\} .\nonumber 
\end{align}

Observe that $V\ge0$ and $E[V]=1$. Therefore, by an application
of Le Cam's third lemma,
\begin{equation}
\sqrt{n}(\hat{\mu}_{n}-\mu_{0})\xrightarrow[P_{n,h_{n}}]{d}\mathcal{L};\ \textrm{where }\mathcal{L}(B):=E\left[\mathbb{I}\{\sigma^ {}x\in B\}V\right]\ \forall\ B\in\mathcal{B}(\mathbb{R}).\label{eq:pf:Thm4:3}
\end{equation}
In other words, for every bounded and continuous $g(\cdot)$, 
\begin{align*}
\mathbb{E}_{n,h_{n}}\left[g\left(\sqrt{n}(\hat{\mu}_{n}-\mu_{0})\right)\right] & \to E\left[g(\sigma^ {}x)\exp\left\{ \left(\frac{\mu}{\sigma}x-\frac{\mu^{2}}{2\sigma^{2}}\right)+\left(Z-\frac{1}{2}\left\Vert \tilde{h}\right\Vert ^{2}\right)\right\} \right]\\
 & =E\left[g(\sigma x)\exp\left\{ \left(\frac{\mu}{\sigma}x-\frac{\mu^{2}}{2\sigma^{2}}\right)\right\} \right],
\end{align*}
where the equality is due to the fact $x$ and $Z$ are independent,
and $Z\sim\mathcal{N}(0,\left\Vert \tilde{h}\right\Vert ^{2})$. Hence,
applying the portmanteau lemma and standard change of measure arguments
yields 
\begin{equation}
\sqrt{n}(\hat{\mu}_{n}-\mu_{0})\xrightarrow[P_{n,h_{n}}]{d}\mathcal{N}\left(\left\langle \psi,h\right\rangle ,\sigma^{2}\right),\label{eq:pf:thm4:5}
\end{equation}
for every $h_{n}\to h$. 

Observe that, by the second part of Assumption 5, 
\begin{equation}
\sqrt{n}\left(\mu(P_{1/\sqrt{n},h_{n}})-\mu_{0}\right)=\left\langle \psi,h\right\rangle +o_{Pn,0}(1)\label{eq:pf:thm4:6}
\end{equation}
for every $h_{n}\to h$. Combining (\ref{eq:pf:thm4:5}), (\ref{eq:pf:thm4:6})
and the fact $\ell(\cdot)$ is bounded by Assumption 4, we obtain
\begin{align*}
l_{n}\left(h_{n},\hat{\delta}_{n}^{*}\right) & =\ell\left(\sqrt{n}\left(\mu(P_{1/\sqrt{n},h_{n}})-\mu_{0}\right)+\sqrt{n}\left(\mu_{0}-\hat{\mu}_{n}\right)\right)\\
 & \xrightarrow[P_{n,h_{n}}]{d}\ell\left(\sigma x\right),\ \textrm{where }x\sim\mathcal{N}(0,1).
\end{align*}
Since $l_{n}(\cdot)$ is uniformly bounded, standard properties of
weak convergence imply that for every $h_{n}\to h$, 
\begin{equation}
\mathbb{E}_{n,h_{n}}\left[e^{l_{n}\left(h_{n},\hat{\delta}_{n}^{*}\right)/\lambda}\right]\to\mathbb{E}_{0}\left[e^{\ell\left(\sigma x\right)/\lambda}\right]=\mathbb{E}_{\left\langle \psi,h\right\rangle }\left[e^{\ell\left(\left\langle \psi,h\right\rangle -\tilde{\delta}^{*}\right)/\lambda}\right],\label{eq:pf:thm4:7}
\end{equation}
where for any $\mu\in\mathbb{R}$, $\mathbb{E}_{\mu}[\cdot]$ represents
the expectation under the limit experiment described in Section \ref{subsec:Lower-bounds-nonparametrics},
and $\tilde{\delta}^{*}:=\sigma x$---with $x\sim\mathcal{N}(\mu/\sigma,1)$
under $P_{\mu}$ ---is the optimal decision rule under that limit
experiment. 

Define
\begin{align*}
f_{n}(h) & :=\mathbb{E}_{n,h_{}}\left[e^{l_{n}(h,\hat{\delta}_{n}^{*})/\lambda}\right]\ \textrm{and}\\
f(h) & :=\mathbb{E}_{\left\langle \psi,h\right\rangle }\left[e^{\ell\left(\left\langle \psi,h\right\rangle -\tilde{\delta}^{*}\right)/\lambda}\right].
\end{align*}
Equation (\ref{eq:pf:thm4:7}) implies continuous convergence of $f_{n}(\cdot)$
to $f(\cdot)$ as functions from $l_{2}$ to $\mathbb{R}$ , i.e.,
$f_{n}(h_{n})\to f(h)$ for every $h_{n}\to h$. But continuous convergence
on compact sets implies uniform convergence, so 
\begin{equation}
\sup_{h\in K_{M}}\vert f_{n}(h)-f(h)\vert\to0.\label{eq:pf:thm4:8}
\end{equation}

Now, consider a sequence of priors $\{\pi_{n}(h)\}_{n}$ along which
\[
\limsup_{n\to\infty}\max_{\pi(h)\in\Delta(K_{M})}\int\mathbb{E}_{n,h}\left[e^{l_{n}(h,\hat{\delta}_{n}^{*})/\lambda}\right]d\pi(h)
\]
is attained. Since $K_{M}\subset l_{2}$ is a compact set, $\Delta(K_{M})$
is compact under the metric of weak convergence. Hence, there exists
a further sub-sequence $\{\pi_{n_{j}}(h)\}_{j}$ such that $\pi_{n_{j}}(h)$
converges weakly to some $\tilde{\pi}(h)\in\Delta(K_{M})$. Furthermore,
since $e^{\ell_{}(\cdot)/\lambda}$ is uniformly bounded, so is $f(\cdot)$.
Combined with (\ref{eq:pf:thm4:8}), standard properties of weak convergence
then imply 
\begin{align*}
 & \limsup_{n\to\infty}\max_{\pi(h)\in\Delta(K_{M})}\int\mathbb{E}_{n,h}\left[e^{l_{n}(h,\hat{\delta}_{n}^{*})/\lambda}\right]d\pi(h)\\
 & =\lim_{j\to\infty}\int f_{n_{j}}(h)d\pi_{n_{j}}(h)\\
 & \le\lim_{j\to\infty}\sup_{h\in K_{M}}\vert f_{n_{j}}(h)-f(h)\vert+\lim_{j\to\infty}\left|\int f(h)d\pi_{n_{j}}(h)-\int f(h)d\tilde{\pi}(h)\right|+\int f(h)d\tilde{\pi}(h)\\
 & =\int f(h)d\tilde{\pi}(h)=\int\mathbb{E}_{\left\langle \psi,h\right\rangle }\left[e^{\ell\left(\left\langle \psi,h\right\rangle -\tilde{\delta}^{*}\right)/\lambda}\right]d\tilde{\pi}(h)\\
 & \le\max_{\rho(\mu)\in\Delta(\mathbb{R})}\int\mathbb{E}_{\mu}\left[e^{\ell\left(\mu-\tilde{\delta}^{*}\right)/\lambda}\right]d\rho(\mu):=V^{*}.
\end{align*}
Since the above is valid for any $M<\infty$, this completes the proof
of Theorem \ref{thm:upper_bound-non-parametric} for the estimation
problem. 

\subsubsection*{Treatment assignment}

We now turn to the case of treatment assignment. Recall that here
we choose $P_{0}$ so that $\mu_{0}:=\mu(P_{0})=0$. Then, (\ref{eq:pf:thm4:5}),
(\ref{eq:pf:thm4:6}) and the second part of Assumption 5 imply 
\begin{align*}
\mathbf{1}\left\{ \hat{\mu}_{n}\ge0\right\}  & =\mathbf{1}\left\{ \sqrt{n}\left(\hat{\mu}_{n}-\mu_{0}\right)\ge0\right\} \\
 & \xrightarrow[P_{n,h_{n}}]{d}\mathbf{1}\left\{ \sigma x\ge0\right\} ,\ \textrm{where }x\sim\mathcal{N}(\left\langle \psi,h\right\rangle /\sigma,1).
\end{align*}
Hence, by standard properties of weak convergence, for each $h_{n}\to h$,
we have
\begin{equation}
\mathbb{E}_{n,h_{n}}\left[\hat{\delta}_{n}^{*}\right]=P_{n,h_{n}}\left(\hat{\mu}_{n}\ge0\right)\to P_{\left\langle \psi,h\right\rangle }\left(\sigma x\ge0\right)=\mathbb{E}_{\left\langle \psi,h\right\rangle }\left[\tilde{\delta}^{*}\right],\label{eq:pf:Thm2:4-8}
\end{equation}
where for any $\mu\in\mathbb{R}$, $P_{\mu}$ and $\mathbb{E}_{\mu}[\cdot]$
represent the probabilities and expectations under the limit experiment
described in Section \ref{subsec:Lower-bounds-nonparametrics}, and
$\tilde{\delta}^{*}:=\mathbf{1}\left\{ \sigma x\ge0\right\} $ ---
with $x\sim\mathcal{N}(\mu/\sigma,1)$ under $P_{\mu}$ --- is the
optimal decision rule under that limit experiment. 

Note that under the treatment assignment loss, 
\[
\mathbb{E}_{n,h_{n}}\left[e^{l_{n}(h_{n},\hat{\delta}_{n}^{*})/\lambda}\right]=e^{l_{n}(h_{n},1)/\lambda}\mathbb{E}_{n,h_{n}}\left[\hat{\delta}_{n}^{*}\right]+e^{l_{n}(h_{n},0)/\lambda}\mathbb{E}_{n,h_{n}}\left[1-\hat{\delta}_{n}^{*}\right].
\]
Now, (\ref{eq:treatment_loss_approx}) implies
\[
l_{n}(h_{n},a)\to l\left(\left\langle \psi,h\right\rangle ,a\right),\ \textrm{for each }a\in\{0,1\}\ \textrm{and }h_{n}\to h,
\]
where $l(\cdot,\cdot)$ denotes the treatment assignment loss under
the limit experiment, as defined in (\ref{eq:minimax_value_limit_experiment_nonparametrics}).
Combined with (\ref{eq:pf:Thm2:4-8}), this proves 
\begin{align*}
\mathbb{E}_{n,h_{n}}\left[e^{l_{n}(h_{n},\hat{\delta}_{n}^{*})/\lambda}\right] & \to e^{l\left(\left\langle \psi,h\right\rangle ,1\right)/\lambda}\mathbb{E}_{\left\langle \psi,h\right\rangle }\left[\tilde{\delta}^{*}\right]+e^{l\left(\left\langle \psi,h\right\rangle ,0\right)}\mathbb{E}_{\left\langle \psi,h\right\rangle }\left[1-\tilde{\delta}^{*}\right]\\
 & =\mathbb{E}_{\left\langle \psi,h\right\rangle }\left[e^{l\left(\left\langle \psi,h\right\rangle ,\tilde{\delta}^{*}\right)/\lambda}\right],\ \textrm{for each }h_{n}\to h.
\end{align*}

As before, define
\begin{align*}
f_{n}(h) & :=\mathbb{E}_{n,h_{}}\left[e^{l_{n}(h,\hat{\delta}_{n}^{*})/\lambda}\right]\ \textrm{and}\\
f(h) & :=\mathbb{E}_{\left\langle \psi,h\right\rangle }\left[e^{l\left(\left\langle \psi,h\right\rangle ,\tilde{\delta}^{*}\right)/\lambda}\right],
\end{align*}
and observe that $f(h)$ is bounded under treatment-assignment loss
whenever $h\in K_{M}$ (which implies $\left\langle \psi,h\right\rangle \le M$).
Consequently, the remainder of the proof follows by applying the same
arguments as in the case of estimation.

\section{Local Asymptotics with Global Priors\label{sec:Local-Asymptotics-with-non-local-prior}}

We can allow unrestricted prior and reference parameter choice by
employing uniform versions of local asymptotic normality and Assumption
2. In what follows, the parameter space, $\Theta$, is assumed to
be a compact set. Let $P_{n,\theta}$ represent the joint probability
measure over the iid $Y_{1},\dots,Y_{n}$ when each $Y_{i}\sim P_{\theta}$,
and let $\mathbb{E}_{n,\theta}[\cdot]$ denote the corresponding expectation.
Also, take $\breve{\Theta}$ to be an open set covering $\Theta$. 

\begin{asmA1} The class $\{P_{\theta}^ {}:\theta\in\breve{\Theta}\}$
satisfies a uniform LAN property, i.e., there exists a score function
$\psi_{\theta}(\cdot)$ and information matrix $I_{\theta}:=\mathbb{E}_{\theta}[\psi_{\theta}\psi_{\theta}^{\intercal}]$
such that for each $\theta_{n}\to\theta\in\Theta$ and $h_{n}\to h\in\mathbb{R}^{d}$,
\[
\ln\frac{dP_{n,\theta_{n}+h_{n}/\sqrt{n}}}{dP_{n,\theta_{n}}}=h^{\intercal}I_{\theta_{n}}^{1/2}x_{n,\theta_{n}}-\frac{1}{2}h^{\intercal}I_{\theta_{n}}h+o_{P_{n,\theta_{n}}}(1),
\]
where 
\[
x_{n,\theta_{n}}:=\frac{I_{\theta_{n}}^{-1/2}}{\sqrt{n}}\sum_{i=1}^{n}\psi_{\theta_{n}}(Y_{i})\xrightarrow[P_{n,\theta_{n}}]{d}\mathcal{N}(0,I).
\]
Furthermore, the information matrix $I_{\theta}:=\mathbb{E}_{\theta}[\psi_{\theta}\psi_{\theta}^{\intercal}]$
is invertible, continuous in $\theta$ and $0<\inf_{\theta\in\Theta}\lambda_{\textrm{min}}(I_{\theta}^{-1})<\sup_{\theta\in\Theta}\lambda_{\textrm{max}}(I_{\theta}^{-1})<\infty$.\footnote{Here, $\lambda_{\textrm{min}}(A)$ and $\lambda_{\textrm{max}}(A)$
represent the minimum and maximum eigenvalues of a matrix $A$.} \end{asmA1}

\begin{asmA2}The function $\mu(\cdot)$ is Lipschitz continuous over
$\breve{\Theta}$. Specifically, there exists $\dot{\mu}_{\theta}\in\mathbb{R}^{d}$
and $\epsilon_{n}\to0$ such that $\sqrt{n}\left(\mu(\theta+h/\sqrt{n})-\mu(\theta)\right)=\dot{\mu}_{\theta}^{\intercal}h+\epsilon_{n}\vert h\vert^{2}$
uniformly over all $\theta\in\breve{\Theta}$ and bounded $h$. Furthermore,
for the treatment-assignment problem, $\dot{\mu}_{\theta}$ is continuous
in $\theta$ and $\inf_{\theta\in\tilde{\Theta}}\left\Vert \dot{\mu}_{\theta}\right\Vert >0$,
where $\tilde{\Theta}:=\{\theta\in\Theta:\mu(\theta)=0\}\subset\textrm{interior}(\Theta)$.\end{asmA2}

Assumption A1 follows \citet[Definition 2.2]{ibragimov1981problem}.
Primitive conditions for this assumption can be found in \citet[Sections 2.6 \& 2.7]{ibragimov1981problem};
essentially, what is required is a uniform version of quadratic mean
differentiability. Many commonly used parametric models satisfy this
assumption. For instance, the Bernoulli distribution satisfies it
provided $\Theta\equiv[\epsilon,1-\epsilon]$ for some $\epsilon>0$.
Assumption A2 is a uniform version of Assumption 2. The main additional
requirement is that we need the derivative $\dot{\mu}_{\theta}$ to
be non-zero on the zero-set of $\mu(\cdot)$, i.e., whenever $\mu(\theta)=0$. 

We now define a reference parameter dependent limit experiment. Suppose
that the decision-maker observes a $d$-dimensional signal $x$, posited
to be drawn from a reference Gaussian likelihood, $P_{\theta,h}(\cdot)\equiv\mathcal{N}(I_{\theta}^{1/2}h,I)$.
Let $V_{\theta}^{*}$ represent the parameter dependent minimal decision-risk
in this experiment:
\begin{align}
V_{\theta}^{*} & :=\min_{\tilde{\delta}}\,\max_{\pi(h)\in\Delta(\mathbb{R})}\int\mathbb{E}_{\theta,h}\left[e^{l_{\theta}(h,\tilde{\delta})/\lambda}\right]d\pi(h),\ \textrm{with}\label{eq:minimax_value_limit_experiment_nonparametrics-1}\\
l_{\theta}(h,\tilde{\delta}) & =\begin{cases}
\ell(\dot{\mu}_{\theta}^{\intercal}h-\tilde{\delta}) & \textrm{for estimation loss},\\
\dot{\mu}_{\theta}^{\intercal}h\,\left\{ \mathbf{1}\{\dot{\mu}_{\theta}^{\intercal}h\geq0\}-\tilde{\delta}\right\}  & \textrm{for treatment-assignment loss}.
\end{cases}\nonumber 
\end{align}
Observe that the corresponding optimal decisions are $\tilde{\delta}_{\theta}^{*}=\dot{\mu}_{\theta}^{\intercal}I_{\theta}^{-1/2}x$
for estimation and $\tilde{\delta}_{\theta}^{*}=\mathbf{1}\left\{ \dot{\mu}_{\theta}^{\intercal}I_{\theta}^{-1/2}x\geq0\right\} $
for treatment assignment. 

We then obtain the following lower bound under global priors:

\begin{thm}\label{thm:lower_bound-non-local} Suppose Assumptions
A1, A2 hold, and there exists $\bar{\theta}\in\textrm{interior}(\Theta)$
such that $V_{\bar{\theta}}^{*}=\sup_{\theta\in\Theta}V_{\theta}^{*}$
for the estimation loss. Then, under both the estimation and treatment-assignment
loss functions, 
\begin{align*}
 & \liminf_{n\to\infty}\,\min_{\delta}\,\max_{\pi(\theta)\in\Delta(\Theta)}\int\mathbb{E}_{n,\theta}\left[e^{l_{n}(\theta,\delta)/\lambda}\right]d\pi(\theta)\\
 & \ge\begin{cases}
\sup_{\theta\in\Theta}V_{\theta}^{*} & \textrm{for estimation loss},\\
\sup_{\{\theta\in\Theta:\mu(\theta)=0\}}V_{\theta}^{*} & \textrm{for treatment-assignment loss}.
\end{cases}
\end{align*}
\end{thm}

To show that the MLE based decisions in (\ref{eq:asymptotically_optimal_decisions})
also remain asymptotically optimal under global priors, we impose
a stronger assumption on the properties of MLE:

\begin{asmA3} The maximum-likelihood estimator $\hat{\theta}_{\textrm{mle}}$
admits a uniform locally linear score-function approximation, i.e.,
for any $\theta_{n}\to\theta\in\Theta$, 
\[
\sqrt{n}I_{\theta_{n}}^{1/2}\left(\hat{\theta}_{\textrm{mle}}-\theta_{n}\right)=x_{n,\theta_{n}}+o_{P_{n,\theta_{n}}}(1)\xrightarrow[P_{n,\theta_{n}}]{d}\mathcal{N}(0,I).
\]
Furthermore, for any $\epsilon>0$, there exists $M<\infty$ such
that
\[
\sup_{\theta\in\Theta}P_{n,\theta}\left(\left|\sqrt{n}\left(\mu(\hat{\theta}_{\textrm{mle}})-\mu(\theta)\right)\right|>M\right)\le\epsilon.
\]
\end{asmA3}

Sufficient conditions for Assumption A3 can be found in \citet[Theorem 3.1]{ibragimov1981problem}.

As in Section \ref{subsec:Formal-results-parametric}, we require
that $\ell(\cdot)$ be bounded. Additionally, we also truncate the
loss for the treatment-assignment problem to avoid issues relating
to the non-existence of moments. We state these requirements as an
additional assumption below. 

\begin{asmA4} The function $\ell(\cdot)$ is bowl-shaped and bounded.
Additionally, for the treatment assignment problem, we replace $l_{n}(\theta,\delta)$
with the truncated loss $l_{n,K}(\theta,\delta)=K\wedge l_{n}(\theta,\delta)$.\end{asmA4}

\begin{thm}\label{thm:upper_bound-non-local} Suppose that Assumptions
A1-A4 hold. Then, under both the estimation and treatment-assignment
loss functions,
\begin{align*}
 & \lim_{K\to\infty}\,\limsup_{n\to\infty}\,\max_{\pi(\theta)\in\Delta(\Theta)}\int\mathbb{E}_{n,\theta}\left[e^{l_{n,K}(\theta,\hat{\delta}_{n}^{*})/\lambda}\right]d\pi(\theta)\\
 & \le\begin{cases}
\sup_{\theta\in\Theta}V_{\theta}^{*} & \textrm{for estimation loss},\\
\sup_{\{\theta\in\Theta:\mu(\theta)=0\}}V_{\theta}^{*} & \textrm{for treatment-assignment loss}.
\end{cases}
\end{align*}
\end{thm}

\subsection{Proof of Theorem \ref{thm:lower_bound-non-local}}

Observe that for any reference parameter $\theta\in\textrm{interior}(\Theta)$,
\begin{align*}
 & \liminf_{n\to\infty}\,\min_{\delta}\,\max_{\pi(\theta)\in\Delta(\Theta)}\int\mathbb{E}_{n,\theta}\left[e^{l_{n}(\theta,\delta)/\lambda}\right]d\pi(\theta)\\
 & \ge\lim_{M\to\infty}\liminf_{n\to\infty}\,\min_{\delta}\,\max_{\pi(h)\in\Delta_{M}(\mathcal{H})}\int\mathbb{E}_{n,\theta+h/\sqrt{n}}\left[e^{l_{n}(\theta+h/\sqrt{n},\delta)/\lambda}\right]d\pi(h).
\end{align*}
For the case of treatment assignment, we choose $\theta\in\{\theta\in\Theta:\mu(\theta)=0\}\subset\textrm{interior}(\Theta)$.
Then, making use of Assumptions A1 and A2, we can employ the same
arguments as in the proof of Theorem \ref{thm:lower_bound} to show
that 
\[
\lim_{M\to\infty}\liminf_{n\to\infty}\min_{\delta}\,\max_{\pi(h)\in\Delta_{M}(\mathcal{H})}\int\mathbb{E}_{n,\theta+h/\sqrt{n}}\left[e^{l_{n}(\theta+h/\sqrt{n},\delta)/\lambda}\right]d\pi(h)\ge V_{\theta}^{*}.
\]
The claim thus follows since the above holds for any $\theta\in\textrm{interior}(\Theta)$
under estimation loss, and any $\theta\in\{\theta\in\Theta:\mu(\theta)=0\}$
under treatment-assignment loss.

\subsection{Proof of Theorem \ref{thm:upper_bound-non-local}}

\subsubsection*{Estimation}

We start with the case of estimation. Consider any sequence $\theta_{n}\to\theta\in\Theta$
and $h_{n}\to h$. By Assumptions A1 and A3, 
\begin{align}
\left(\begin{array}{c}
\sqrt{n}(\hat{\theta}_{\textrm{mle}}-\theta_{n})\\
\ln\frac{dP_{n,\theta_{n}+h_{n}/\sqrt{n}}}{dP_{n,\theta_{n}}}
\end{array}\right) & \xrightarrow[P_{n,\theta_{n}}]{d}\left(\begin{array}{c}
I_{\theta}^{-1/2}x\\
h^{\intercal}I_{\theta}^{1/2}x-\frac{1}{2}h^{\intercal}I_{\theta}h
\end{array}\right),\ \textrm{where }x\sim\mathcal{N}(0,I).\label{eq:pf:Thm-A2:weak convergence}
\end{align}
Le Cam's third lemma then yields 
\begin{equation}
\sqrt{n}(\hat{\theta}_{\textrm{mle}}-\theta_{n})\xrightarrow[P_{n,\theta_{n}+h_{n}/\sqrt{n}}]{d}\mathcal{N}(h,I_{\theta}^{-1}).\label{eq:pf:Thm2:0-1}
\end{equation}
Therefore, in view of Assumption A2, it follows that for each $\theta_{n}\to\theta$
and $h_{n}\to h$, 
\begin{align*}
l_{n}\left(\theta_{n}+h_{n}/\sqrt{n},\hat{\delta}_{n}^{*}\right) & =\ell\left(\sqrt{n}\left(\mu(\theta_{n}+h_{n}/\sqrt{n})-\mu(\hat{\theta}_{\textrm{mle}})\right)\right)\\
 & =\left(\sqrt{n}\left(\mu(\theta_{n}+h_{n}/\sqrt{n})-\mu(\theta_{n})\right)+\sqrt{n}\left(\mu(\theta_{n})-\mu(\hat{\theta}_{\textrm{mle}})\right)\right)\\
 & \xrightarrow[P_{n,\theta_{n}+h_{n}/\sqrt{n}}]{d}\ell\left(\dot{\mu}_{\theta}^{\intercal}I_{\theta}^{-1/2}x\right),\ \textrm{where }x\sim\mathcal{N}(0,I).
\end{align*}
Since $l_{n}(\cdot)$ is uniformly bounded by Assumption A4, standard
properties of weak convergence and the fact $\tilde{\delta}_{\theta}^{*}(\cdot):=\dot{\mu}_{\theta}^{\intercal}I_{\theta}^{-1/2}\cdot$
imply 
\begin{equation}
\mathbb{E}_{n,\theta_{n}+h_{n}/\sqrt{n}}\left[e^{l_{n}(\theta_{n}+h_{n}/\sqrt{n},\hat{\delta}_{n}^{*})/\lambda}\right]\to\mathbb{E}_{\theta,0}\left[e^{\ell\left(\dot{\mu}_{\theta}^{\intercal}I_{\theta}^{-1/2}x\right)/\lambda}\right]=\mathbb{E}_{\theta,h}\left[e^{\ell\left(\dot{\mu}_{\theta}^{\intercal}h-\tilde{\delta}_{\theta}^{*}\right)/\lambda}\right]\label{eq:pf:thm2:1-1}
\end{equation}
for every sequence $(\theta_{n},h_{n})\to(\theta,h)$. 

Define
\begin{align*}
f_{n}(\theta,h) & :=\mathbb{E}_{n,\theta+h/\sqrt{n}}\left[e^{l_{n}(\theta+h/\sqrt{n},\hat{\delta}_{n}^{*})/\lambda}\right]\ \textrm{and}\\
f(\theta,h) & :=\mathbb{E}_{\theta,h}\left[e^{\ell\left(\dot{\mu}_{\theta}^{\intercal}h-\tilde{\delta}_{\theta}^{*}\right)/\lambda}\right]=\mathbb{E}_{\theta,h}\left[e^{l_{\theta}\left(h,\tilde{\delta}_{\theta}^{*}\right)/\lambda}\right].
\end{align*}
Equation (\ref{eq:pf:thm2:1-1}) implies continuous convergence of
$f_{n}(\cdot)$ to $f(\cdot)$, i.e., $f_{n}(\theta_{n},h_{n})\to f(\theta,h)$
for every $(\theta_{n},h_{n})\to(\theta,h)$. But continuous convergence
on compact sets implies uniform convergence, so 
\begin{equation}
\sup_{(\theta,h)\in\Theta\times[-M,M]}\vert f_{n}(\theta,h)-f(\theta,h)\vert\to0.\label{eq:pf:thm2:2-1}
\end{equation}

Now, observe that for any $M<\infty$,
\begin{align*}
 & \limsup_{n\to\infty}\,\max_{\pi(\theta)\in\Delta(\Theta)}\int\mathbb{E}_{n,\theta}\left[e^{l_{n}(\theta,\hat{\delta}_{n}^{*})/\lambda}\right]d\pi(\theta)\\
 & \le\limsup_{n\to\infty}\,\max_{\theta\in\Theta}\max_{\pi(h)\in\Delta_{M}(\mathcal{H})}\int\mathbb{E}_{n,\theta+h/\sqrt{n}}\left[e^{l_{n}(\theta+h/\sqrt{n},\hat{\delta}_{n}^{*})/\lambda}\right]d\pi(h).
\end{align*}
Consider a sequence $\left\{ \left(\theta_{n},\pi_{n}(h)\right)\right\} _{n}$
along which the limsup on the right hand side is attained. Since $\Delta_{M}(\mathcal{H})$
represents the space of compactly supported priors, it is compact
under the metric of weak convergence. Hence, there exists a further
sub-sequence $\left\{ \left(\theta_{n_{j}},\pi_{n_{j}}(h)\right)\right\} _{j}$
such that $\left(\theta_{n_{j}},\pi_{n_{j}}(h)\right)$ converges
weakly to some $(\tilde{\theta},\tilde{\pi}(h))\in\Theta\times\Delta_{M}(\mathcal{H})$.
Furthermore, as $e^{\ell_{}(\cdot)/\lambda}$ is uniformly bounded,
so is $f(\cdot)$. Combining these observations with (\ref{eq:pf:thm2:2-1}),
standard properties of weak convergence yield
\begin{align*}
 & \limsup_{n\to\infty}\,\max_{\theta\in\Theta}\,\max_{\pi(h)\in\Delta_{M}(\mathcal{H})}\int\mathbb{E}_{n,\theta+h/\sqrt{n}}\left[e^{l_{n}(\theta+h/\sqrt{n},\hat{\delta}_{n}^{*})/\lambda}\right]d\pi(h)\\
 & =\lim_{j\to\infty}\int f_{n_{j}}(\theta_{n_{j}},h)d\pi_{n_{j}}(h)\\
 & \le\lim_{j\to\infty}\,\sup_{\theta\in\Theta}\,\sup_{\vert h\vert\le M}\vert f_{n_{j}}(\theta,h)-f(\theta,h)\vert\\
 & \quad+\lim_{j\to\infty}\left|\int f(\theta_{n_{j}},h)d\pi_{n_{j}}(h)-\int f(\theta_{n_{j}},h)d\tilde{\pi}(h)\right|+\lim_{j\to\infty}\int f(\theta_{n_{j}},h)d\tilde{\pi}(h)\\
 & =\int\lim_{j\to\infty}f(\theta_{n_{j}},h)d\tilde{\pi}(h).
\end{align*}
Now, observe that by Assumption A1 (continuity of $I_{\theta}$) and
standard properties of the Gaussian distribution, $\lim_{j\to\infty}f(\theta_{n_{j}},h)=f(\tilde{\theta},h)$
for each $h$. Consequently, 
\begin{align*}
 & \limsup_{n\to\infty}\,\max_{\theta\in\Theta}\,\max_{\pi(h)\in\Delta_{M}(\mathcal{H})}\int\mathbb{E}_{n,\theta+h/\sqrt{n}}\left[e^{l_{n}(\theta+h/\sqrt{n},\hat{\delta}_{n}^{*})/\lambda}\right]d\pi(h)\\
 & =\int f(\tilde{\theta},h)d\tilde{\pi}(h)\le\sup_{\theta\in\Theta}\max_{\pi(h)\in\Delta(\mathcal{H})}\int\mathbb{E}_{\theta,h}\left[e^{l_{\theta}\left(h,\tilde{\delta}_{\theta}^{*}\right)/\lambda}\right]d\pi(h):=\sup_{\theta\in\Theta}V_{\theta}^{*}.
\end{align*}
This proves Theorem \ref{thm:upper_bound-non-local} for the estimation
problem. 

\subsubsection*{Treatment assignment}

Recall the definition $\tilde{\Theta}:=\{\theta\in\Theta:\mu(\theta)=0\}$
from Assumption A2 and note that $\tilde{\Theta}$ a compact set due
to the Lipschitz continuity of $\mu(\theta)$, as imposed in Assumption
A2. In addition, denote 
\begin{align*}
\tilde{\Theta}_{n,M} & :=\{\theta\in\Theta:\vert\mu(\theta)\vert\le M/\sqrt{n}\},\\
\tilde{\Theta}_{n,M}^{+} & :=\{\theta\in\Theta:\mu(\theta)>M/\sqrt{n}\},\textrm{ and}\\
\tilde{\Theta}_{n,M}^{-} & :=\{\theta\in\Theta:\mu(\theta)<-M/\sqrt{n}\}.
\end{align*}
We can decompose
\begin{align}
 & \limsup_{n\to\infty}\,\max_{\pi(\theta)\in\Delta(\Theta)}\int\mathbb{E}_{n,\theta}\left[e^{l_{n,K}(\theta,\hat{\delta}_{n}^{*})/\lambda}\right]d\pi(\theta)\nonumber \\
 & \le\limsup_{n\to\infty}\,\max_{\theta\in\tilde{\Theta}_{n,M}}\mathbb{E}_{n,\theta}\left[e^{l_{n,K}(\theta,\hat{\delta}_{n}^{*})/\lambda}\right]\nonumber \\
 & \;\quad\vee\limsup_{n\to\infty}\,\max_{\theta\in\tilde{\Theta}_{n,M}^{+}}\mathbb{E}_{n,\theta}\left[e^{l_{n,K}(\theta,\hat{\delta}_{n}^{*})/\lambda}\right]\vee\limsup_{n\to\infty}\,\max_{\theta\in\tilde{\Theta}_{n,M}^{-}}\mathbb{E}_{n,\theta}\left[e^{l_{n,K}(\theta,\hat{\delta}_{n}^{*})/\lambda}\right].\label{eq:pf:thmA2:1}
\end{align}

Note that $l_{n,K}(\theta,\hat{\delta}_{n}^{*})\le K$. Then, by the
second part of Assumption A4, we may choose $M<\infty$ large enough
such that
\begin{align*}
 & \limsup_{n\to\infty}\,\max_{\theta\in\tilde{\Theta}_{n,M}^{+}}\mathbb{E}_{n,\theta}\left[e^{l_{n,K}(\theta,\hat{\delta}_{n}^{*})/\lambda}\right]\\
 & \le\limsup_{n\to\infty}\left\{ e^{K/\lambda}\sup_{\theta\in\Theta}P_{n,\theta}\left(\sqrt{n}\left(\mu(\hat{\theta}_{\textrm{mle}})-\mu(\theta)\right)<-M\right)+1\right\} \le1+\eta,
\end{align*}
for any $\eta$ arbitrarily small. In a similar vein,
\[
\limsup_{n\to\infty}\max_{\theta\in\tilde{\Theta}_{n,M}^{-}}\mathbb{E}_{n,\theta}\left[e^{l_{n,K}(\theta,\hat{\delta}_{n}^{*})/\lambda}\right]\le1+\eta.
\]

It remains to analyze the term 
\[
\limsup_{n\to\infty}\,\max_{\theta\in\tilde{\Theta}_{n,M}}\mathbb{E}_{n,\theta}\left[e^{l_{n,K}(\theta,\hat{\delta}_{n}^{*})/\lambda}\right]=\limsup_{n\to\infty}\,\max_{\pi\in\Delta(\tilde{\Theta}_{n,M})}\int\mathbb{E}_{n,\theta}\left[e^{l_{n,K}(\theta,\hat{\delta}_{n}^{*})/\lambda}\right]d\pi(\theta).
\]

We now claim that there exist $c>0,\delta>0$ such that 
\begin{equation}
\vert\mu(\theta)\vert>c\cdot d(\theta,\tilde{\Theta})\ \forall\ \theta\ \textrm{with }d(\theta,\tilde{\Theta})\le\delta.\label{eq:pf:thmA2:2}
\end{equation}
To see this, note first that the uniform expansion in Assumption A2
implies $\mu(\cdot)$ is strictly differentiable at every $\theta\in\tilde{\Theta}$
with derivative $\dot{\mu}_{\theta}$; indeed, writing $u=h/\sqrt{n}$
and choosing $n\to\infty$ as $u\to0$, the expansion yields $|\mu(\theta+u)-\mu(\theta)-\dot{\mu}_{\theta}^{\intercal}u|\le\rho(|u|)|u|$
with $\rho(t)\to0$ as $t\to0$, uniformly over $\theta$, which is
a uniform version of strict differentiability. Since $\mu(\cdot)$
is scalar-valued and $\inf_{\theta\in\tilde{\Theta}}\left\Vert \dot{\mu}_{\theta}\right\Vert >0$,
the derivative is surjective at each $\theta\in\tilde{\Theta}$. Consequently,
by the Lyusternik--Graves theorem (\citealp[Theorem 5D.5]{dontchev2009implicit};
see also Section 1F for the equation case), $\mu(\cdot)$ is metrically
regular at each $\theta\in\tilde{\Theta}$, i.e., there exists a $\kappa_{\theta}<\infty$
and an open neighborhood $\mathcal{N}_{\theta}$ of $\theta$ such
that $d(\theta,\tilde{\Theta})=d(\theta,\mu^{-1}(0))\le\kappa_{\theta}\vert\mu(\theta)|\ \forall\ \theta\in\mathcal{N}_{\theta}$.
The regularity modulus admits the bound $\kappa_{\theta}\le(1+o(1))/\left\Vert \dot{\mu}_{\theta}\right\Vert \le2/\inf_{\theta\in\tilde{\Theta}}\left\Vert \dot{\mu}_{\theta}\right\Vert $.
Since $\tilde{\Theta}\subseteq\cup_{\theta\in\tilde{\Theta}}\mathcal{N}_{\theta}$
and $\tilde{\Theta}$ is a compact set, extracting a finite sub-cover
yields (\ref{eq:pf:thmA2:2}) with $c:=\inf_{\theta\in\tilde{\Theta}}\left\Vert \dot{\mu}_{\theta}\right\Vert /2$
and some $\delta>0$. 

Now, since $m:=\inf\{|\mu(\theta)|:\theta\in\Theta,d(\theta,\tilde{\Theta})>\delta\}>0$
by continuity of $\mu(\cdot)$ and compactness of $\Theta$, any $\theta\in\tilde{\Theta}_{n,M}$
satisfies $d(\theta,\tilde{\Theta})\le\delta$ once $M/\sqrt{n}\le m$,
so (\ref{eq:pf:thmA2:2}) applies and gives, for $n$ large enough
and $L:=M/c<\infty$,
\[
\tilde{\Theta}_{n,M}\subseteq\left\{ \theta+h/\sqrt{n}:\theta\in\tilde{\Theta},\vert h\vert\le L\right\} .
\]
Hence, 
\begin{align*}
 & \limsup_{n\to\infty}\,\max_{\pi\in\Delta(\tilde{\Theta}_{n,M})}\int\mathbb{E}_{n,\theta}\left[e^{l_{n,K}(\theta,\hat{\delta}_{n}^{*})/\lambda}\right]d\pi(\theta)\\
 & \le\limsup_{n\to\infty}\,\max_{\theta\in\tilde{\Theta}}\,\max_{\pi\in\Delta([-L,L])}\int\mathbb{E}_{n,\theta+h/\sqrt{n}}\left[e^{l_{n,K}(\theta+h/\sqrt{n},\hat{\delta}_{n}^{*})/\lambda}\right]d\pi(h).
\end{align*}

Consider any sequence $(\theta_{n},h_{n})\in\tilde{\Theta}\times[-L,L]$
such that $(\theta_{n},h_{n})\to(\theta,h)$. Since $\tilde{\Theta}$
is a compact set, $\theta\in\tilde{\Theta}$, i.e., $\mu(\theta)=0$.
Then, (\ref{eq:pf:Thm2:0-1}) and Assumption A2 imply 
\begin{align*}
\mathbf{1}\left\{ \mu(\hat{\theta}_{\textrm{mle}})\ge0\right\}  & =\mathbf{1}\left\{ \sqrt{n}\left(\mu(\hat{\theta}_{\textrm{mle}})-\mu(\theta_{n}+h_{n}/\sqrt{n})\right)+\sqrt{n}\left(\mu(\theta_{n}+h_{n}/\sqrt{n})-\mu(\theta)\right)\ge0\right\} \\
 & \xrightarrow[P_{n,\theta_{n}+h_{n}/\sqrt{n}}]{d}\mathbf{1}\left\{ \dot{\mu}_{\theta}^{\intercal}I_{\theta}^{-1/2}x\ge0\right\} ,\ \textrm{where }x\sim\mathcal{N}(I_{\theta}^{-1/2}h,I).
\end{align*}
Hence, by standard properties of weak convergence, for each $(\theta_{n},h_{n})\in\tilde{\Theta}\times[-L,L]\to(\theta,h)$,
\begin{equation}
\mathbb{E}_{n,\theta_{n}+h_{n}/\sqrt{n}}\left[\hat{\delta}_{n}^{*}\right]=P_{n,\theta_{n}+h_{n}/\sqrt{n}}\left(\mu(\hat{\theta}_{\textrm{mle}})\ge0\right)\to P_{\theta,h}\left(\dot{\mu}_{\theta}^{\intercal}I_{\theta}^{-1/2}x\ge0\right)=\mathbb{E}_{\theta,h}\left[\tilde{\delta}_{\theta}^{*}\right].\label{eq:pf:ThmA2:3}
\end{equation}

Note that for the treatment assignment loss, 
\begin{align*}
 & \mathbb{E}_{n,\theta_{n}+h_{n}/\sqrt{n}}\left[e^{l_{n,K}(\theta_{0}+h_{n}/\sqrt{n},\hat{\delta}_{n}^{*})/\lambda}\right]\\
 & =e^{l_{n,K}(\theta_{n}+h_{n}/\sqrt{n},1)/\lambda}\mathbb{E}_{n,\theta_{n}+h_{n}/\sqrt{n}}\left[\hat{\delta}_{n}^{*}\right]+e^{l_{n,K}(\theta_{n}+h_{n}/\sqrt{n},0)/\lambda}\mathbb{E}_{n,\theta_{n}+h_{n}/\sqrt{n}}\left[1-\hat{\delta}_{n}^{*}\right].
\end{align*}
Now, Assumption A2 yields
\[
l_{n,K}(\theta_{n}+h_{n}/\sqrt{n},a)\to l_{\theta,K}(h,a),\ \textrm{for each }a\in\{0,1\}\ \textrm{and }(\theta_{n},h_{n})\to(\theta,h),
\]
where 
\[
l_{\theta,K}(h,a):=K\wedge\dot{\mu}_{\theta}^{\intercal}h\,\left\{ \mathbf{1}\{\dot{\mu}_{\theta}^{\intercal}h\geq0\}-a\right\} .
\]
Combined with (\ref{eq:pf:ThmA2:3}), this proves 
\begin{align*}
 & \mathbb{E}_{n,\theta_{n}+h_{n}/\sqrt{n}}\left[e^{l_{n,K}(\theta_{n}+h_{n}/\sqrt{n},\hat{\delta}_{n}^{*})/\lambda}\right]\\
 & \to e^{l_{\theta,K}(h,1)/\lambda}\mathbb{E}_{\theta,h}\left[\tilde{\delta}_{\theta}^{*}\right]+e^{l_{\theta,K}(h,0)/\lambda}\mathbb{E}_{\theta,h}\left[1-\tilde{\delta}_{\theta}^{*}\right]\\
 & =\mathbb{E}_{\theta,h}\left[e^{l_{\theta,K}(h,\tilde{\delta}_{\theta}^{*})/\lambda}\right],\ \textrm{for each }(\theta_{n},h_{n})\in\tilde{\Theta}\times[-L,L]\to(\theta,h).
\end{align*}

As before, define $f_{n},f:\tilde{\Theta}\times[-L,L]\to\mathbb{R}$
such that
\begin{align*}
f_{n}(\theta,h) & :=\mathbb{E}_{n,\theta+h/\sqrt{n}}\left[e^{l_{n,K}(\theta+h_{}/\sqrt{n},\hat{\delta}_{n}^{*})/\lambda}\right]\ \textrm{and}\\
f(\theta,h) & :=\mathbb{E}_{h}\left[e^{l_{\theta,K}\left(h,\tilde{\delta}_{\theta}^{*}\right)/\lambda}\right].
\end{align*}
Observe that $f(\theta,h)$ is bounded under treatment-assignment
loss by construction. Consequently, by similar arguments as in the
case of estimation, it follows 
\begin{align*}
 & \limsup_{n\to\infty}\,\sup_{\theta\in\tilde{\Theta}}\,\max_{\pi\in\Delta([-L,L])}\int\mathbb{E}_{n,\theta}\left[e^{l_{n,K}(\theta,\hat{\delta}_{n}^{*})/\lambda}\right]d\pi(h)\\
 & \le\sup_{\theta\in\tilde{\Theta}}\,\max_{\pi(h)\in\Delta(\mathcal{H})}\int\mathbb{E}_{\theta,h}\left[e^{l_{\theta,K}\left(h,\tilde{\delta}_{\theta}^{*}\right)/\lambda}\right]d\pi(h)\\
 & \le\sup_{\theta\in\tilde{\Theta}}\,\max_{\pi(h)\in\Delta(\mathcal{H})}\int\mathbb{E}_{\theta,h}\left[e^{l_{\theta}\left(h,\tilde{\delta}_{\theta}^{*}\right)/\lambda}\right]d\pi(h):=\sup_{\theta\in\tilde{\Theta}}V_{\theta}^{*}.
\end{align*}

Observe that, by definition, $V_{\theta}^{*}>1$ for any $\theta\in\tilde{\Theta}$.
Combined with (\ref{eq:pf:thmA2:1}) and the fact $\eta>0$ can be
set arbitrarily small, it follows that 
\[
\limsup_{n\to\infty}\max_{\pi(\theta)\in\Delta(\Theta)}\int\mathbb{E}_{n,\theta}\left[e^{l_{n,K}(\theta,\hat{\delta}_{n}^{*})/\lambda}\right]d\pi(\theta)\le\sup_{\theta\in\tilde{\Theta}}V_{\theta}^{*}\ \forall\ K,
\]
as stated by the theorem. 

\section{Asymmetric Loss Functions: The Case of Linex\label{sec:Asymmetric-Loss-Functions}}

In the limit experiment, the linex loss takes the form 
\[
l(h,\tilde{\delta})=e^{(\dot{\mu}_{0}^{\intercal}h-\tilde{\delta})}-(\dot{\mu}_{0}^{\intercal}h-\tilde{\delta})-1.
\]
We start by analyzing the minimax optimal decision under no misspecification.
Observe that $\dot{\mu}_{0}^{\intercal}I_{0}^{-1/2}x$ is a sufficient
statistic for $\dot{\mu}_{0}^{\intercal}h$. Also, the loss function
is convex and location-invariant, in the sense that $l(h+z,\tilde{\delta}+\dot{\mu}_{0}^{\intercal}z)$
for all $z\in\mathbb{R}^{d}$. Consequently, the Hunt-Stein theorem
implies that the minimax optimal estimator should be the minimum risk
equivariant estimator. Any equivariant estimator in this setting must
be of the form $\tilde{\delta}_{z}(x)=\dot{\mu}_{0}^{\intercal}\left(I_{0}^{-1/2}x+z\right)$
for some constant $z$. The frequentist risk of $\tilde{\delta}_{z}$
is
\begin{align*}
R(h,\tilde{\delta}_{z}) & =e^{-\dot{\mu}_{0}^{\intercal}z}\mathbb{E}_{h}\left[e^{(\dot{\mu}_{0}^{\intercal}h-\dot{\mu}_{0}^{\intercal}I_{0}^{-1/2}x)}\right]+\dot{\mu}_{0}^{\intercal}z-1\\
 & =e^{-\dot{\mu}_{0}^{\intercal}z}e^{\dot{\mu}_{0}^{\intercal}I_{0}^{-1}\dot{\mu}_{0}}+\dot{\mu}_{0}^{\intercal}z-1.
\end{align*}
Optimizing over the value of $z$ --- or equivalently, that of $\dot{\mu}_{0}^{\intercal}z$
--- we find that the frequentist risk is minimized at $\dot{\mu}_{0}^{\intercal}z^{*}=\frac{1}{2}\dot{\mu}_{0}^{\intercal}I_{0}^{-1}\dot{\mu}_{0}$.
Consequently, in the absence of misspecification, the minimax optimal
estimator takes the form 
\[
\tilde{\delta}^{*}=\dot{\mu}_{0}^{\intercal}I_{0}^{-1/2}x+\frac{1}{2}\dot{\mu}_{0}^{\intercal}I_{0}^{-1}\dot{\mu}_{0}.
\]

Recall that optimal decisions under misspecification are equivalent
to minimax optimal decisions with an exponential tilt of the loss
function. Since the resulting risk would be infinite under the Gaussian
experiment, we truncate the linex loss at some large value $M$ to
obtain $l_{M}(h,\tilde{\delta})=\min\left\{ l(h,\tilde{\delta}),M\right\} $.
The decision-risk under misspecification is then given by 
\[
V^{*}=\min_{\tilde{\delta}}\,\max_{\pi(h)}\int\mathbb{E}_{h}\left[e^{l_{M}(h,\tilde{\delta})/\lambda}\right]d\pi(h).
\]
Observe that $\exp\{l_{M}(h,\tilde{\delta})/\lambda\}$ is still convex
and location-invariant. Consequently, we can apply the Hunt-Stein
theorem again to conclude that the minimax optimal estimator must
be of the form $\tilde{\delta}_{z}(x)=\dot{\mu}_{0}^{\intercal}I_{0}^{-1/2}x+\dot{\mu}_{0}^{\intercal}z$.
Since the frequentist risks, $R_{\lambda}(h,\tilde{\delta}_{z}):=\mathbb{E}_{h}\left[e^{l_{M}(h,\tilde{\delta}_{z})/\lambda}\right],$
of equivariant estimators are independent of $h$, we have 
\[
R_{\lambda}(h,\tilde{\delta}_{z})=R_{\lambda}(0,\tilde{\delta}_{z})=\mathbb{E}_{0}\left[e^{l_{M}(0,\tilde{\delta}_{z})/\lambda}\right].
\]
The optimal value of $\dot{\mu}_{0}^{\intercal}z^{*}$ therefore solves
\begin{align*}
\dot{\mu}_{0}^{\intercal}z^{*}=\Delta_{M}^{*}(\lambda) & :=\arg\min_{\Delta}\mathbb{E}_{0}\left[e^{l_{M}(0,\dot{\mu}_{0}^{\intercal}I_{0}^{-1/2}x+\Delta)/\lambda}\right]\\
 & =\arg\min_{\Delta}\mathbb{E}_{Y}\left[e^{l_{M}(0,Y+\Delta)/\lambda}\right],\ \textrm{where }Y\sim\mathcal{N}(0,\dot{\mu}_{0}^{\intercal}I_{0}^{-}\dot{\mu}_{0}).
\end{align*}
Applying the implicit function theorem, some tedious but straightforward
algebra shows that $\partial_{\lambda}\Delta_{M}^{*}(\lambda)<0$,
i.e., the minimax optimal shift decreases in $\lambda$. Indeed, as
$M\to0$ and $\lambda\to\infty$ --- i.e., misspecification risk
decreases --- the minimax optimal shift converges to that under no
misspecification: 
\[
\lim_{M\to\infty}\lim_{\lambda\to\infty}\Delta_{M}^{*}(\lambda)=\frac{1}{2}\dot{\mu}_{0}^{\intercal}I_{0}^{-1}\dot{\mu}_{0}.
\]
On the other hand, when $\lambda\to0$ --- i.e., misspecification
risk explodes --- we find that $\Delta_{M}^{*}(\lambda)\to\infty$
and the minimax optimal estimator also becomes $\infty$. Hence, in
the case of linex loss, the optimal estimator depends on the degree
of misspecification. 
\end{document}